\documentclass[twoside,11pt]{article}

\usepackage{graphicx, psfrag,epsf}
\usepackage{setspace}
\usepackage{amsthm,amsmath,bbm,amsfonts,amssymb, mathrsfs}
\usepackage{enumerate}
\usepackage{cancel}
\usepackage{natbib}
\usepackage{soul}
\usepackage{subcaption} % multiple panels per figure
\usepackage{booktabs}
\usepackage{multirow}
\usepackage{adjustbox} % for adding tables
\usepackage{etoolbox}  % for if-else in fig and tab commands
\usepackage{accents}
\usepackage{apptools}
\usepackage{float}
\usepackage{xcolor}
\usepackage{scalerel,stackengine}
\usepackage{mathtools}
\usepackage{caption}
\usepackage{subcaption}
\usepackage{tikz}
\usepackage{pgf}
\usepackage{algorithm}
\usepackage{algpseudocode}
\usepackage{graphicx}
\usepackage{nicefrac}  
\usepackage{titletoc}
\usepackage{tocloft}

\usepackage{JMLR/jmlr2e}

\usetikzlibrary{positioning, calc, shapes.geometric, shapes.multipart, shapes, arrows.meta, arrows, 	decorations.markings, external, trees}

\tikzstyle{Arrow} = [
thick,
decoration={
	markings,
	mark=at position 0.99 with {  % fix bug with tikz arrow head interference with decoration
		\arrow[thick]{latex}
	}
},
shorten >= 3pt, preaction = {decorate}
]
\tikzstyle{BiArrow} = [
thick,
decoration={
	markings,
	mark=at position 0.99 with {  % fix bug with tikz arrow head interference with decoration
		\arrow[thick]{latex}
	}
},
shorten >= 3pt, preaction = {decorate}
]
\tikzstyle{RedArrow} = [
red, 
line width = 0.5mm,
decoration={
	markings,
	mark=at position 0.99 with {
		\arrow[line width = 0.5mm,red]{latex}
	}
},
shorten >= 3pt, preaction = {decorate}
]

\tikzstyle{BlueArrow} = [
blue, 
line width = 0.5mm,
decoration={
	markings,
	mark=at position 0.99 with {
		\arrow[line width = 0.5mm,blue]{latex}
	}
},
shorten >= 3pt, preaction = {decorate}
]

\tikzstyle{Dashed} = [
thick,dashed,
decoration={
	markings,
	mark=at position 0.99 with {  % fix bug with tikz arrow head interference with decoration
		\arrow[thick]{latex}
	}
},
shorten >= 3pt, preaction = {decorate}
]

\tikzset{
    BiArrow/.style={
        thick,
        decoration={
            markings,
            mark=at position 0.99 with {
                \arrow[thick]{latex}
            },
            mark=at position 0 with {
                \arrow[thick]{latex}
            }
        },
        shorten >= 3pt,
        preaction = {decorate}
    }
}

\theoremstyle{definition}
\newtheorem{exmp}{Example}
\AfterEndEnvironment{exmp}{\noindent\ignorespaces}
\def \bz {\boldsymbol{z}}
\def \bZ {\boldsymbol{Z}}
\def \bA {\boldsymbol{A}} 

\def \bX {\boldsymbol{X}}
\def \bY {\boldsymbol{Y}}

\def \bS {\boldsymbol{S}}

\def \bU {\boldsymbol{U}}

\def \bD {\boldsymbol{\mathcal{D}}}

\def \cA {\mathcal{A}}
\def \bV {\boldsymbol{V}}

\def \beps {\boldsymbol{\varepsilon}}
\def \balpha {\boldsymbol{\alpha}}
\def \bphi {\boldsymbol{\phi}}
\def \boeta {\boldsymbol{\eta}}
\def \bgamma {\boldsymbol{\gamma}}
\def \btheta {\boldsymbol{\theta}}
\def \bTheta {\boldsymbol{\Theta}}
\def \bbeta {\boldsymbol{\beta}}

\def \bbE {\mathbb{E}}

\newcommand\indep{\protect\mathpalette{\protect\independenT}{\perp}}
\def\independenT#1#2{\mathrel{\rlap{$#1#2$}\mkern2mu{#1#2}}}

\begin{document}

\title{Bayesian Estimation of Causal Effects Using Proxies of a Latent Interference Network}

\author{\name Bar Weinstein \\
        Daniel Nevo
\email  \{barwein, danielnevo\}@tauex.tau.ac.il \\
       \addr
      Department of Statistics and Operations Research \\
      Tel Aviv University
      % \\
      % Tel Aviv, 6997801, Israel
      }
      %  \AND
      %  \name Daniel Nevo \email danielnevo@tauex.tau.ac.il \\
      %  \addr  Department of Statistics and Operations Research \\
      % Tel Aviv University \\
      % Tel Aviv, 6997801, Israel
      % }

% \editor{My editor}

\maketitle

\begin{abstract}% 
Network interference occurs when treatments assigned to some units affect the outcomes of others.
Traditional approaches often assume that the observed network correctly specifies the interference structure. However, in practice, researchers frequently only have access to proxy measurements of the interference network due to limitations in data collection or potential mismatches
between measured networks and actual interference pathways. 
In this paper, we introduce a framework for estimating causal effects when
only proxy networks are available. Our approach leverages a structural causal
model that accommodates diverse proxy types, including noisy measurements, 
multiple data sources, and multilayer networks, and defines causal effects as
interventions on population-level treatments. 
The latent nature of the true interference network poses significant challenges.
To overcome them, we develop a Bayesian inference framework.
We propose a Block Gibbs sampler with Locally Informed Proposals to update the latent network, thereby efficiently exploring the high-dimensional posterior space composed of both discrete and continuous parameters.  The latent network updates are driven by information from the proxy networks, treatments, and outcomes. We illustrate the performance of our method through numerical experiments, demonstrating its accuracy in recovering causal effects even when only proxies of the interference network are available.
\end{abstract}


% \newpage

\section{Introduction}
\label{sec:intro}
In causal inference it is commonly assumed that there is no interference between units,
 that is, the treatment assigned to one unit does not affect 
 the outcomes of other units \citep{Cox1958}.
However, researchers are often interested in settings where 
the treatment of interest can spill over from treated units to untreated ones,
leading to interference between units.
The term interference encompasses various forms of transmission mechanisms between interacting units,
 including social influence, information diffusion, contagion, and protection against infectious diseases.
% Interference arises when treatments of some units affect the outcomes of others. 
% The term interference encompasses various forms of causal transmission between interacting units,
% such as information sharing, protection against infectious diseases, 
% and the spread of social behaviors or beliefs. 
Researchers often seek to estimate the effects of a postulated treatment policy, 
e.g., vaccination allocation in a population. 
% Due to possible interference, the policies are expressed in terms of population-level treatments. 
To estimate the influence of such policies under interference,
 researchers typically rely on detailed information about the network 
 structure that encodes the structure of pairwise unit interactions.   
% These interactions are often represented by a network capturing the interference structure.

However, the network that researchers use to encode the interference structure might not be the actual interference network for at least two reasons. First,  accurately measuring networks of social interactions is a formidable task. For example, 
social networks are often measured through self-reported elicitation of social relations, where respondents are asked to list the names of others with whom they have relevant social interactions. 
 The challenge of correctly measuring social interactions is well known in the network science literature.
% \citep{wang2012measurement}. 
Common remedies for this challenge include improvements in research design \citep{marsden1990network},
 or the use of models that account for measurement error
\citep{Butts2003,Young2021,de2023latent}. 

Second, even when a network is measured without error, the observed social ties may not correspond to
the specific pathways through which interference operates. Researchers often rely on networks that are available or convenient to measure, yet the actual interference might propagate through a different, unobserved structure. For example, researchers may exploit available online social media connections, while in reality, physical contacts, which are more challenging to measure, are more relevant to that specific study.
In such cases, the observed network is a valid measurement of a social structure, 
but it is merely a proxy for the causal interference network. 
For instance, while an observed 'friendship' network is distinct from a latent 'physical contact' network, 
the two are likely correlated due to shared latent social traits.
Treating the observed proxy as the true interference structure can result
in misspecification that can bias causal effect estimates \citep{Weinstein2023}.
% Even when network measurements are accurate, 
% the observed relationships may not necessarily correspond to the 
% pathways through which interference occurs.
% This discrepancy can arise because the relational connections captured in
% the measured network may differ from those relevant to the specific interference mechanism under study.
% Consequently, even a correctly measured network might not accurately represent the interference structure.

Thus, researchers often observe only proxy measurements of the actual interference network, which remains latent. Despite this, it is common in the literature to assume that the measured network accurately represents the interference structure \citep{ugander2013, aronow_estimating_2017,sofrygin2017semi, Forastiere2020, Tchetgen2020, leung2022causal, forastiere2022jmlr, Ogburn2022, hayes2025}.

In this paper, we propose a structural causal models (SCMs) framework for estimating causal effects under network interference, when only proxy measurements of the true interference structure are available. Such scenarios include the case of a single erroneously measured network, repeated measures of noisy networks, multiple networks from varying sources, and multilayer networks (see Section \ref{subsec:sources} for extended examples). 
The SCMs view the true interference network as a realization from a random network model.
% and allows for latent common causes between the network and outcomes.
Crucially, our framework accommodates various structural configurations, including both causal and non-causal proxy mechanisms, and explicitly addresses scenarios where treatment assignment depends on either the latent interference network or the observed proxies.
We focus on the causal effects of interventions on the population treatments, addressing various policy-related estimands.

We discuss the identifiability of the structural parameters in simplified settings, illustrating how the outcomes, treatments, and proxy networks jointly contribute to identification. Identifying these structural parameters allows for the recovery of the latent network's conditional distribution, motivating our proposed Bayesian framework.

We develop a Bayesian inferential framework for estimation. Since the true interference network is latent, the posterior is composed of a mixed space of discrete and continuous components, where the discrete space (latent network) is large. That makes sampling from the posterior challenging and computationally expensive. 
To address this, we propose a Block Gibbs sampler that alternates between updating the continuous parameters and the discrete latent network. For the network updates, we employ Locally Informed Proposals \citep{zanella2019informed} to efficiently explore the high-dimensional discrete space, utilizing gradient-based approximations to render these proposals computationally feasible. 
% These procedures allow researchers to select the appropriate models for their applications flexibly. 
This inference framework allows researchers to flexibly estimate the parameters of the various SCMs described above.
We demonstrate the performance of our proposed methods through numerical illustrations on fully- and semi-synthetic data. Code is available at \url{https://github.com/barwein/BayesProxyNets}.

The remainder of this paper is organized as follows. In Section~\ref{sec:settings}, 
we introduce our formal setup, present the SCMs, and define the causal estimands of interest. 
Section~\ref{sec:proxies} categorizes the different types of proxy measurements and provides concrete examples of true and proxy network models. 
Section~\ref{sec:identifiability} investigates the identification of structural parameters in simplified settings.
Section~\ref{sec:inference} describes the proposed Bayesian algorithm for posterior sampling and the causal effects estimation scheme.
Section~\ref{sec:sim} demonstrates the performance of our method through numerical experiments.
Finally, Section~\ref{sec:discussion} concludes with a discussion of extensions and directions for future work.

\subsection{Related Literature}
\label{subsec:lit}

A significant portion of the interference literature operates under the
assumption that the interference network is fixed, correctly measured,
and exogenous to the structural models \citep{sofrygin2017semi,Tchetgen2020, Ogburn2022}. 
% This includes works investigating asymptotic properties under specific network models \citep{Li2022}
% or joint generation models that account for unit dependence \citep{Clark2024}. 
In contrast, our framework treats the interference network as a realization from 
a random network model, integrating it directly into the SCM.
Other researchers have taken similar perspectives while assuming that the
observed network accurately represents the interference structure \citep{pinkham2013,Li2022, Clark2024}.

When the assumption of correctly specified observed network is relaxed, 
existing solutions often rely on specific structural assumptions. 
A prominent line of work focuses on the linear-in-means model \citep{Manski1993}, 
developing corrections for specific error mechanisms such as random measurement
error \citep{Hardy2019, boucher2021estimating} or edge censoring \citep{Griffith2021}.
These approaches typically leverage the parametric structure of the linear framework for identification.
Other strategies require specific data availability, such as repeated noisy measurements 
to achieve identification via design-based estimators \citep{Li2021}, 
or rely on aggregating relational summaries when the full network is unobserved \citep{reeves2024model}.
Alternatively, \citet{Toulis2013} proposed modeling unobservable interactions through sufficient
statistics of treated neighbors.

Beyond specific error corrections, our work relates to the broad literature on 
latent network inference. While probabilistic models for recovering latent networks from observed proxy networks are well-studied in the network science literature 
\citep{Butts2003,Young2021,de2023latent},
 they are typically applied as distinct tasks, separate from causal effect estimation. 
A natural heuristic involves a two-stage approach: first inferring the network and then treating it 
as fixed for causal analysis. 
We show in the numerical experiments (Section~\ref{sec:sim}) 
that such two-stage approaches can lead to estimation errors and underestimation of uncertainty.
Our Bayesian framework improves upon this by jointly estimating the latent network and causal parameters, thereby properly propagating uncertainty from the network inference stage to the causal estimates.

Beyond causal inference under network interference, 
there is an emerging literature on regression problems with network-linked covariates.
While this literature typically assumes the network is observed without error \citep{zhu2017network,li2019prediction},
exceptions exist that account for sampling uncertainty of the network \citep{le2022}.
% A notable excpetion is \citet{le2022}, who studied 
% network-linked regression settings where the observed network 
% diverge from the true relational structure due to sampling uncertainty. 
% Our framework is sufficiently general to adapt to such network regression problems where 
% only proxies of the true network are observed.
Although our primary focus is on causal estimation, our framework is adaptable to such network regression problems where only proxies are observed.
% We discuss this adaptation and provide guidance for its implementation in Section \ref{sec:discussion}.

In network interference scenarios, treatments of other units are typically assumed 
to affect the potential outcome of one unit only through values of a 
summarizing exposure mapping  \citep{ugander2013, Manski2013, aronow_estimating_2017}. 
However, recent research has highlighted the challenges associated with this approach. 
\citet{Saevje2023} demonstrated that incorrect specification of these mappings 
can introduce bias and obscure the interpretation of results. 
Building on this, \citet{Weinstein2023} showed that such biases may result
 from misspecification of the network structure itself. 
 % They proposed an estimator that integrates multiple networks simultaneously, 
 % ensuring unbiased estimation when at least one network is correctly specified. 
 In this paper, we address these challenges by strictly distinguishing between modeling assumptions and causal targets. While we utilize exposure mappings as summarization tools to render the high-dimensional outcome model tractable, we define causal estimands as interventions on population-level treatments (policies). This separation allows us to define and estimate meaningful static, dynamic, and stochastic policy effects even when precise unit-level exposures are uncertain, thus mitigating the interpretability issues raised by \citet{Saevje2023}.
 % Distinct from our approach of redefining the target estimand, 
 Other methodological advancements have focused on robust estimation using multiple proxy networks \citep{Weinstein2023}, testing the exposure mapping specification through randomization tests \citep{Athey2018, Puelz2022}, or learning the exposure mapping from the data while assuming the observed network is correct \citep{ohnishi2022degree}. 
 % validating the exposure mapping itself
 % In addition, \citet{Athey2018} and \citet{Puelz2022} developed randomization 
 % tests that can assess the validity of different exposure mapping specifications. 
 % In a different direction, \cite{ohnishi2022degree} proposed a Bayesian nonparametric 
 % approach to learn the exposure mapping directly from the data, although still assuming
 %  the observed network is correctly specified.

Finally, we note that certain causal effects can be estimated without network data. 
\citet{Saevje2021} showed that the Expected Average Treatment Effect (EATE), 
which is an effect marginalized over other units' treatments, 
can be consistently estimated with the common design-based estimators
under limiting interference dependence between units. \citet{Yu2022} showed that the 
Total Treatment Effect (TTE) -- treating all units versus none -- can be unbiasedly estimated,
under restrictions on the potential outcomes and the experimental design. 
In addition, \citet{shirani2024causal} developed methods for estimating TTE 
using only an assumed distribution of the network.
However, analyzing more granular policy interventions, such as the dynamic or stochastic policies considered in this paper, requires correct network measurements that are often unavailable 
to researchers, as discussed in Section \ref{sec:proxies}. 

\subsection{Our Contribution}
\label{subsec:contribution}

Our work advances the literature on causal inference with network interference through three primary contributions.

% \begin{enumerate}[(I)]
    % \item \textbf{Flexible Structural Causal Models.} 
    \paragraph{Flexible Structural Causal Models.} 
    We propose a general SCMs framework where the interference network is treated as a latent variable. 
    This framework generalizes existing approaches by accommodating diverse structural configurations, including both \textit{causal proxies} (e.g., noisy measurements) and \textit{non-causal proxies} (e.g., multilayer networks sharing latent coordinates). 
    Furthermore, it explicitly models distinct treatment assignment mechanisms, where treatment allocation depends either on the latent interference network or on the observed proxies.

    \paragraph{Joint Estimation.} 
    We develop a unified Bayesian framework that jointly estimates the causal effects and reconstructs the latent interference network. 
    This contrasts with heuristic two-stage approaches that first infer the network and then treat it as fixed, which often leads to underestimation of uncertainty.
    Our end-to-end estimation leverages information from all available modules (e.g., proxies, treatments, and outcomes) to recover the latent structure, while ensuring that uncertainty in the network is propagated to the final causal estimates.

    \paragraph{Inference for Discrete Latent Networks.} 
    The joint estimation poses a significant computational challenge due to the high-dimensional discrete space of the latent network. 
    We address this by developing a Block Gibbs sampler equipped with Locally Informed Proposals (LIP) \citep{zanella2019informed}. 
    We utilize gradient-based approximations to explore the posterior space, rendering Bayesian inference feasible for network sizes where brute-force marginalization is impossible.
     Furthermore, we provide a formal analysis of the approximation error, deriving a refined efficiency bound that exploits the structure of network updates to improve upon general global bounds \citep[e.g.,][]{grathwohl_2021}.

\section{Settings and Causal Models}
\label{sec:settings}

\subsection{Setup and Notations}
\label{subsec:notations}

Consider a finite population of $N$ units indexed by 
$i\in [N] = \{1,\dots,N\}$. 
Denote the treatment vector of the entire population by $\bZ = (Z_1,\ldots,Z_N)$.
% , and let $\mathcal{Z}$, by the treatments' space (the support of $\bZ$). 
% Let $\bZ$ be the treatment vector of the entire population, and let $\mathcal{Z}$ denote the treatments' space, which is assumed to be finite. 
Denote the observed outcomes by $\bY = (Y_1,\dots, Y_N)$.
The interference network is represented by its $N \times N$ adjacency matrix $\bA^\ast$.
% The interference structure between units is represented by a network. Let $\bA^\ast$ be the $N \times N$ adjacency matrix of the true, latent, interference network. 
Denote by $\bA^\ast_i$ the $i$-th row of $\bA^\ast$.
For simplicity, $\bA^\ast$ is assumed to be unweighted and undirected, that is,
 $A^\ast_{ij}=A^\ast_{ji}=1$ only if there is an edge between units $i$ and $j$,
  and by convention $A^\ast_{ii}=0$. 
%   Extensions for weighted and directed networks are possible (Section~\ref{sec:discussion}). 

The true network $\bA^\ast$ is not observed. Instead, we assume that researchers observe $B\geq 1$ proxy measurement(s) $\cA= \big\{\bA_1,\dots,\bA_B\big\}$ of $\bA^\ast$, such that each $\bA_b$, $b\in [B] =\{1,\ldots,B\}$, is an adjacency matrix of the same dimension as $\bA^\ast$.
Section \ref{subsec:sources} provides detailed examples of possible proxies.
Crucially, treatment interference is assumed to transmit through $\bA^\ast$ and not through any of the proxies in $\cA$. This assumption is formalized in the SCMs given below.
Let $\bX \in \mathbb{R}^{N\times q}$ be a matrix of baseline covariates, where $\bX_i$ is the $q$ covariates of unit $i$. The observed data are $\boldsymbol{\mathcal{D}} = (\bY,\bZ,\cA,\bX)$. Lastly, throughout the paper, we use $p(\cdot)$ and $p(\cdot \mid \cdot)$  as generic notations for  (conditional) probability distributions.

\subsection{Structural Causal Models}
\label{subsec:struc}

We use an SCM framework \citep{Pearl2009}, which models the data-generation mechanism through sequential evaluation of structural equations. 
Our proposed SCMs generalize previous models \citep{sofrygin2017semi, Ogburn2022} by including the interference network generation, as well as the proxy formation, as part of the model.
The directed acyclic graphs (DAGs) in \citet{Ogburn2022} are essentially conditioned on a fixed network, while our SCMs treat the interference network as a realization from a random network model.
That is, each DAG in \citet{Ogburn2022} can be viewed as a slice of our SCMs, given a specific network $\bA^\ast$ realization.
The underlying assumed structures can be represented by population-level DAGs, as illustrated in Figure~\ref{fig:DAGs}. 
\begin{figure}[H]
  \centering
  
  % Top Row: Causal Proxies
  \begin{subfigure}[b]{0.45\textwidth}
   \centering
    \begin{tikzpicture}[>=Stealth,line width=.8pt, scale=0.9, transform shape]
      \node (X) at (0, 2.5) {$\bX$};
      \node (A*) at (2, 2) {$\bA^\ast$};
      \node (A) at (2, .5) {$\cA$};
      \node (Z) at (4, 2) {$\bZ$};
      \node (Y) at (6, 2) {$\bY$};
      \node (U) at (4, .5) {$\bU$};

      \draw[->] (X) -- (A*);
      \draw[->] (X) -- (A);
      \draw[->] (X) to [out=10,in=160] (Z);
      \draw[->] (X) to [out=20,in=150] (Y);
      \draw[->, color=red] (A*) -- (Z); % Observational
      \draw[->, color=blue] (A*) -- (A); % Causal
      \draw[->] (Z) -- (Y);
      \draw[->] (A*) to [out=25,in=160] (Y);
    %   \draw[->] (A) -- (Z); % Removed for Observational
      \draw[->, dashed] (U) to (A*);
      \draw[->, dashed] (U) to (Y);
    \end{tikzpicture}
   \caption{Causal proxies \& Latent assignment}
  \end{subfigure}
  \hfill
  \begin{subfigure}[b]{0.45\textwidth}
   \centering
    \begin{tikzpicture}[>=Stealth,line width=.8pt, scale=0.9, transform shape]
      \node (X) at (0, 2.5) {$\bX$};
      \node (A*) at (2, 2) {$\bA^\ast$};
      \node (A) at (2, .5) {$\cA$};
      \node (Z) at (4, 2) {$\bZ$};
      \node (Y) at (6, 2) {$\bY$};
      \node (U) at (4, .5) {$\bU$};

      \draw[->] (X) -- (A*);
      \draw[->] (X) -- (A);
      \draw[->] (X) to [out=10,in=160] (Z);
      \draw[->] (X) to [out=20,in=150] (Y);
      % \draw[->] (A*) -- (Z); % Removed for RCT
      \draw[->, color=blue] (A*) -- (A); % Causal
      \draw[->] (Z) -- (Y);
      \draw[->] (A*) to [out=25,in=160] (Y);
      \draw[->, color=red] (A) -- (Z); % RCT
      \draw[->, dashed] (U) to (A*);
      \draw[->, dashed] (U) to (Y);
    \end{tikzpicture}
   \caption{Causal proxies \& Proxy assignment}
  \end{subfigure}

  \vspace{1cm} % Add some vertical space between rows

  % Bottom Row: Non-Causal Proxies
  \begin{subfigure}[b]{0.45\textwidth}
   \centering
    \begin{tikzpicture}[>=Stealth,line width=.8pt, scale=0.9, transform shape]
      \node (X) at (0, 2.5) {$\bX$};
      \node (A*) at (2, 2) {$\bA^\ast$};
      \node (A) at (2, .5) {$\cA$};
      \node (Z) at (4, 2) {$\bZ$};
      \node (Y) at (6, 2) {$\bY$};
      \node (U) at (4, .5) {$\bU$};
      \node (V) at (0, 1) {$\bV$};

      \draw[->] (X) -- (A*);
      \draw[->] (X) -- (A);
      \draw[->] (X) to [out=10,in=160] (Z);
      \draw[->] (X) to [out=20,in=150] (Y);
      \draw[->, color=red] (A*) -- (Z); % Observational
      \draw[->] (Z) -- (Y);
      \draw[->] (A*) to [out=25,in=160] (Y);
      % \draw[->] (A) -- (Z); % Removed for Observational
      \draw[->, dashed] (U) to (A*);
      \draw[->, dashed] (U) to (Y);
      \draw[->, color=blue] (V) -- (A); % Non-causal
      \draw[->, color=blue] (V) -- (A*); % Non-causal
    \end{tikzpicture}
    \caption{Non-causal proxies \& Latent assignment}
  \end{subfigure}
  \hfill
  \begin{subfigure}[b]{0.45\textwidth}
   \centering
    \begin{tikzpicture}[>=Stealth,line width=.8pt, scale=0.9, transform shape]
      \node (X) at (0, 2.5) {$\bX$};
      \node (A*) at (2, 2) {$\bA^\ast$};
      \node (A) at (2, .5) {$\cA$};
      \node (Z) at (4, 2) {$\bZ$};
      \node (Y) at (6, 2) {$\bY$};
      \node (U) at (4, .5) {$\bU$};
      \node (V) at (0, 1) {$\bV$};

      \draw[->] (X) -- (A*);
      \draw[->] (X) -- (A);
      \draw[->] (X) to [out=10,in=160] (Z);
      \draw[->] (X) to [out=20,in=150] (Y);
      % \draw[->] (A*) -- (Z); % Removed for RCT
      \draw[->] (Z) -- (Y);
      \draw[->] (A*) to [out=25,in=160] (Y);
      \draw[->, color=red] (A) -- (Z); % RCT
      \draw[->, dashed] (U) to (A*);
      \draw[->, dashed] (U) to (Y);
      \draw[->, color=blue] (V) -- (A); % Non-causal
      \draw[->, color=blue] (V) -- (A*); % Non-causal
    \end{tikzpicture}
    \caption{Non-causal proxies \& Proxy assignment}
  \end{subfigure}
  \caption{DAGs representing the assumed structural models for four different settings, corresponding to causal or non-causal proxies and treatment assignment based on the latent network $\bA^\ast$ or the proxies $\cA$: (a) Causal proxies with latent assignment. (b) Causal proxies with proxy assignment. (c) Non-causal proxies with latent assignment. (d) Non-causal proxies with proxy assignment. The dashed arrows are optional.}
  \label{fig:DAGs}
\end{figure}
These SCMs also include latent variables $\bU$. 
According to the SCMs, the covariates $\bX$ and the latent variables $\bU$ 
are initially generated for the entire population. 
Subsequently, the true interference network $\bA^\ast$ is generated,
 potentially depending on both $\bX$ and $\bU$. 
 We consider two distinct scenarios for how the observed proxies $\cA$ relate to $\bA^\ast$.
\begin{enumerate}[(1)]
    \item \textbf{Causal proxies} (Figures~\hyperref[fig:DAGs]{1(a)-(b)}). 
        Here, the proxies $\cA$ are direct descendants of the true network $\bA^\ast$.
         For example, such proxies represent noisy measurements of the true interference structure.
    \item \textbf{Non-causal proxies} (Figures~\hyperref[fig:DAGs]{1(c)-(d)}).
     In this scenario, both the true network $\bA^\ast$ and the proxies $\cA$ depend
      on additional latent variables $\bV$, and no direct edges in the DAG link $\bA^\ast$ and $\cA$.
       % Thus, conditionally on $\bX$ and $\bV$, the true network and the proxies are independent.
        For example, in a multilayer network \citep{kivela2014multilayer}, 
        each layer represents a distinct type of relationship. 
        The latent variables $\bV$ may represent latent positions \citep{Hoff2002}, 
        and the latent interference network $\bA^\ast$ corresponds to one of these layers.
\end{enumerate}
Next, we distinguish between two scenarios for the treatment assignment mechanism.
In the first, treatments are assigned based on the latent network $\bA^\ast$ (Figures~\hyperref[fig:DAGs]{1(a) and (c)}).
This \textbf{latent assignment} is typical of observational studies, where researchers do not control the treatment assignment, and $\bZ$ may be a function of the true network $\bA^\ast$ rather than the proxies.
% For example, in observational studies, where
% the researchers do not control the treatment assignment, $\bZ$ may be a function of the true latent network $\bA^\ast$ rather than the proxies $\cA$.
In the second scenario, treatments are assigned based on the proxy networks $\cA$ (Figures~\hyperref[fig:DAGs]{1(b) and (d)}).
% For example, if the DAGs correspond to observations from randomized controlled trials where the network is clustered and treatments are randomized to clusters \citep{ugander2013, Eckles2017}.
% When only the proxies $\cA$ are observed, treatment allocation will be a function of the proxies $\cA$ rather than the latent network $\bA^\ast$.
This \textbf{proxy assignment} is common in randomized trials, such as graph cluster randomization, where the network is clustered and treatments are randomized to clusters \citep{ugander2013, Eckles2017}.
As the true network $\bA^\ast$ is latent, treatment allocation in these designs will often be a function of the observed proxies $\cA$.

Each of the DAGs in Figure~\ref{fig:DAGs} implies a different set of 
structural equations for each of the four scenarios. We describe here the structural equations for the scenario of causal proxies with latent assignment (Figure~\hyperref[fig:DAGs]{1(a)}), and provide the set of equations corresponding to each of the other DAGs in Appendix~\ref{apdx.sec:SCM_extra}. 
The structural equations corresponding to Figure~\hyperref[fig:DAGs]{1(a)} are
\begin{equation}
    \label{eq:SCM}
    \begin{aligned}
        \bU_i &= f_U\big(\beps_{U_i} \big), &&\quad i \in [N] \\ 
        \bX_i &= f_X\big( \beps_{X_i} \big), &&\quad i \in [N] \\
        \bA^\ast &= f_{A^\ast}\big(\bX, \bU, \beps_{A^\ast}\big), \\
        \bA_b &= f_{A_b}\big(\bA^\ast,\bX,\beps_{A_b}\big) ,&&\quad b \in [B] \\
        Z_i &= f_Z\big(\bA^\ast, \bX_i, \bX_{-i},\varepsilon_{Z_i}\big), &&\quad i\in [N] \\
        Y_i &= f_Y\big(Z_i,\bX_i, \bZ_{-i},\bX_{-i},\bA^\ast, \bU_i, \varepsilon_{Y_i} \big) &&\quad i\in [N],
    \end{aligned}
\end{equation}
where $f_U, f_X, f_{A^\ast},f_{A_b}, f_Z, f_Y$ are fixed functions and $\beps_U, \beps_X, \beps_{A^\ast},\beps_{A_b}, \beps_Z, \beps_Y$ are noise terms. 
We assume the pairwise noise vectors are independent, e.g., $\beps_Z\indep \beps_Y$,  
but do not limit the dependence structure between units (e.g., the dependence of $\varepsilon_{Y_i}$
 and $\varepsilon_{Y_j}$). 
 The proxy models $f_{A_b}$ can be adapted to include settings where proxies also influence each other. For example, if the proxies follow an auto-regressive model where $\bA_b$ is a function of $\bA^\ast$ and all preceding proxies $\bA_1,\ldots, \bA_{b-1}$, we can modify $\bA_b$ model in \eqref{eq:SCM} to 
 \begin{equation*}
     \begin{aligned}
        \bA_1 &=
        f_{A_1}\big(\bA^\ast,\bX,\beps_{A_1}\big) ,&& \\
        \bA_b &= f_{A_b}\big(\bA_1,\ldots,\bA_{b-1},\bA^\ast,\bX,\beps_{A_b}\big) ,&&\quad b>1. \\
     \end{aligned}
 \end{equation*}
Specific network models for $\bA^\ast$ and proxies $\cA$ are detailed in Section~\ref{subsec:net_models}.
 In the non-causal proxies scenarios (Figures~\hyperref[fig:DAGs]{1(c) and 1(d)}),
  we add a model for the latent variables $\bV_i = f_{V}(\beps_{V_i}), 
  i\in [N]$, let $\bV$ influence $\bA^\ast$ by including $\bV$ in $f_{A^\ast}$, 
  and replace $\bA^\ast$ with $\bV$ in $f_{A_b}$. 
In the proxy-based treatment assignments (Figures~\hyperref[fig:DAGs]{1(b) and 1(d)}), 
we replace $\bA^\ast$ in $f_Z$ with the proxies $\cA$. Details are given in Appendix~\ref{apdx.sec:SCM_extra}.

% For example, in randomized controlled trials (RCTs), treatment assignment often depends on the network structure, e.g., by first clustering the network \citep{ugander2013,Eckles2017}.
% When only the proxies $\cA$ are observed, treatment allocation will depend directly on the proxies $\cA$ rather than on the latent network $\bA^\ast$,
% implying the edge $\cA \rightarrow \bZ$ while $\bA^\ast \not \rightarrow \bZ$ in Figures~\hyperref[fig:DAGs]{1(b)} and \hyperref[fig:DAGs]{1(d)}.
% On the other hand, in observational studies
% treatment assignment is not controlled by the researchers and therefore 
% a function the true latent network $\bA^\ast$ rather than the proxies $\cA$. 
% This situation is represented by the edge $\bA^\ast \rightarrow \bZ$ while $\cA \not \rightarrow \bZ$ in Figures~\hyperref[fig:DAGs]{1(a)} and \hyperref[fig:DAGs]{1(c)}.

Under all the SCMs, the outcomes $\bY$ are a function of all preceding variables except $\cA$. 
Specifically, $\bZ$ influences the outcome of each unit $Y_i$ through the direct treatment, $Z_i$, and the treatments of other units, $\bZ_{-i}$.
Structural assumptions regarding interference can be expressed by the specification of $f_Y$.
For example, if interference is assumed to occur only between neighbors in $\bA^\ast$, 
then $f_Y$ depends on $\bZ_{-i}$ only through $\bZ_{\mathcal{N}^\ast_i}$, 
where $\mathcal{N}^\ast_i = \{j : \bA^\ast_{ij}=1\}$ is the neighbor set of unit $i$ in $\bA^\ast$. 
The outcome model in \eqref{eq:SCM} includes high-dimensional variables: 
covariates of other units $\bX_{-i}\in \mathbb{R}^{(N-1) \times q}$, 
treatments of other units, $\bZ_{-i}\in \mathbb{R}^{N-1}$, 
and the true network $\bA^\ast \in \{0,1\}^{N\times N}$.  
Practical applications often require simplifying this complex dependence through modeling assumptions.
We assume that $Y_i$ depends on these variables only through 
summarizing functions $\phi_1, \phi_2$, and $\phi_3$, leading to a modified SCM for $\bY$
\begin{equation}
    \label{eq:SCM_mod}
    % \begin{aligned}
        % Z_i &= f_Z\big(\bX_i, \phi_2(\bX_{-i},\bA^\ast), \phi_{3,i}(\bA^\ast),\varepsilon_{Z_i}\big) &&\quad i\in [N] \\
        Y_i = f_Y\big(Z_i,\bX_i, \phi_1(\bZ_{-i},\bA^\ast), \phi_2(\bX_{-i},\bA^\ast),\phi_{3,i}(\bA^\ast), \bU_i, \varepsilon_{Y_i} \big), \; i\in [N].
    % \end{aligned}
\end{equation}
The function $\phi_1(\bZ_{-i},\bA^\ast)$ 
defines the \emph{exposure mapping},
summarizing treatments assigned to units other than $i$. 
For example, under the common assumption of neighborhood interference \citep{Forastiere2020,Ogburn2022}, 
this function simplifies to $\phi_1(\bZ_{\mathcal{N}^\ast_i},\bA^\ast)$, 
which depends solely on the treatments of immediate neighbors. For binary treatments, 
a typical example of such an exposure mapping is the proportion of treated neighbors,
$\phi_1(\bZ_{\mathcal{N}^\ast_i},\bA^\ast)=(d_i^{\ast})^{-1} \sum_{j \in \mathcal{N}^\ast_i}Z_j$,
where $d^\ast_i=\lvert \mathcal{N}^\ast_i \rvert$ is the degree of unit $i$. 
The neighborhood interference assumption can be relaxed, for example, by considering 
higher-order neighbors with decreasing magnitudes \citep{leung2022causal}. 
Heterogeneous effects can be included by defining appropriate $\phi_1 = \phi_1(\bZ_{\mathcal{N}^\ast_i},\bX_{\mathcal{N}^\ast_i},\bA^\ast)$.
The function $\phi_2(\bX_{-i},\bA^\ast)$ typically summarizes covariates of neighbors, for example, by using the mean or sum of neighbors' covariates \citep{Tchetgen2020, Ogburn2022}. 
Finally, the function $\phi_3(\bA^\ast)$ reduces the network $\bA^\ast$ into summary statistics such as node centrality, degree, or eigenvectors; see \citet{Kolaczyk2009} for additional examples and analysis of network summary statistics. Here, $\phi_{3,i}$ represents the statistics corresponding to unit $i$. 
% Although different functions $\phi_2$ and $\phi_3$ could be specified for $\bZ$ and $\bY$, we use identical functions for simplicity of presentation.

Figures~\hyperref[fig:DAGs]{1(a)-(c)} imply that all back-door paths between $\bZ$ and 
$\bY$ are blocked by conditioning on the covariates $\bX$ and the true network $\bA^\ast$, in those three scenarios\footnote{A subtle case 
arises in Figure~\hyperref[fig:DAGs]{1(d)} in the scenario of non-causal proxies $\cA$ and proxy-based treatment assignment where $\cA \rightarrow \bZ$. In that case, $\bA^\ast$ is a collider since conditioning on it opens the back-door path 
$\bY \leftarrow \bU \rightarrow \bA^\ast \leftarrow \bV \rightarrow \cA \rightarrow \bZ$.
However, this back-door path can be blocked by the observed proxies $\cA$.}.
Our SCMs explicitly allow latent variables 
$\bU$ to affect both the true interference network $\bA^\ast$ and the outcomes $\bY$, but exclude direct effects of $\bU$ on treatments $\bZ$. 
Such latent confounding might arise due to \emph{latent homophily} \citep{Shalizi2011}, 
where unmeasured similarities among units influence both network formation and outcomes
 --- a phenomenon also known as \emph{network endogeneity} \citep{pinkham2013}.
As we discuss in Section~\ref{sec:inference}, our Bayesian inferential framework
can naturally accommodate this latent confounding.
% Nevertheless, certain causal estimands, such as those involving interventions on the network structure,
% are generally not identifiable from the observed data alone,
% thus their estimation relies on additional model assumptions and priors (Section~\ref{sec:discussion}).

We adopt the standard positivity assumption from graphical models \citep{Pearl2009}, requiring that all variables have positive joint probability across their respective supports, formally $p(\bY,\bZ,\cA,\bA^\ast,\bV,\bX,\bU)>0$. This assumption can be relaxed to conditional positivity, requiring positivity to hold only given $\bV,\bU,\bX$.

\subsection{Causal Estimands}
\label{subsec:estimands}

Our analysis focuses on estimands that quantify the effects of hypothetical interventions on the population-level treatments, corresponding to various treatment assignment policies.  
Although the definition of estimands does not depend on the exposure mapping $\phi_1$, 
their numerical values rely on modeling assumptions. We explicitly distinguish between structural assumptions, such as those concerning the exposure mapping, and causal estimands, which represent interventions directly on the population treatments $\bZ$ \citep{Saevje2023}. 
Interventions on exposure levels (defined by $\phi_1$) 
are generally not feasible in practice, especially at the population level. 
Throughout the paper, we treat the population covariates $\bX$ as fixed. 
All estimands we define are explicitly conditioned on the true latent network $\bA^\ast$
 and covariates $\bX$, although we often omit this conditioning for brevity. 
 Following the terminology in \citet{Li2023}, the estimands can be viewed as mixed average treatment effects.

Let $Y_i(\bz)$ be the outcome of unit $i$ under an intervention that sets the treatment vector $\bZ$
to a specific value $\bz$ with probability one, 
similar to the operator $do(\bZ=\bz)$ \citep{Pearl2009}. 
Using the SCM notations in \eqref{eq:SCM}, $Y_i(\bz)$ corresponds to the outcome under the intervention that sets
\begin{equation*}
    \begin{aligned}
        Z_i &= z_i &&\quad i\in [N] \\
        Y_i &= f_Y\big(z_i,\bX_i, \bz_{-i},\bX_{-i},\bA^\ast, \bU_i, \varepsilon_{Y_i} \big) &&\quad i\in [N].
    \end{aligned}
\end{equation*}
We consider estimands corresponding to three distinct types of treatment policies:
 \emph{static}, \emph{dynamic}, or \emph{stochastic} \citep{Pearl2009,Ogburn2022}. 
Let $\mu_i(\bz) \equiv \mu_i(\bz; \bA^\ast,\bX) = \bbE[Y_i(\bz) \mid \bA^\ast, \bX]$ be the expected outcome for unit $i$ under an intervention setting the treatment vector $\bZ$ to $\bz$, given the true network and covariates.

We define the \emph{static policy} as setting the treatment vector $\bZ$ to a fixed value $\bz$.
 Let $\mu(\bz) = N^{-1}\sum_{i=1}^{N}\mu_i(\bz)$. 
 The corresponding estimand compares two static interventions
\begin{equation*}
    % \label{eq:avg_static_po}
  % $
    % \mu(\bz) = N^{-1}\sum_{i=1}^{N}\mu_i(\bz).
    \tau(\bz,\bz') = \mu(\bz) - 
    \mu(\bz').%, \; \text{where}\; \mu(\bz) = N^{-1}\sum_{i=1}^{N}\mu_i(\bz).
  % $
\end{equation*}
For example, for binary treatments,  $\tau(\boldsymbol{1},\boldsymbol{0})$ 
is the TTE -- the effect of treating all units versus treating none. 
% commonly called the total treatment effect \citep{Yu2022}. 
% For continuous treatments, one might consider a marginal perturbation of the observed treatments, such as $\tau(\bZ + \boldsymbol{\delta},\bZ)$ for some small $\boldsymbol{\delta} \in \mathbb{R}^N$.

A \emph{dynamic policy} assigns treatments deterministically based on covariates and the network structure. Let $h(\bX)$ denote a deterministic function that sets treatments based on covariates $\bX$. With a slight abuse of notation, we write the corresponding average expected outcome as $\mu(h) = N^{-1}\sum_{i=1}^{N}\mu_i(h)$. The estimand
\begin{equation*}
    % $ 
    % \label{eq:dynamic_estimand}
    % \mu(h) = N^{-1}\sum_{i=1}^{N}\mu_i(h).
    \tau(h_1,h_0)=\mu(h_1)-\mu(h_0),
    % $
\end{equation*}
compares two dynamic policies, $h_1$ and $h_0$. For example,
 comparing two different age thresholds when studying policies that treat all individuals 
 above a certain age.
 % , or evaluating the impact of treating only highly influential units in the network.

Finally, a \emph{stochastic policy} assigns treatments according to a distribution $\pi_{\balpha}(\bz)$, parameterized by 
$\balpha$. The policy $\pi_{\balpha}$ can depend on $\bX$, but we omit this option for brevity. 
Under such a policy, the intervention $do(\bZ=\bz)$ occurs with probability $\pi_{\balpha}(\bz)$ \citep{Pearl2009}. The average of expected outcomes marginalized over the stochastic policy is
\begin{equation*}
    % \label{eq:mean_stochastic_inter}
    % $
    \mu(\balpha) = N^{-1}\sum_{i=1}^{N}\sum_{\bz}\pi_{\balpha}(\bz)\mu_i(\bz).
    % $
\end{equation*}
Accordingly, the estimand $\tau(\balpha_1,\balpha_0)$ 
evaluates the impact of implementing one stochastic policy $\pi_{\balpha_1}$ 
compared to another $\pi_{\balpha_0}$. For example, the effect of randomly treating
 $80\%$ versus $20\%$ of the population.
The estimand $\tau(\balpha_1,\balpha_0)$ is similar to the overall causal effect defined by
 \citet{Hudgens2008}. Compared to \citet{Hudgens2008}, who considered a design-based framework,
  we replace the outcomes under the intervention with their expected value. 
  Similar approaches were taken by others \citep{Tchetgen2020, Ogburn2022}. 

Under the SCMs described in Figures~\hyperref[fig:DAGs]{1(a)--(c)}, all back-door paths between $\bZ$ and $\bY$ are blocked by conditioning on $\bX$ and $\bA^\ast$. Thus, using the back-door adjustment formula \citep{Pearl2009}, we have
\begin{equation*}
% \label{eq:identitfication}
        % $
        \mu_i(\bz) = \bbE[Y_i \mid \bZ=\bz,\bA^\ast,\bX]
        % $.
\end{equation*}
Under the modified SCM \eqref{eq:SCM_mod}, this further simplifies to
$\bbE[Y_i \mid z_i,\bX_i, \phi_1(\bz_{-i},\bA^\ast),\phi_2(\bX_{-i},\bA^\ast),\phi_{3,i}(\bA^\ast)]$.  In the case of Figure~\hyperref[fig:DAGs]{1(d)},
where conditioning on $\bA^\ast$ opens a back-door path between $\bZ$ and $\bY$,
an additional conditioning on the proxies $\cA$ is required for identification.
The back-door adjustment is therefore $\mu_i(\bz) = \bbE[Y_i \mid \bZ=\bz,\cA,\bA^\ast,\bX]$.

Since the interference network $\bA^\ast$
and the parameters governing the conditional expectation of outcomes $\bY$ are both latent, all estimands depend on these unknowns. The interpretation of the estimands is therefore conditional on the true interference network $\bA^\ast$,
similar to \citet{Ogburn2022}, with the distinction that here $\bA^\ast$ is latent. Thus, the interpretation of the estimands is the same had the network been known. Note also that because we consider interventions on the population treatments, the interpretation of the causal estimands does not need to take the network into account. For example, the interpretation of randomly treating 80\% versus 20\% of the population remains the same.
In Section~\ref{subsec:bayes_g_formula}, we describe how to estimate the estimands using samples from the posterior distribution. 
% Before presenting the estimation scheme, we first review types of proxy measurements and potential models for the true and proxy networks.

\subsection{Probabilistic Models}
\label{subsec:prob_models}

While the SCMs described above provide a general causal framework, practical applications require specifying probabilistic models for some of the structural equations. 
The components requiring modeling depend on the assumed SCM structure. For example, whether the proxies are causal or not and whether treatment assignment is based on the latent network or observed proxies. Throughout our analysis, we treat the covariates $\bX$ as fixed.

Consider the setting of causal proxies with latent assignment (Figure~\hyperref[fig:DAGs]{1(a)}).
In this case, researchers must specify probabilistic models for the outcome, treatment assignment, latent network, and proxy networks measurement mechanisms.
Let $\boeta$ parameterize the outcome model $p(\bY \mid \bZ,\bA^\ast,\bX,\boeta)$,
$\bbeta$ parameterize the treatment assignment model $p(\bZ \mid \bA^\ast,\bX,\bbeta)$,
$\btheta$ the true network model $p(\bA^\ast \mid \bX,\btheta)$, 
and $\bgamma$ the observed proxy networks model $p(\cA \mid \bA^\ast, \bX,\bgamma)$.
Let $\bTheta = (\boeta, \bbeta, \btheta, \bgamma)$ denote the set of all the structural model parameters. Under this specification, the complete data likelihood of $(\bD, \bA^\ast)$  can be factorized according to Figure~\hyperref[fig:DAGs]{1(a)} as
\begin{equation}
    \label{eq:data_factorization}
    p(\bD, \bA^\ast \mid \bTheta) = p(\bY \mid \bZ,\bA^\ast,\bX, \boeta)
    p(\bZ \mid \bA^\ast, \bX, \bbeta)p(\cA \mid \bA^\ast, \bX, \bgamma)p(\bA^\ast \mid \bX, \btheta),
    % p(\bY,\bZ,\cA,\bA^\ast \mid\bX) = p(\bY \mid \bZ,\bA^\ast,\bX)p(\bZ \mid \bA^\ast, \bX)p(\cA \mid \bA^\ast, \bX)p(\bA^\ast \mid \bX).
\end{equation}
highlighting that the complete data likelihood is composed of three data modules (outcomes, treatments, and proxies) and the prior network model.

In scenarios with proxy-based treatment assignment (Figures~\hyperref[fig:DAGs]{1(b) and (d)}), the treatment assignment model $p(\bZ \mid \cA, \bX)$ may depend on the proxies $\cA$ but not on $\bA^\ast$. Similarly, in non-causal proxy settings (Figures~\hyperref[fig:DAGs]{1(c)-(d)}), the proxy model $p(\cA \mid \bA^\ast, \bX, \bgamma)$ and true network model $p(\bA^\ast \mid \bX, \btheta)$ are replaced by models conditioned on the latent common causes $\bV$, i.e., $p(\cA \mid \bV, \bX, \bgamma)$ and $p(\bA^\ast \mid \bV, \bX, \btheta)$.
Scenarios with latent variables $\bU$ are discussed in Appendix~\ref{apdx.sec:posterior}.

\section{Proxy Measurements and Network Models}
\label{sec:proxies}

\subsection{Types of Proxy Networks}
% \subsection{Types of Proxy Measurements}
\label{subsec:sources}

We categorize the observed proxy measurements $\cA$ of the true latent interference network $\bA^\ast$ into the following types:

\begin{enumerate}[(i)]
\item \textbf{Measurement error.} Researchers observe a single \citep{Butts2003} or multiple \citep{de2023latent} proxy networks, each reflecting the true network $\bA^\ast$ with varying degrees of error, which can depend on covariates. For unweighted networks, measurement error corresponds to edge misclassification.
For instance, in survey-derived networks, respondents might forget certain connections or misunderstand reporting instructions, resulting in missing or spurious edges. 

Another scenario occurs under partial interference assumptions \citep{Hudgens2008}, where researchers observe networks composed of distinct clusters, while the true interference network may have cross-cluster relations, thus ignoring potential cross-cluster interactions \citep{Weinstein2023}. 

\item \textbf{Multiple sources.} Researchers obtain several proxies derived from distinct data sources, each potentially providing complementary information about the latent interference network. These sources can include online platforms, surveys, biological measurements, physical traces, or spatial data, each capturing different relational aspects. For example, \citet{Goyal2023} combined epidemiological contact tracing, genetic testing, and behavioral surveys to estimate disease transmission networks.
Although valuable, integrating multiple proxy networks poses inference challenges, as proxies may differ in their accuracy and relevance for capturing the interference mechanism under study

\item \textbf{Multilayer networks.}
Researchers collect multiple networks representing distinct types of relationships among the same set of units \citep{kivela2014multilayer}. 
Unlike multiple-source proxies that offer complementary information about a single underlying structure, multilayer networks explicitly represent fundamentally different social relations.
An example would be one network layer representing friendship ties and another representing work collaboration.
Within our framework, multilayer networks can be represented using two modeling approaches.
In the first, the observed layers are non-causal proxies (Figures~\hyperref[fig:DAGs]{1(c)-(d)}), 
generated from shared latent positions \citep{salter2017latent,Sosa_Betancourt_2022}.
Alternatively, one can view the true latent network $\bA^\ast$ as a common ancestor 
of all the observed layers, representing causal proxies (Figures~\hyperref[fig:DAGs]{1(a)-(b)}). 
\end{enumerate}

Next, we show how these different types of proxies can be expressed mathematically within our modeling framework.

\subsection{True and Proxy Network Models}
\label{subsec:net_models}

Statistical analysis of network data encompasses a wide variety of models suited for
 different applications \citep{Fienberg2012}. 
Common examples include random graphs 
\citep{erdHos1959random}, stochastic block models (SBM) \citep{Holland1983}, exponential random graphs \citep{Holland1981}, 
and latent space models (LSM) \citep{Hoff2002}.
Our framework is flexible regarding the choice of the model for the true network $\bA^\ast$,
 subject to two practical constraints: computational feasibility and parameter identifiability.
 We further discuss the parametric identifiability of the structural model parameters $\bTheta$ in Section~\ref{sec:identifiability}.

The observed proxy network(s) can be modeled to reflect different measurement scenarios. 
We now illustrate how some of the previously mentioned examples (Section \ref{subsec:sources}) can be formalized within our framework.

\begin{exmp}[Random noise]
    \label{exmp:rnd.noise}
  A single network $\bA$ is observed, generated by independently measuring each edge of the latent network $\bA^\ast$ with random error
  % The false and true positive rates are denoted by $\gamma_0$ and $\gamma_1$, respectively.
  % Formally,
    % \begin{equation*}
    %     p(\bA \mid \bA^\ast, \bX,\bgamma) =
    %     \prod_{i>j}\Pr(A_{ij} \mid A^\ast_{ij},\bgamma),
    % \end{equation*}
    % where
    $\Pr(A_{ij}=1 \mid A^\ast_{ij}=k,\bgamma) = \gamma_k,\; k=0,1$, corresponding to false positive rate when $k=0$ and true positive rate when $k=1$. Extending to covariate-dependent measurement error is possible by modeling the error rates as a function of covariates $\bX$.
\end{exmp}

\begin{exmp}[Censoring]
    \label{exmp:censor}
    The observed network $\bA$ is a censored version of $\bA^\ast$, namely,  each unit is allowed to report at most $C>0$ edges (friends) from its true connections \citep{Griffith2021}. If unit $i$ has more than $C$ true edges (i.e., if $d^\ast_i>C$ where $d^\ast_i$ is its degree in $\bA^\ast$), assume the unit independently reports a random subset of $C$ neighbors; otherwise, it reports all. 
    Assuming an edge is observed if either node reports it and no false positives, i.e.,
    $\Pr(A_{ij}=1 \mid A^\ast_{ij}=0)=0$, the reporting probability for existing edges is
\begin{equation*}
\Pr(A_{ij}=1 \mid A^\ast_{ij}=1) =
1 - 
\Bigg[1- \min\left(1, \frac{C}{d^\ast_i}\right)\Bigg]
\times
\Bigg[1- \min\left(1, \frac{C}{d^\ast_j}\right)\Bigg].
\end{equation*}
 % Covariate dependence can be included by adjusting the reporting probability.
    % Dependence on covariates $\bX$ can be added, e.g., by taking the friend's selection process as a function of $\bX$. Possible misclassification of reported edges can be modeled by multiplying the friend's selection probabilities $\min(1,C/d_i)$ by some parameter similar to Example~\ref{exmp:rnd.noise}.
\end{exmp}

\begin{exmp}[Cross-cluster contamination]
\label{exmp:contamination}
A single network $\bA$ consisting of well-separated clusters is observed, while $\bA^\ast$ is generated from a network model that may include cross-cluster edges, such as SBM \citep{Holland1983}. Researchers assume that interference is only possible within clusters (partial interference), but cross-cluster contamination exists. Let $X_{1,i}$ denote the  cluster of unit $i$. One possible measurement model is
    \begin{equation*}
        \Pr(A_{ij}=1 \mid A^\ast_{ij}=k, X_{1,i}, X_{1,j}, \bgamma) = \gamma_k \mathbb{I}\{X_{1,i} = X_{1,j}\}, \;  k=0,1,
    \end{equation*}
     with false positive rate $\gamma_0$ and true positive rate $\gamma_1$ within clusters, and no edges observed between clusters. This model can be generalized, e.g., by replacing $\gamma_k$ with $\gamma_{k,ij}$ that depends on the cluster memberships. 
 While the proxy network $\bA$ provides no information regarding cross-cluster edges, our estimation framework allows for learning these structures through other available modules, as described in Section~\ref{sec:inference}.
\end{exmp}

\begin{exmp}[Repeated noisy measurements]
\label{exmp:repeated.noisy}
    Multiple proxy networks are observed ($B > 1$), each is a noisy measurement of the latent network $\bA^\ast$ \citep{de2023latent}.
    Each proxy follows some measurement model
    % \begin{equation*}
    $
    p(\bA_b \mid \bA^\ast,\bX,\bgamma_b) 
    % =  \prod_{i>j} \text{Pr}_b (A_{b,ij} \mid A^\ast_{ij},\bgamma_b),
    $
    % \end{equation*}
    which may be identical or vary across proxies. 
    Dependence among proxies can be modeled hierarchically or temporally. In a temporal setting, the measurement of network $\bA_b$ can follow an auto-regressive model, e.g., by taking $p(\bA_b \mid \bA_1,\ldots,\bA_{b-1}, \bA^\ast, \bX, \bgamma_b)$.
    
    % for some function $\Pr_b(A_{b,ij} \mid A^\ast_{ij},\bgamma_b)$ (e.g., as in  Examples \ref{exmp:rnd.noise}-\ref{exmp:censor}), and as a special case the models can be identical, i.e., $\Pr_b(A_{b,ij} \mid A^\ast_{ij},\bgamma_b)= \Pr(A_{b,ij} \mid A^\ast_{ij},\bgamma_b)$. The joint distribution of $\bA_1,\dots,\bA_B$ can be independent or dependent (e.g., hierarchical model). It can be extended for auto-regressive models by viewing $B$ as time and taking $p(\bA_b \mid \bA^\ast, \bX, \bA_{<b}, \gamma_b)$ where $\bA_{<b}$ are all proxies measured prior to $\bA_b$.
\end{exmp}

\begin{exmp}[Multilayer networks]
\label{exmp:multi.layer.nets}
Researchers observe multiple networks on the same $N$ units, each representing distinct relationship types. As previously discussed, these observed multilayer networks can be modeled under two different assumptions about their relationship with the latent network $\bA^\ast$:
\begin{enumerate}[(a)]
\item \textbf{Non-causal proxies.} Both the true network $\bA^\ast$ and the observed proxies $\cA$ 
depend on common latent variables $\bV$ (Figure~\hyperref[fig:DAGs]{1(a)-(b)}).
 Given these latent variables and covariates $\bX$,
 the observed proxies and
 the latent network are independent. For instance, each network layer (including $\bA^\ast$)
  can be generated from a latent space model \citep[LSM;][]{Hoff2002}, 
  where $\bV$ represents latent unit-level positions. 
Latent positions $\bV$ can either be shared across all layers, 
including the latent network \citep{salter2017latent}, or structured hierarchically \citep{Sosa_Betancourt_2022}. 
Under this setup, proxy networks $\cA$ provide information about $\bA^\ast$ only indirectly through 
their common latent positions $\bV$.
\item  \textbf{Causal proxies.} 
The network $\bA^\ast$ is a latent predecessor of all the observed proxies $\cA$ (Figure~\hyperref[fig:DAGs]{1(c)-(d)}).  
For example, each proxy network layer $b$ can be modeled as
\begin{equation*}
\text{Pr}_{b}(A_{b,ij}=1 \mid A^\ast_{ij},\bX,\bgamma_b) = r^{-1}\big(\gamma_{b,0} A^\ast_{ij} + \bgamma_b'\widetilde{\bX}_{ij}\big),
\end{equation*}
where $r: [0,1] \to \mathbb{R}$ is a link function and $\widetilde{\bX}_{ij}$ are edge-level covariates, derived from the nodal covariates $\bX_i$ and $\bX_j$ (e.g., as done in Section~\ref{sec:sim}).
\end{enumerate}
   
    % Researchers observe multi-layer collection of networks $\bA = (\bA_1,\dots,\bA_B)$. Each $\bA_b$ measures different social relationships. It is possible to model the observed multi-layer networks using LSM with shared latent positions:
    % \begin{equation*}
    %     p(\bA_b \mid \bX,\bU, \gamma_b) = \prod_{i>j}\psi(\gamma_b - \lVert \bU_i - \bU_j \rVert)^{A_{b,ij}}(1-\psi(\gamma_b - \lVert \bU_i - \bU_j \rVert))^{1-A_{b,ij}},
    % \end{equation*}
    % for $\psi: \mathbb{R} \to [0,1]$ and latent positions $\bU_i \in \mathbb{R}^d$. Covariates can also be included. Dependence between the networks can be modeled with multinomial \citep{salter2017latent} or hierarchical \citep{Sosa_Betancourt_2022} models. 
    % The true network $\bA^\ast$ can be viewed as either a ``consensus" or aggregate version or a latent ``ancestor" of the observed proxy networks. That can be achieved with a simple union \citep{salter2017latent}, or in a model-based approach \citep{Sosa_Betancourt_2022}.
\end{exmp}

\section{Identifiability of Structural Parameters}
\label{sec:identifiability}
Under the proposed class of SCMs (Figure~\ref{fig:DAGs}), establishing the identifiability of the structural parameters $\bTheta$ is a fundamental prerequisite for valid statistical inference.
It ensures that the true structural parameters $\bTheta$ can, in principle,
be uniquely recovered from the distribution of the observed data $\bD = (\bY, \bZ, \cA, \bX)$.
In the absence of identifiability, the posterior distribution in 
a Bayesian framework (Section~\ref{sec:inference})
may be heavily influenced by the prior, as the data alone are insufficient to uniquely recover the parameters.
% may essentially 
% reflect the prior rather than the information contained in the data.

In our framework, the network $\bA^\ast$ is latent, rendering the observed data likelihood
a complex mixture model of the 
complete data likelihood \eqref{eq:data_factorization}, $p(\bD \mid \bTheta) = \sum_{\bA^\ast} p(\bD,\bA^\ast \mid \bTheta)$.
Since this summation is computationally intractable due to the high-dimensional discrete space of $\bA^\ast$,
establishing global identifiability analytically for the general case is prohibitively complex.

Instead, we investigate the global identifiability of the structural parameters in simplified settings.
We illustrate how the conditional independence structure implied by the SCMs allows for the unique recovery of parameters using moment equations \citep{rothenberg1971}.
Specifically, because the outcomes and proxy networks are conditionally independent given the latent network $\bA^\ast$, their joint moments provide distinct moment equations that can be solved to recover the parameters $\bTheta$.
This analysis highlights the mechanism through which the distinct data modules inform the structural parameters.
This perspective aligns with the identification results for latent 
variable models using at least three independent views \citep{Allman2009}.

In addition, we discuss sufficient conditions for local identifiability, which require that the likelihood function is not flat in a neighborhood of the true parameters.
This is assessed via the nonsingularity of the Fisher Information Matrix (FIM).
In Appendix~\ref{apdx.sec:identif}, we utilize the Missing Information Principle \citep{Louis1982} to characterize the observed data FIM as the difference between the complete data information and the missing information induced by the latent network.
Intuitively, we show that the parameters are locally identifiable if the information provided by the distinct data modules dominates the uncertainty surrounding the latent network.
Details are provided in Appendix~\ref{apdx.sec:identif}.

\begin{remark}[From Structural Parameters to Causal Estimands]
\label{remark:1}
% It is important to clarify the relationship between the identification of the structural parameters $\bTheta$ and the causal estimands defined in Section~\ref{subsec:estimands}.
The causal estimands defined in Section~\ref{subsec:estimands} are functions of the specific realization of the latent network $\bA^\ast$.
Consequently, establishing the identifiability of $\bTheta$ does not imply that the causal effects are nonparametrically identified, as $\bA^\ast$ remains latent.
Instead, the identification of $\bTheta$ allows for the recovery of the conditional distribution of the latent network given the data, $p(\bA^\ast \mid \bD, \bTheta)$, and subsequently, the distribution of the causal estimands.
This distinction is central to our framework: we do not seek to recover a point estimate of the causal effects, but rather to recover their posterior distribution.
This motivates our Bayesian inference approach in Section~\ref{sec:inference}, which targets this full posterior distribution.
\end{remark}

\subsection{Global Identifiability In Simplified Settings}

Global identifiability requires that the map from the structural parameters $\bTheta$ to the observed data likelihood $p(\bD \mid \bTheta)$ is injective. 
Define the \textit{moment map} $\mathcal{M}(\bTheta) = \bbE_{\bD \mid \bTheta}[\mathbf{m}(\bD)]$ as the expected value of a vector of observable moments $\mathbf{m}(\bD)$.
The parameter $\bTheta$ is globally identified with respect to these moments if
$\mathcal{M}(\bTheta)$ is injective \citep{rothenberg1971}.

In the context of the proposed SCMs, we construct these moment equations by exploiting the conditional independence structure of the model.
For instance, in Figure~\ref{fig:DAGs}(a), the independence of outcomes $\bY$ and proxies $\cA$ given the latent network $\bA^\ast$ and covariates $\bX$ implies that the observed correlation (given $\bX$) between $\bY$ and $\cA$ is entirely mediated by $\bA^\ast$.
Consequently, the observable cross-moments become explicit functions of $\bTheta$.
By deriving these analytical expectations, we can construct a system of moment equations and assess the injectivity of the map.
We now demonstrate this identification strategy by explicitly deriving the system of moment equations for a simplified class of models.

Consider a setup without covariates,
where the true network $\bA^\ast$ follows an Erd\H{o}s-R\'enyi model \citep{erdHos1959random},
the outcome model is linear with exposure mapping $\phi_1$ (as in \eqref{eq:SCM_mod})
representing the number of treated neighbors, 
and proxy network measurements are subject to random measurement error as in Example~\ref{exmp:rnd.noise}.
The models are specified as follows:
\begin{equation}
    \label{eq:simple_model}
    \begin{aligned}
        A^\ast_{ij} &\sim \text{Ber}(\theta), \\
        A_{ij} \mid A^\ast_{ij} = k &\sim \text{Ber}(\gamma_k), \; k=0,1,\\
        Z_i &\sim \text{Ber}(p_z), \\
        Y_i &= \eta_1 Z_i + \eta_2 \sum_{j \neq i} A^\ast_{ij} Z_j + \varepsilon_i, \\
    \end{aligned}
\end{equation}
where $\varepsilon_i$ are independent and zero-mean error terms.
The set of parameters is $\bTheta = (\theta, \gamma_0, \gamma_1, p_z, \eta_1, \eta_2)$.
The parameters $p_z$ and $\eta_1$ are identified from the observed first moments $\bbE[Z_i]$ and 
$\bbE[Y_i \mid Z_i=1] - \bbE[Y_i \mid Z_i=0]$, respectively.
Deriving the analytical moments (Appendix~\ref{apdx.sec:identif}),
 we obtain the following system of three equations:
\begin{align*}
    \bbE[Y_i] &= \eta_1 p_z + \eta_2 (N-1) p_z \theta, \\
    \Pr(A_{ij}=1) &= \gamma_1 \theta + \gamma_0 (1-\theta), \nonumber \\
    \text{Cov}(Y_i, d_i^{obs}) &= \eta_2 (N-1) p_z \theta (1-\theta) (\gamma_1 - \gamma_0),
\end{align*}
where $d_i^{obs} = \sum_{j \neq i} A_{ij}$ is the observed degree of unit $i$ in the proxy network.
This system presents three equations for four unknowns ($\theta, \gamma_0, \gamma_1, \eta_2$). 
Consequently, the simplified model \eqref{eq:simple_model} with
 a single proxy network is not globally identifiable without further restrictions.
However, this lack of identification is resolved under several realistic extensions of the model,
which provide the necessary additional information to achieve global identifiability (see Appendix~\ref{apdx.sec:identif} for derivations):
\begin{enumerate}
    \item \textbf{Constraints on Measurement Error:} If reliability data is 
    available such that the false positive rate $\gamma_0$ is exogenously identified (or fixed to zero),
     the system reduces to three unknowns, rendering the parameters globally identifiable.
    
    \item \textbf{Multiple Proxies:} When a second proxy network $\bA_2$ is available,  such that $\bA_2\indep \bA_1| \bA^\ast$, and  $\bA_1$ and $\bA_2$ are identically distributed,  we can utilize the expected disagreement between the proxies 
    $\frac{2}{N(N-1)} \sum_{i<j} |A_{1,ij} - A_{2,ij}|$.
    % , we derive a fourth independent equation:
    % \begin{equation*}
    %     \bbE[m_{AD}] = 2 \left[\gamma_1 (1-\gamma_1) \theta + \gamma_0 (1-\gamma_0) (1-\theta)\right].
    % \end{equation*}
    This additional moment allows for the unique recovery of all four parameters ($\theta, \gamma_0, \gamma_1, \eta_2$).
    
    \item \textbf{Network-Dependent Treatment:} 
    Consider settings where treatment assignment depends on the latent network (Figures~\hyperref[fig:DAGs]{1(a),(c)}). For example, the treatment assignment mechanism is $\Pr(Z_i=1 \mid \bA^\ast) \propto d_i^\ast$, 
    where $d_i^\ast$ is the degree of $i$ in the true network.
    The observed treatments act as an additional proxy source.
    Treatments are independent only conditional on the latent network,
    thus $\bbE[Z_i Z_j]$ provides an additional equation involving $\theta$ and $p_z$.
    Also, the covariance $\text{Cov}(Z_i, d_i^{obs})$ provides an additional
     equation linking $\theta$ and the proxy error rates.
     Thus, we have over-identification of the structural parameters.
     % Global identifiability is thus achieved.

     \item \textbf{Network-Correlated Outcomes:} In many applications,
      the outcome errors $\varepsilon_i$ are correlated among neighbors. 
      For example, $\text{Cov}(\varepsilon_i, \varepsilon_j) = \rho A^\ast_{ij}$.
       While this introduces a new parameter $\rho$,
    it allows us to leverage the second moments of the outcome.
     Specifically, the conditional expectations $\bbE[Y_i Y_j \mid A_{ij}=k]$ for $k \in \{0,1\}$ provides two distinct equations that allow for the unique identification of the full parameter set, including $\rho$.
\end{enumerate}
These results highlight a fundamental insight:
 while the latent nature of $\bA^\ast$ poses a challenge, 
the parameters $\bTheta$ are recoverable provided there 
 is sufficient information from the outcomes, treatments, and proxy networks.
Furthermore, under \eqref{eq:simple_model}, had we ignored the outcomes and treatments 
and tried to identify $(\gamma_0, \gamma_1, \theta)$ from the proxy networks $\cA$ alone,
we would have needed at least three independent proxy networks to achieve identifiability \citep{chang2022estimation}.
That illustrates that in our framework, each data module (e.g., outcomes, treatments, proxies) contributes to identifying the structural parameters.
That is, $\theta$ and $\gamma$ are recovered not just using the proxies $\cA$, but also through information contained in the outcomes $\bY$ and possibly the treatments $\bZ$ modules.

\section{Bayesian Estimation and Inference}
\label{sec:inference}

We propose a Bayesian framework for estimation.
We describe here inference in the scenario without latent variables $\bU$ and with causal proxies $\cA$ (Figures~\hyperref[fig:DAGs]{1(a)-(b)}). 
The details for the scenarios involving latent variables $\bU$ or 
non-causal proxies (Figures~\hyperref[fig:DAGs]{1(c)-(d)}) are 
given in Appendix \ref{apdx.sec:posterior}.

We first consider the setting of treatment assignment based on the latent network  (Figure~\hyperref[fig:DAGs]{1(a)}).
% Because all estimands of interest are conditional on the observed covariates $\bX$, we replace the covariates distribution $p(\bX)$ with a point mass. 
We assume prior independence $p(\boeta, \bbeta, \btheta,\bgamma)=p(\boeta)p(\bbeta)p(\btheta)p(\bgamma)$.
Recall that the observed data are $\bD = (\bY,\bZ,\cA,\bX)$. 
The latent variables are the true network $\bA^\ast$ and 
the parameters $\bTheta = (\boeta, \bbeta, \btheta,\bgamma)$. 
The factorization \eqref{eq:data_factorization} and prior independence yield the joint posterior
\begin{equation}
\label{eq:observed.posterior}
    \begin{aligned}
  p(\bTheta,\bA^\ast \mid \bD) 
    \propto
    &\; p(\bY \mid \bZ,\bA^\ast,\bX,\boeta) p(\boeta) 
    \\ & \times
    p(\bZ \mid \bA^\ast, \bX,\bbeta) p(\bbeta)
    \\ & \times
     p(\cA \mid \bA^\ast, \bX,\bgamma) p(\bgamma)
    \\ & \times
     p(\bA^\ast \mid \bX,\btheta) p(\btheta).
    \end{aligned}
\end{equation}
% Due to prior independence, the treatment assignment model $p(\bZ \mid \bA^\ast, \bX, \bbeta)$ provides information only about $\bbeta$ and the latent network $\bA^\ast$.
% That is, the treatment model does not inform the outcome model parameters $\boeta$ directly,
% but only indirectly through $\bA^\ast$.
When treatment assignment is based on the proxy networks (Figure~\hyperref[fig:DAGs]{1(b)}), 
 the treatment model becomes $p(\bZ \mid \cA, \bX, \bbeta)$. 
 Consequently, by replacing $p(\bZ \mid \bA^\ast, \bX,\bbeta)$ with $p(\bZ \mid \cA, \bX, \bbeta)$ in \eqref{eq:observed.posterior}, we can marginalize 
the posterior $p(\bTheta,\bA^\ast \mid \bD)$  over $\bbeta$, yielding
\begin{equation}
\label{eq:observed.posterior.obs}
\begin{aligned}
    p(\boeta, \btheta,\bgamma,\bA^\ast \mid \bD) 
      \propto &\; 
      \int_{\bbeta} p(\bTheta,\bA^\ast \mid \bD) d\bbeta
    \\ \propto &\; 
    p(\bY \mid \bZ,\bA^\ast,\bX,\boeta) p(\boeta) 
    \\ & \times
     p(\cA \mid \bA^\ast, \bX,\bgamma) p(\bgamma)
    \\ & \times
     p(\bA^\ast \mid \bX,\btheta) p(\btheta).
\end{aligned}
\end{equation}
In either scenario \eqref{eq:observed.posterior} or \eqref{eq:observed.posterior.obs}, augmenting the outcome model with propensity scores as additional covariates is possible but often requires the cutting of `model feedback' \citep{zigler2013}.
We describe such extensions in Appendix \ref{apdx.sec:posterior}.
% Additionally, if latent network confounding, represented by latent variables $\bU$ (Section \ref{subsec:struc}), is present, it can be modeled and adjusted for \citep{McFowland2021, um2024bayesian}. We provide details for these extensions in Appendix \ref{apdx.sec:posterior}.

Sampling from the joint posterior
\eqref{eq:observed.posterior} or \eqref{eq:observed.posterior.obs} 
is challenging due to its mixed space, 
which includes continuous parameters $\bTheta$ and 
discrete latent variables $\bA^\ast$. The discrete space grows super-exponentially with $N$,
 containing $\mathcal{O}(2^{N^2})$ terms.
A common approach for mixed-space posteriors is to marginalize discrete variables \citep{stan_latent_discrete}.
However, in our setting, the outcome model for each unit depends on $\bA^\ast$ through its entire neighborhood $\bA^\ast_i$ (and possibly beyond), 
inducing complex interdependencies that preclude obtaining a closed-form 
marginalized posterior (indeed, even for the simple linear model in \eqref{eq:simple_model}, the marginal likelihood is analytically intractable).
Brute-force marginalization of $\bA^\ast$ requires summing over $\mathcal{O}\big(2^{N^2}\big)$ terms, making it computationally infeasible. 

Given this intractability, we explored various alternatives. 
One approach is to relax the discrete space into a continuous one using 
reparameterization techniques 
\citep{zhang_2012, maddison2017concrete, nishimura_discontinuous_2020}.
However, we found that in our setting, such continuous relaxation methods
introduce approximation errors that result in unsatisfactory performance (Section~\ref{sec:sim}).
Another option is to use Mixed Hamiltonian Monte Carlo (HMC), which explicitly accounts for the discrete space \citep{zhou_mixed_2020}. We found that Mixed-HMC does not scale well in terms of computation time, even for small $N$. These limitations highlight the need for more efficient sampling strategies, which we explore next.

To this end, we propose a Block Gibbs MCMC algorithm that decomposes the sampling process into two blocks
consisting of continuous parameters $\bTheta$ and discrete network $\bA^\ast$.
Given the network state $\bA^\ast$, updating the continuous parameters $\bTheta$ is straightforward using standard MCMC techniques.
The challenge lies in updating the latent network $\bA^\ast$.
Due to its large discrete spaces, 
updating $\bA^\ast$ entries with a Metropolis-Hastings algorithm using a random walk Markov kernel can lead to poor mixing and slow convergence. 
 Therefore, we use a variant of Locally Informed Proposals \citep[LIP;][]{zanella2019informed}. 
 LIP uses information from the conditional posterior of $\bA^\ast$ to explore possible updates efficiently and overcome the poor mixing of standard random walk proposals in high-dimensional discrete spaces.
We now provide details about LIP in our settings and then describe the complete Block Gibbs algorithm.

\subsection{Locally Informed Proposals}
\label{subsec:lip}
 
In our case of a latent network with binary edges, each update of the current $\bA^\ast$ state 
corresponds to `flipping' some entries ($A^\ast_{ij}=1 \to A^\ast_{ij}=0$ or 
$A^\ast_{ij}=0 \to A^\ast_{ij}=1$). 
A standard random walk kernel proposes entries to be flipped at random.
In high-dimensional spaces, such uninformed proposals may lead to states with low posterior probability,
resulting in high rejection rates and slow mixing. To address this, we adopt the LIP framework, which constructs a proposal distribution that biases flips towards states with higher posterior probability. Theoretically, LIPs are designed to be locally reversible with respect to the target distribution,
 a property shown to improve the asymptotic efficiency and mixing times of the Markov chain \citep{zanella2019informed}.

 However, implementing exact LIP requires evaluating the posterior probability 
 for every possible proposal of a single-edge flip. For a network of size $N$, this entails $\mathcal{O}(N^2)$ posterior evaluations per iteration,
which can be computationally prohibitive. 
To render LIP feasible, we leverage the method proposed by \citet{grathwohl_2021},
which approximates the differences using gradients.
This reduces the computational cost of generating a proposal to a single gradient evaluation per iteration.
% Moreover, given that the posterior differences are Lipschitz continuous,
% the approximation error can be shown to not severely impact the convergence properties of the Markov chain \citep{grathwohl_2021}.

Consider the joint posterior \eqref{eq:observed.posterior}.
Given the other unknowns $\bTheta$ and the data $\bD$,
the posterior of $\bA^\ast$ is proportional to
\begin{equation}
    \label{eq:a_star_post}
    % p(\bA^\ast \mid \cdot) \equiv
     p(\bA^\ast \mid \bD, \bTheta) \propto 
    %  p(\bY \mid \bZ, \bA^\ast, \bX, \boeta) 
    %  p(\bZ \mid \bA^\ast, \bX, \bbeta)
    %  p(\cA \mid \bA^\ast, \bX, \bgamma)
    %  p(\bA^\ast \mid \bX, \btheta).
    \underbrace{p(\bY \mid \bZ, \bA^\ast, \bX, \boeta)}_{\text{Outcomes}} \times
     \underbrace{p(\bZ \mid \bA^\ast, \bX, \bbeta)}_{\text{Treatments}} \times
     \underbrace{p(\cA \mid \bA^\ast, \bX, \bgamma)}_{\text{Proxies}} \times
     \underbrace{p(\bA^\ast \mid \bX, \btheta)}_{\text{Prior}}.
\end{equation}
Equation~\eqref{eq:a_star_post} illustrates that the conditional posterior of $\bA^\ast$ is composed of the likelihood terms from all data modules dependent on $\bA^\ast$ (outcomes, treatments, and proxies in this scenario), alongside the prior network model. Since the LIP algorithm utilizes this conditional posterior to propose updates for the Markov chain, the reconstruction of the latent network is not driven solely by the observed proxies $\cA$. Instead, the inference actively leverages the signal embedded in the outcomes and, where applicable, treatments. 
In the specific case of posterior \eqref{eq:observed.posterior.obs}, treatment allocation depends solely on $\cA$ and $\bX$. Consequently, the treatment assignment model is constant with respect to $\bA^\ast$, and the conditional posterior $p(\bA^\ast \mid \bD, \bTheta)$ is driven only by the outcomes, proxies, and the prior network model.

% Equation~\eqref{eq:a_star_post} serves as the computational realization of the 
% \textit{information triangulation} principle established in Section~\ref{sec:identifiability}.
% While Section~\ref{sec:identifiability} relied on these distinct data views to prove parameter
%  recoverability, the LIP sampler actively leverages this structure for efficient exploration.
% By conditioning on all data modules simultaneously, the update step for $\bA^\ast$ is 
% constrained not just by the observed proxies $\cA$, but also by the signal embedded in 
% the outcomes $\bY$ and treatments $\bZ$.
% This ensures that the proposed network updates are made using all available sources 
% of information, effectively ``triangulating'' the true latent structure within the Bayesian framework.

Let $p(\bA^\ast \mid \cdot) \equiv
     p(\bA^\ast \mid \bD, \bTheta)$ denote the conditional posterior in \eqref{eq:a_star_post}.
Since the interference network is undirected, we only need to update the
 upper (or lower) triangle values of $\bA^\ast$ and set the lower (or upper) triangle values to be the same.
The LIP for updating network state from $\bA^\ast_t$ to $\bA^\ast_{t+1}$ have the structure \citep{zanella2019informed}
\begin{equation}
    \label{eq:local_balanced_ip}
    Q(\bA^\ast_{t+1} \mid \bA^\ast_t, \cdot) \propto g\bigg(\frac{p(\bA^\ast_{t+1} \mid \cdot)}{p(\bA^\ast_{t} \mid  \cdot)}  \bigg) \mathbb{I}\big(\bA^\ast_{t+1} \in H(\bA^\ast_t)\big),
\end{equation}
where $H(\bA^\ast_t) = \big\{\bA^* : \sum_{i<j} \lvert A^\ast_{ij} - A^\ast_{t,ij} \rvert = 1 \big\}$ is 
the Hamming ball of size one, with $A^\ast_{t,ij}$ being the $ij$ entry of $\bA^\ast_t$, and $g$ is a function that must satisfy the balancing condition $g(a) = a g(1/a)$ for all $a>0$ for the proposal $Q(\cdot|\cdot)$ to be locally reversible with respect to $p(\bA^\ast \mid \cdot)$ \citep{zanella2019informed}.
Common choices include $g(a) = \sqrt{a}$ or $g(a) = \frac{a}{1+a}$.
The proposals in \eqref{eq:local_balanced_ip} use the posterior ratio to 
inform possible moves constrained to flipping only one of $\bA^\ast_t$ entries, and are typically followed by an accept/reject step. 

Let $\log p(\bA^\ast \mid \cdot)$ be the
 log of the conditional posterior in \eqref{eq:a_star_post} and denote the difference 
 between two subsequent states by 
 $\Delta(\bA^\ast_{t+1}, \bA^\ast_t \mid \cdot) = \log p(\bA^\ast_{t+1} \mid \cdot) - \log p(\bA^\ast_t \mid \cdot)$.
Assume that $g(a) = \sqrt{a}$. The proposals \eqref{eq:local_balanced_ip} can be written as 
\begin{equation}
    \label{eq:lip_logpost}
     Q(\bA^\ast_{t+1} \mid \bA^\ast_t, \cdot) \propto \exp\Big(\frac{1}{2}\Delta (\bA^\ast_{t+1}, \bA^\ast_t \mid \cdot)\Big) \mathbb{I}\big(\bA^\ast_{t+1} \in H(\bA^\ast_t)\big).
\end{equation}
Let $\bA^\ast_{t+1}(ij)$ be $\bA^\ast_t$ with only entry $A^\ast_{t,ij}$ flipped.
 With a slight abuse of notation, denote 
 $\Delta(ij,t \mid \cdot) \equiv \Delta (\bA^\ast_{t+1}(ij), \bA^\ast_t \mid \cdot)$. 
 Since $\bA^\ast_{t+1}(ij) \in H(\bA^\ast_t)$, proposing an update using \eqref{eq:lip_logpost} 
 is equivalent to sampling an entry with probability 
 $q(ij \mid \bA^\ast_t ,\cdot) \equiv \text{Softmax} \Big(\frac{1}{2}\Delta(ij,t \mid \cdot)\Big)$ 
 and flip its value. 
 However, that entails computing $\log p(\bA^\ast_{t+1}(ij) \mid \cdot)$ 
 for all $1\leq i < j \leq N$ which entails $\mathcal{O}(N^2)$ evaluations of the log-posterior. 

To overcome this computational challenge, 
we approximate the differences $\Delta(ij,t \mid \cdot)$ using gradients of the log-posterior.
The key insight is that although $\bA^\ast$ is discrete, the structure of the log-posterior
 often allows for accurate gradient-based approximations \citep{grathwohl_2021}.
This structure enables us to approximate the differences $\Delta$ with a first-order Taylor series
\begin{equation}
    \label{eq:grad_approx_diff_single}
    \widetilde{\Delta}(ij,t \mid \cdot) = -(2A^\ast_{t,ij} -1) \frac{\partial \log p(\bA^\ast_t \mid \cdot)}{\partial A^\ast_{t,ij}} \approx \Delta(ij,t \mid \cdot).
\end{equation}
% Thus, the approximation is $\widetilde{\Delta}(ij,t \mid \cdot) \approx \Delta(ij,t \mid \cdot)$. 
In vector form, we can compute the gradients $\nabla\log p(\bA^\ast_t \mid \cdot)$ with respect to entries in the upper triangle of $\bA^\ast$ and approximate all the $N(N-1)/2$ differences $\Delta(ij,t \mid \cdot)$ through
\begin{equation}
    \label{eq:grad_approx_diff}
    \boldsymbol{\widetilde{\Delta}}(\bA^\ast_t \mid \cdot) = -(2  \text{vec}(\bA^\ast_t) -1) \odot \nabla \log p(\bA^\ast_t \mid \cdot).
\end{equation}
where $\odot$ is the element-wise product and $\text{vec}(\bA^\ast)$ is the vector of $N(N-1)/2$ upper triangle values of $\bA^\ast$.
Thus, 
$ \boldsymbol{\widetilde{\Delta}}(\bA^\ast_t \mid \cdot) = \big\{\widetilde{\Delta}(ij,t): 1 \leq i < j \leq N \big\}$ is the set of all approximate differences \eqref{eq:grad_approx_diff_single}.
Computing \eqref{eq:grad_approx_diff} requires only a single gradient evaluation and can 
be performed with any automatic differentiation library. 
We formally analyze the approximation error in Appendix~\ref{apdx.sec:gradient_approx_error}. We demonstrate that for an extended version of the simplified model \eqref{eq:simple_model}, the efficiency of the gradient-based approximation relative to the exact proposal decreases as the interference signal-to-noise ratio increases, but, crucially, remains independent of the network size $N$. Furthermore, this approximation was accurate in the numerical illustrations (Appendix~\ref{apdx.sec:sim}).

The pseudocode for a single update of $\bA^\ast$ state with LIP is given in Algorithm~\ref{algo:lip}.
The pseudocode is applicable for both posterior \eqref{eq:observed.posterior} and \eqref{eq:observed.posterior.obs}
by replacing $\log p(\cdot)$ with the corresponding log-posterior.
We extend the single-entry update by proposing $K \geq 1$ flips per iteration. 
Instead of selecting a single entry, we sample $K$ entries without replacement 
using the Gumbel-Max trick (Step 3). 
The selected indices are saved in a set $\mathcal{I}$. 
The function \texttt{FlipEntries} simply takes a network state $\bA^\ast_t$ 
and flips the values of all entries in the set of indices $\mathcal{I}$.
For $K=1$, the algorithm performs \emph{exact} LIP but is likely to require more iterations, 
while for $K>1$, the updates become \emph{approximate} but can accelerate convergence. 
Increasing $K$ poses a trade-off between accuracy and computational efficiency. 
In our numerical illustrations (Section \ref{sec:sim}), we varied $K$ depending 
on the network model and sample size.

\begin{center}
    % \begin{minipage}{0.8\textwidth}
        \begin{algorithm}[ht]
        \captionsetup{width=1\textwidth}  
        \caption{A Single  $\bA^\ast$ Update with Locally Informed Proposals}
        \label{algo:lip}
        % \begin{center}  
            % \begin{minipage}{0.8\textwidth} 
        \begin{algorithmic}[1]
        \Statex \textbf{Input:} Data $\bD$, continuous parameters $\bTheta$, current state $\bA^\ast_t$, log-posterior $\log p(\cdot)$,  number of entries to update $K\geq 1$.
        \State Compute differences $\boldsymbol{\widetilde{\Delta}}(\bA^\ast_t \mid \cdot) \equiv \boldsymbol{\widetilde{\Delta}}(\bA^\ast_t \mid \bD, \bTheta)$ using \eqref{eq:grad_approx_diff}.
        \State Compute $q(ij \mid \bA^\ast_t, \cdot) = \text{Softmax} \Big(\frac{1}{2}\boldsymbol{\widetilde{\Delta}}(\bA^\ast_t \mid \cdot) \Big)$.
        \State Sample without replacement $K$ edges $\mathcal{I}$ using Gumbel-Max trick with probabilities $q(ij \mid \bA^\ast_t, \cdot)$.
        \State Compute \emph{forward} probability $q(\mathcal{I} \mid \bA^\ast_t, \cdot) = \prod_{ij \in \mathcal{I}}q(ij \mid \bA^\ast_t, \cdot)$.
        \State $\bA^\ast_{t+1} \gets \texttt{FlipEntries}(\bA^\ast_t, \mathcal{I})$.
        \State Compute \emph{backward} probability $q(\mathcal{I} \mid \bA^\ast_{t+1}, \cdot) = \prod_{ij \in \mathcal{I}}q(ij \mid \bA^\ast_{t+1}, \cdot)$ similarly to step 2.
        \State Accept with probability
        \[
        \min \left( \exp\big(\Delta(\bA^\ast_{t+1}, \bA^\ast_t \mid \cdot)\big) \frac{q(\mathcal{I} \mid \bA^\ast_{t+1}, \cdot)}{q(\mathcal{I} \mid \bA^\ast_t, \cdot)}, 1 \right).
        \]
        \end{algorithmic}
        \end{algorithm}
    % \end{minipage}
\end{center}
 
\subsection{Block Gibbs Algorithm}
\label{subsec:block_gibbs}
We propose a Block Gibbs approach that alternates between updating $\bA^\ast$
 with LIP and the continuous parameters $\bTheta$ using a continuous kernel
  $Q_c(\bTheta \mid \bA^\ast, \bD)$, 
  such as HMC \citep{neal2011mcmc} or No-U-Turn-Sampler \citep[NUTS,][]{hoffman2014no}.
Algorithm~\ref{algo:block_gibbs} outlines the procedure.
% We describe it for the posterior \eqref{eq:observed.posterior.obs}, but analogous definitions apply for \eqref{eq:observed.posterior}.
In each iteration, we first perform $L$ updates of $\bA^\ast$, for 
a maximum of $L \times K$ entry flips per iteration. 
Given its new network state $\bA^\ast_{t+1}$, we perform a single update 
of the continuous parameters. The latter can also consist of multiple steps, 
for example, if an HMC with multiple integration steps is used \citep{neal2011mcmc}. 
See \citet{betancourt2018} for more details regarding practical applications of HMC.
\begin{center}
    % \begin{minipage}{0.8\textwidth}
        \begin{algorithm}[ht]
        \captionsetup{width=1\textwidth}  
        \caption{Block Gibbs Posterior Sampling}
        \label{algo:block_gibbs}
        % \begin{center}  
            % \begin{minipage}{0.8\textwidth} 
        \begin{algorithmic}[1]
        \Statex \textbf{Input:} Data $\bD$, Continuous kernel $Q_c$, total steps $T$, LIP steps $L$, entries per update $K$.
        \State Initialize $\bA^\ast_0, \bTheta_0$. 
        \For{$t=0$ to $T-1$}
            \State $\bA^\ast_{t+1} \gets \bA^\ast_t$.
                % \Comment{Initialize for LIP steps}
            \For{$\ell = 1$ to $L$}
            \Comment{LIP updates using Algorithm~\ref{algo:lip}}
                \State $\bA^\ast_{t+1} \gets \text{LIP}(\bA^\ast_{t+1}, K \mid \bTheta_t, \bD)$. 
            \EndFor
            \State $\bTheta_{t+1} \gets Q_c(\bTheta_{t} \mid \bA^\ast_{t+1}, \bD).$
        \EndFor
        \State \textbf{Return:} $\Big\{\bA^\ast_t, \bTheta_t\Big\}_{t=1}^{T}$
        \end{algorithmic}
        \end{algorithm}
    % \end{minipage}
\end{center}
There is an inherent trade-off in choosing $L$ and $K$.
Larger values allow for more extensive exploration of the network space per iteration,
potentially improving mixing. However, they also increase the computational cost per iteration.
Furthermore, performing excessive updates of $\bA^\ast$ before updating $\bTheta$ may lead to 
drift in the Markov chains of the continuous parameters,
as changing the network structure can significantly alter the geometry of the conditional posterior of $\bTheta$.

We found that this trade-off can be mitigated by a robust initialization strategy.
Specifically, we employ the following multi-stage procedure grounded in
 Bayesian modularization \citep{Bayarri2009, Jacob2017} to ensure the chain starts
in a high-probability region:
\begin{enumerate}
    \item Estimate the parameters governing the 
    network and proxy models, $(\btheta, \bgamma)$, 
    by sampling from the marginalized network module 
    $p(\btheta,\bgamma \mid \cA,\bX) \propto p(\btheta)p(\bgamma)\sum_{\bA^\ast} p(\cA \mid \bA^\ast,\bX,\bgamma) p(\bA^\ast \mid \bX,\btheta)$.
     This approach mirrors latent network reconstruction methods using proxies
      \citep{Young2021}. We denote the resulting estimates (e.g., posterior means) 
      by $(\hat{\btheta}, \hat{\bgamma})$.
    \item Conditionally on $(\hat{\btheta}, \hat{\bgamma})$, sample multiple networks from
     $p(\bA^\ast \mid \cA, \bX, \hat{\btheta}, \hat{\bgamma})$ 
    and select $\bA^\ast_0$ as the network with the highest conditional probability.
    \item Estimate the outcome model parameters $\boeta$ 
    (and $\bbeta$ if applicable) using $\bA^\ast_0$ as a fixed covariate,
    yielding initial estimates $\bTheta_0$.
    \item Finally, refine $\bA^\ast_0$ by performing a short sequence
     of LIP updates (Algorithm \ref{algo:lip}) conditional on the initialized full continuous parameters $\bTheta_0$.
\end{enumerate}
Steps 1--3 provide initial estimates using the cut-posterior distribution, 
which removes the `feedback' between the network and outcome modules. 
The final LIP refinement step modifies $\bA^\ast_0$ using information from the 
outcome and treatment models, effectively moving the initial state towards the
 higher density region of the full posterior. 
 The resulting values are used to initialize Algorithm~\ref{algo:block_gibbs}. 
 Further technical details regarding this procedure are provided in Appendix \ref{apdx.sec:posterior}.
% We found that this trade-off can be mitigated by a robust initialization strategy, especially of $\bA^\ast$.
% We found that sampling from the cut-posterior \citep{Bayarri2009, Jacob2017},
%  which removes the `feedback' between the networks and outcome models,
%  and then further refining $\bA^\ast$ via several LIP steps (Algorithm~\ref{algo:lip}) 
%  yielded good initial values and worked well in practice. 
%  This procedure is  
%  We describe the details of this initialization strategy in Appendix \ref{apdx.sec:posterior}.

% The causal estimands, e.g., $\tau(\bz,\bz')$ (Section \ref{subsec:estimands}) can be expressed as functionals of $\boeta$ and $\bA^\ast$. Therefore, given $T$ posterior samples from Algorithm~\ref{algo:block_gibbs}, we can draw from the posterior predictive distribution of the outcomes to obtain $T$ estimates of each estimand of interest. These estimates approximate the posterior distribution of each estimand, accounting for uncertainty in both the interference structure $\bA^\ast$ and the parameters governing the outcome model.   

\subsection{Causal Effects Estimation}
\label{subsec:bayes_g_formula}
We describe how to estimate the causal estimands (Section~\ref{subsec:estimands}) given $T$ posterior samples. For simplicity of presentation, we focus on the static policy. Following our Bayesian framework, let
$\mathbb{E}\big[Y_i(\bz) \mid \bD \big]$ be the expected outcome of unit $i$ under the intervention that sets $\bZ = \bz$, given the observed data $\bD$. Under the SCMs described in Figures~\hyperref[fig:DAGs]{1(a)--(c)}, we have
\begin{equation*}
    \begin{split}
        \mathbb{E}\big[Y_i(\bz) \mid \bD \big] 
        &=
    \mathbb{E}_{\bA^\ast, \boeta}\Big[\mathbb{E}\big[Y_i(\bz) \mid \bD,  \bA^\ast, \boeta\big] \mid \bD \Big] 
        \\ &=
    \mathbb{E}_{\bA^\ast, \boeta}\Big[\mathbb{E}\big[Y_i \mid \bZ=\bz, \bX, \bA^\ast, \boeta\big] \mid \bD \Big],
    % \\ 
    % &\approx
    % T^{-1} \sum_{t=1}^{T}\mathbb{E}\big[Y_i \mid \bZ=\bz, \bX, \bA^\ast_t, \boeta_t \big],      
    \end{split}
\end{equation*}
% .
where the second equality follows since all back-door paths between $\bZ$ and $\bY$ are blocked (Section~\ref{sec:settings}).
Under Figure~\hyperref[fig:DAGs]{1(d)} we have to condition on the proxies $\cA$ as well, as described in Section~\ref{subsec:estimands}.
With a slight abuse of notation let $\widehat{\mu}_i(\bz \mid \bD)$ be the posterior expectation of $\mathbb{E}\big[Y_i(\bz) \mid \bD \big]$.  We can approximate this expectation using the posterior samples to obtain an estimator
\begin{equation}
    \label{eq:mu_i_est}
    \begin{split}
         \widehat{\mu}_i(\bz \mid \bD)
         \approx 
         T^{-1} \sum_{t=1}^{T}\mathbb{E}\big[Y_i \mid \bZ=\bz, \bX, \bA^\ast_t, \boeta_t \big],
    \end{split}
\end{equation}
where $\bA^\ast_t,\boeta_t$, $t=1,\ldots,T$ are the samples from the the posterior distribution $p(\bA^\ast,\boeta \mid \bD)$.
If the inner conditional expectation $\mathbb{E}\big[Y_i \mid \bZ=\bz, \bX, \bA^\ast_t, \boeta_t \big]$ does not have a closed-form expression, we can approximate it via Monte Carlo by drawing $Y_i$ from its posterior predictive distribution (Appendix~\ref{apdx.sec:posterior}). 
Therefore, the point estimate of a static policy for unit $i$ is $\widehat{\tau}_i(\bz,\bz' \mid \bD) = \widehat{\mu}_i(\bz \mid \bD) - \widehat{\mu}_i(\bz' \mid \bD)$, and the population-level point estimate is $\widehat{\tau}(\bz,\bz' \mid \bD) = N^{-1} \sum_i \widehat{\tau}_i(\bz,\bz' \mid \bD)$. In addition, credible intervals can be computed using the percentiles of
\begin{equation*}
    \widehat{\tau}^{(t)}(\bz,\bz'\mid \bD) =
    N^{-1} \sum_i \left(\mathbb{E}\big[Y_i \mid \bZ=\bz, \bX, \bA^\ast_t, \boeta_t \big] - \mathbb{E}\big[Y_i \mid \bZ=\bz', \bX, \bA^\ast_t, \boeta_t \big] \right),
\end{equation*}
over the $t=1,\ldots,T$ posterior samples.
Note that we can also write $\widehat{\tau}(\bz,\bz' \mid \bD) = T^{-1}\sum_{t=1}^{T} \widehat{\tau}^{(t)}(\bz,\bz'\mid \bD)$, i.e., as the posterior mean.
We can similarly obtain estimates and intervals for dynamic and stochastic estimands (Section~\ref{subsec:estimands}).

% Bayesian modularization offers several advantages \citep{Bayarri2009}. Notably, it severely reduces, and often eliminates, contamination between modules due to misspecification in some modules.
% This approach also allows for separate model checking and selection of the network models (true and observed) from that of the outcome model. Researchers can first perform model diagnostics on the network modules and subsequently repeat the process for the outcome model, effectively separating the diagnostic process for the interference network and the outcome model.
% Moreover, by breaking the complex full posterior into manageable parts, modularization provides a computationally tractable approach where direct sampling would be challenging.

\section{Numerical Illustrations}
\label{sec:sim}

We conducted numerical experiments using fully- and semi-synthetic data to evaluate the performance 
of our proposed methods. In the fully-synthetic scenario, all data components are simulated.
In contrast, the semi-synthetic scenario combines real network data with simulated treatments and outcomes. 
We compare our proposed estimator based on the Block Gibbs sampler against several baselines.
Key details and main results are presented below, while specific parameter settings and additional results are provided in Appendix \ref{apdx.sec:sim}.

In both fully- and semi-synthetic scenarios, we assigned treatments independently to each unit with 
    $\Pr(Z_i=1 \mid \bA^\ast) \propto d^\ast_i$, namely, treatment assignment was proportional to the degree of $i$ in the true interference network.
Outcomes were generated from a linear model with $\bY = \boldsymbol{\mu} + \varepsilon_i$,
where $\varepsilon_i \sim N(0,\sigma^2)$ are independent errors.
% \begin{equation*}
%     \bY = \boldsymbol{\mu} + \bpsi. 
% \end{equation*}
The mean component $\boldsymbol{\mu}$ for unit $i$ was $\mu_i =
     \eta_1Z_i + \eta_2 \phi_1(\bZ_{-i},\bA^\ast) + \boeta_x'\bX_i,$
% \begin{equation*}
%     \mu_i =
%      \eta_1Z_i + \eta_2 \phi_1(\bZ_{-i},\bA^\ast) + \boeta_x'\bX_i,
% \end{equation*}
and the exposure mapping $\phi_1(\mathbf{Z}_{-i},\mathbf{A}^\ast)$ was set to summarize neighbors' treatments through a weighted average  
% \begin{equation*}
    $
    \phi_1(\bZ_{-i},\bA^\ast) = \sum_{j \neq i}w_jZ_j A^\ast_{ij}.
    $
% \end{equation*}
The weights were taken to be the degree centrality, $w_j =\frac{d^\ast_j}{N-1}$, assigning greater influence to units with higher network centrality.
Here, $\eta_1$ represents the direct treatment effect, while $\eta_2$ captures the interference effect
 through neighbors' treatment. 
%  The vector $\boldsymbol{\psi} = (\psi_1, \ldots, \psi_N)$ represents network random effects and follows a Conditional Autoregressive distribution 
% \citep{banerjee2003hierarchical}
% \begin{equation*}
    % \label{eq:car_model}
    % $
    % \boldsymbol{\psi} \sim N\big(\boldsymbol{0},\sigma(\boldsymbol{D}-\rho\bA^\ast)^{-1}\big),
    % $
% \end{equation*}
% where $\boldsymbol{D}$ is the diagonal matrix of node degrees, $\sigma>0$ is a precision parameter, and $\rho \in (0,1)$ controls the strength of network dependence. Thus, the interference network $\mathbf{A}^*$ influences the outcomes $\mathbf{Y}$ through their mean and covariance structures.

\subsection{Fully-Synthetic Data}
\label{subsec:sim_fully}

We generated fully-synthetic data for $N=500$ units. In each simulation iteration, we sampled data from the following SCM. 
Covariates $\bX_i = (X_{1,i}, X_{2,i})$, were simulated from $X_{1,i} \sim N(0,1)$ and $X_{2,i} \sim Ber(0.1)$. The edge-level covariates were defined as $\widetilde{X}_{1,ij}=\big\lvert X_{1,i} - X_{1,j} \big\rvert$, capturing the covariates distance, and $\widetilde{X}_{2,ij} = \mathbb{I}\{X_{2,i} + X_{2,j} = 1\}$, an indicator of exactly one of the two units having $X_2=1$. 
Edges in the latent interference network $\bA^\ast$ were generated independently according to
\begin{equation}
\label{eq:a_star_lsm}
    \Pr(A^\ast_{ij}=1 \mid \bX,\btheta) = 
    s\big(\theta_0 + \theta_1 \widetilde{X}_{1,ij} + \theta_2\widetilde{X}_{2,ij}),
\end{equation}
where $s(x) = \frac{1}{1 + e^{-x}}$ is the sigmoid function.
Parameters $\btheta$ were selected to produce sparse networks with a heavy-tailed degree distribution, reflecting common features of real-world networks (see Appendix~\ref{apdx.sec:sim} for details).
Given the true network $\bA^\ast$, we generated two independent proxy networks $\cA = \{\bA_1, \bA_2\}$
from a differential measurement error model
\begin{equation}
\label{eq:proxy_nets}
\begin{split}
    \Pr(A_{ij}=1 \mid A^\ast_{ij}, \bX, \bgamma) &= s\big(A^\ast_{ij}\gamma_0 + (1-A^\ast_{ij})(\gamma_1 + \gamma_2\widetilde{X}_{1,ij} + \gamma_3\widetilde{X}_{2,ij})\big). \\
    % \Pr(A_{1,ij}=1 \mid A^\ast_{ij}, \bX, \bgamma) &= s\big(A^\ast_{ij}\gamma_0 + (1-A^\ast_{ij})(\gamma_1 + \gamma_2\widetilde{X}_{1,ij} + \gamma_3\widetilde{X}_{2,ij})\big) \\
    % \Pr(A_{2,ij}=1 \mid A_{1,ij}, A^\ast_{ij}, \bgamma) &= s\big(A^\ast_{ij}(\gamma_3 + \gamma_4 A_{1,ij}) + (1-A^\ast_{ij})(\gamma_5 + \gamma_6 A_{1,ij})\big).
\end{split}
   \end{equation}
We varied $\bgamma$ values to control the (dis)similarity between proxies and the true network.
This model represents a scenario of causal proxies.
Specifically, false-positive edges ($A_{ij}=1, A^\ast_{ij}=0$) occurred with probabilities that depend on covariate similarities, while true edges ($A_{ij}^\ast=1$) were missing in $\bA_1$ completely at random. 
%  The second proxy $\bA_2$ is a noisy measure based on the true network and the first proxy.

We took $N(0,3^2)$ priors for all coefficients and $\sigma \sim \text{HalfNormal}(0,3^2)$ for the noise term.
In each iteration, we re-sampled all data from the SCM but fixed the true parameter values throughout.  
We compared two variants of the proposed Block Gibbs (BG) algorithm, 
using either a single ($\bA_1$) or both proxies ($\cA$), to several baselines.
These include a naive estimation approach that takes the first proxy $\bA_1$ (``Obs.") as 
the interference network and the oracle that uses the true network $\bA^\ast$ (``True"). 
Also, we tested the performance of continuous relaxation of $\bA^\ast$ edges into 
$[0,1]$ using the CONCRETE distribution \citep{maddison2017concrete} (``Cont. relax"), which is a differentiable relaxation of the discrete Gumbel-Max trick.
We also tested how using a two-stage approach that first estimates $\bA^\ast$ from the proxies and then estimates the outcome model while treating the estimated network as fixed.
In this approach, we further considered two methods to obtain the estimated fixed network:
 (i) using only the cut-posterior samples of $\bA^\ast$ (``Two-stage (cut)"),
  corresponding to steps 1--3 of the BG initialization (Section~\ref{subsec:block_gibbs}), and 
 (ii) refining the samples via several LIP steps (``Two-stage (LIP)"),
  which incorporates step 4 of the initialization strategy.
 We also assessed the contribution of the treatment signal to the inference by running a model without the treatment model (``BG (no Z)").
 The impact of model misspecification was also evaluated by omitting the covariate 
 $\bX_1$ from all models (``BG (misspec)"). 

The fixed-network methods (``Obs.", ``True", and two-stage methods) require
 only an outcome model and were estimated using NUTS \citep{hoffman2014no}.
In the BG algorithms, we used NUTS for the continuous kernel. 
Each BG iteration included $L=1$ LIP steps with $K=1$ edge updates per step. 
Initial values were obtained from the cut-posterior followed by LIP refinement (Appendix \ref{apdx.sec:posterior}).
The continuous relaxation method was estimated using Stochastic Variational Inference, 
see Appendix~\ref{apdx.sec:sim} for details. 
All MCMC methods used $T=1.2 \times 10^4$ posterior samples ($3,000$ samples across four chains).

Our analysis evaluated dynamic and static policies (Section~\ref{subsec:estimands}).  For static policies, we considered the TTE, $\tau(\boldsymbol{1}, \boldsymbol{0})$,
comparing the scenarios where all units are treated versus none.
We considered the dynamic policy  $ h_{c,i}(X_{1}) = \mathbb{I}\{X_{1,i} \not\in [-c,c]\}$,
assigning treatment only to units whose covariate exceeds a threshold $c$. 
The dynamic estimand was taken to be $\tau(h_{c_1},h_{c_2}) = \mu(h_{c_1}) - \mu(h_{c_2})$, with thresholds $c_1=0.75$ and $c_2=1.5$.

% The stochastic policy $\pi_\alpha(\bz)$ assigns treatments independently with probability $\alpha$. The stochastic estimand is $\tau(\alpha_1, \alpha_0)$, with probabilities $\alpha_1=0.7$ and $\alpha_0 = 0.3$. 
% The dynamic estimand is computed analytically, while the stochastic estimand is approximated by drawing $10^3$ treatment vectors from $\pi_{\alpha}$ for each $\alpha$.
Let $\widehat{\tau}_i$ be the point estimate for units $i$ of either the dynamic or static policies (Section~\ref{subsec:bayes_g_formula}).
Estimation accuracy was assessed by unit-level Mean Absolute Percentage Error (MAPE) between posterior point
 estimates and true estimands $\tau_i$,
% Let $\hat{\tau}_{i,t}$ be the estimate from the posterior predictive distribution of unit $i \in [N]$ in posterior sample $t=1,\ldots,T$.
% The point estimate is $\hat{\tau}=\frac{1}{NT}\sum_t\sum_i \hat{\tau}_{i,t}$ and the RE is 
defined by $\text{MAPE} = N^{-1} \sum_{i \in [N]} \big\lvert \frac{\hat{\tau}_i-\tau_i}{\tau_i}\big\rvert$.
Lower MAPE values indicate better estimation accuracy.
We also evaluated the mean relative error of the population-level estimands, defined by $\bbE\left(\left \vert \frac{\hat{\tau}-\tau}{\tau}\right \vert\right)$. These results are presented in Appendix~\ref{apdx.sec:sim} and exhibit performance patterns consistent with the unit-level MAPE.

Figure~\ref{fig:sim_rrmse} summarizes the $\text{MAPE} \; (\text{mean} \; \pm \;\text{standard deviation})$ values over $300$ iterations across
varying proxies strength ($\gamma_2$).
As proxies become weaker (larger $\gamma_2$), naively estimating $\tau$ while 
treating the observed proxies as the true network leads to higher estimation error.
In contrast, the BG algorithms consistently achieves accurate estimates.
Specifically, the BG with either a single or two proxies attains the lowest MAPE 
across all proxy strengths and estimands compared to all baselines other than the true oracle.
Ignoring the treatment model in the BG sampler (``BG (no Z)") slightly increases MAPE, illustrating the benefit of integrating all available data modules.
Moreover, model misspecification (``BG (misspec)") only mildly affected the MAPE.
The continuous relaxation approach also yielded high MAPE values.
The two-stage methods performed worse than all of the BG samplers.
However, the two-stage approaches with LIP refinement significantly outperformed the cut-posterior variant. 
Nevertheless, the two-stage methods (even with LIP refinement) tend to underestimate the uncertainty of the estimands, as reflected by empirical coverage rates of the $95\%$ credible intervals (CIs) 
 being below the nominal level (Appendix~\ref{apdx.sec:sim}).
In contrast, the BG algorithms achieved empirical coverage rates close to
 the nominal level across all proxy strengths,
demonstrating the necessity of the full iterative sampling procedure.

Figure~\ref{fig:networks_heatmap} illustrates how the posterior distribution of $\bA^\ast$ 
concentrates around the true network structure, by comparing the true network $\mathbf{A}^\ast$, 
the observed proxy network $\mathbf{A}_1$, and the posterior edge probabilities derived from the BG algorithm.
 The figure highlights that the posterior distribution effectively captures the true network topology even
  under weak proxy information ($\gamma_2=3$).
Additional results, including traceplots, mixing diagnostics, and scaling BG to larger networks, are provided in Appendix~\ref{apdx.sec:sim}.

\begin{figure}[!htbp]
    \centering
    \includegraphics[width=0.75\linewidth]{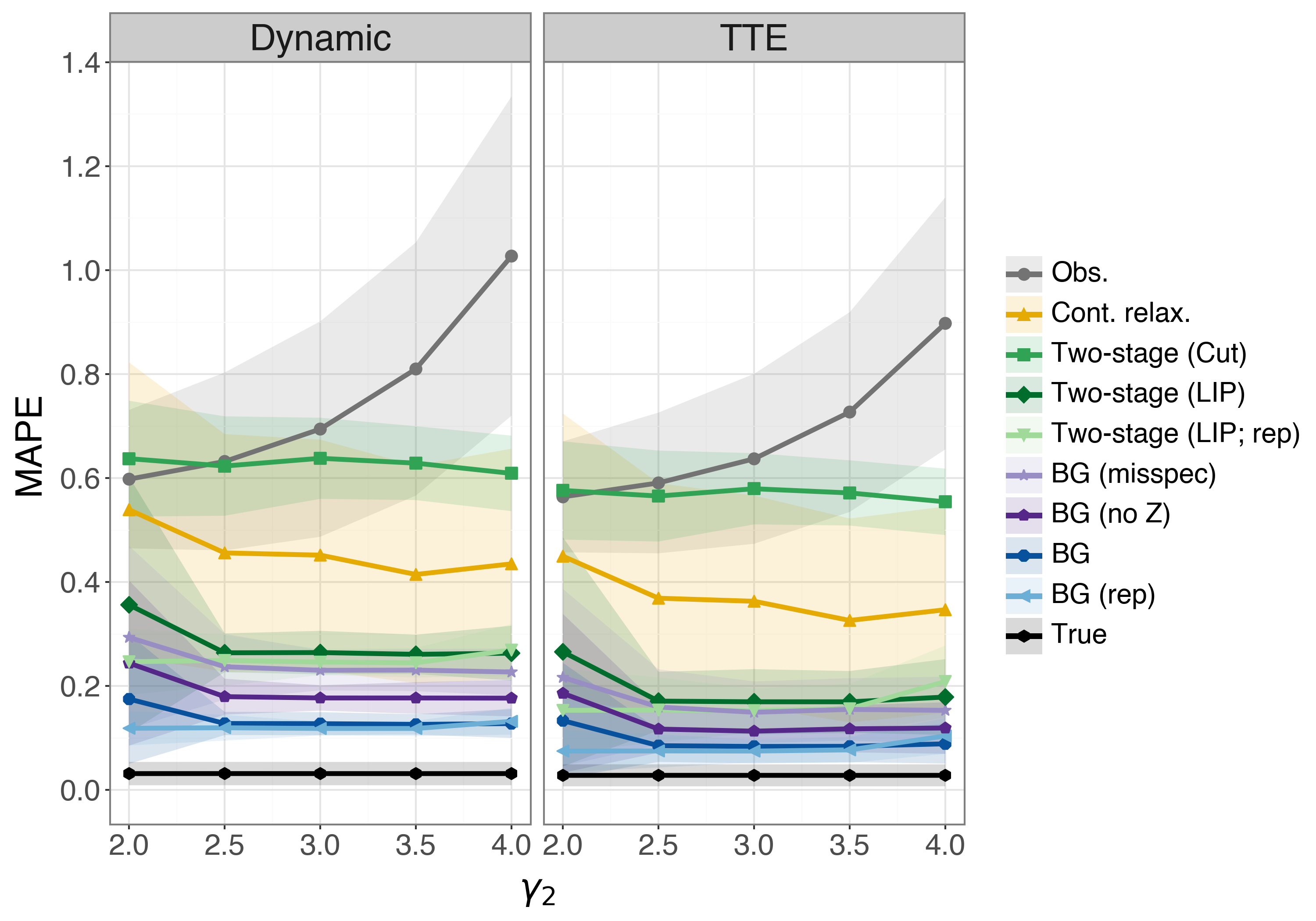}
    \caption{Mean ($\pm \; \text{SD}$) of MAPE values across $300$ simulation iterations for dynamic and TTE estimands. 
    Proxy networks $\cA$ are weaker as $\gamma_2$ (x-axis) increases.
    %  `True' when the true network $\bA^\ast$ is used, `BG' and `BG (rep)' for Block Gibbs sampler with either
    %   a single or two proxy networks, respectively. 
    %   `BG (no Z)' denotes the BG sampler without the treatment model, and `BG (misspec)' 
    %   the BG sampler with misspecified models omitting covariate $X_1$.
    %   `Two-stage (cut)' and `Two-stage (LIP)' represent the two-stage methods using either cut-posterior samples or LIP refinement, respectively.
    %   `Two-stage (LIP, rep)' uses two proxies in the two-stage method with LIP refinement.
    %     `Cont. relax' stands for the continuous relaxation method.
    %   `Obs.' stands for observed proxy $\bA_1$. 
      Smaller values indicate better performance.}
    \label{fig:sim_rrmse}
\end{figure}

\begin{figure}[!htbp]
    \centering
    \includegraphics[width=0.7\linewidth]{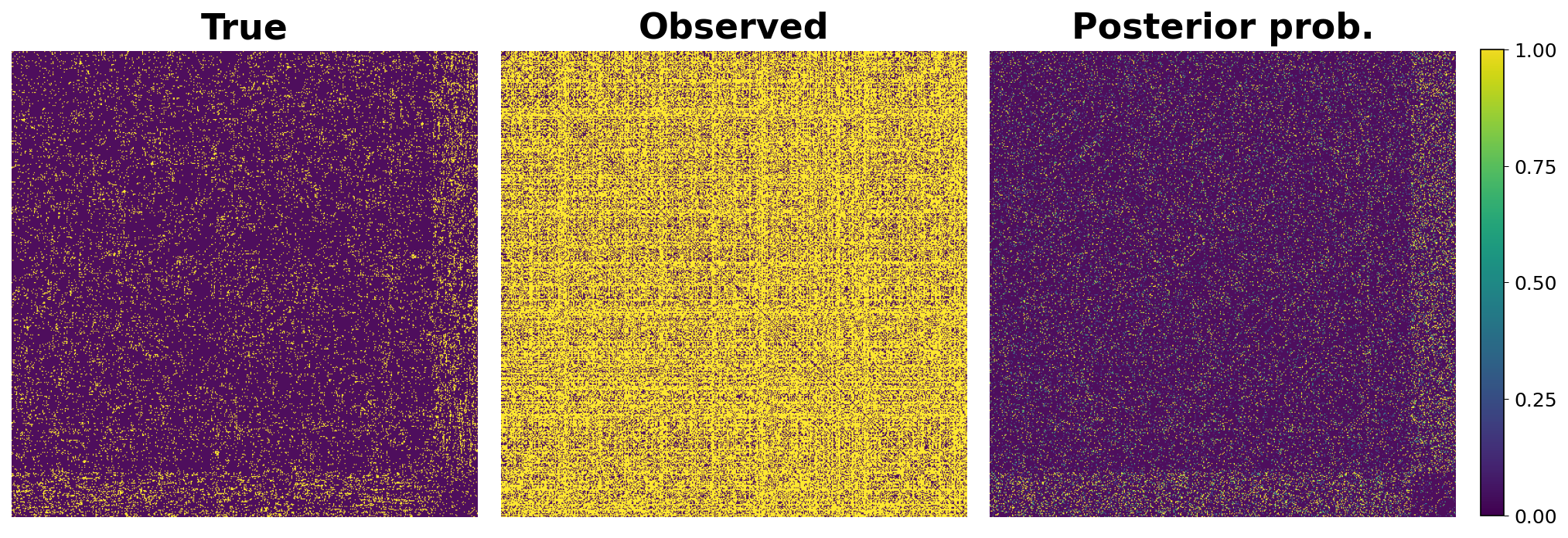}
    \caption{Heatmap of true $\bA^\ast$, observed proxy $\bA_1$, 
    and posterior probabilities under one iteration with $\gamma_2=3$. 
    In the True and Observed heatmaps, entries are binary. 
    The posterior probabilities are the proportion of times an edge existed in the posterior network samples using a Block Gibbs sampler with a single proxy network. Nodes were first rearranged by hierarchical clustering for clearer visualization.}
    \label{fig:networks_heatmap}
\end{figure}

\subsection{Semi-Synthetic Data}
\label{subsec:sim_semi}
In the second experiment, we analyzed a multilayer network dataset collected by \citet{magnani2013}, which originally contained five layers describing social connections among $N=61$ employees at Aarhus University's Department of Computer Science. Three layers are obtained through surveys asking employees about regular working relationships (``Work"), repeated leisure activities (``Leisure"), and eating lunch together (``Lunch"). Two additional layers captured online Facebook friendships and publication co-authorship. Since the co-authorship network is almost fully contained within other layers \citep{magnani2013}, it was excluded from our analysis, leaving us with four unweighted and undirected networks.
Summary statistics for each layer are presented in Table~\ref{tab:network_stats}.
\begin{table}[!htbp]
    \centering
    \begin{tabular}{lcccc}
        \hline
         & Facebook & Leisure & Lunch & Work \\
        \hline
        Number of edges & 124 & 88 & 193 & 194 \\
        Average degree & 7.75 & 3.74 & 6.43 & 6.47 \\
        Number of isolated units & 29 & 14 & 1 & 1 \\
        \hline
    \end{tabular}
    \caption{Summary statistics of the multilayer network. Each column represents a layer.}
    \label{tab:network_stats}
\end{table}

To illustrate the performance of our proposed method, we sequentially designated each layer as the latent interference network $\bA^\ast$, treating the remaining three layers as observed proxies $\cA$. 
The joint distribution of the four layers was modeled using LSM with shared bivariate latent positions $\bV$
\begin{equation*}
    \Pr(A_{b,ij}=1 \mid \bV, \bgamma) = s\left(\gamma_{b,0} - e^{\gamma_{b,1}} \lVert \bV_i-\bV_j \rVert_2 \right),
\end{equation*}
where $\gamma_{b,0}$ controls the baseline probability of edge creation in layer $b$ and $\gamma_{b,1}$ controls how latent distances affects edge probabilities. 
We assumed independent Gaussian distribution for the latent positions $\bV_i \sim N_2(\boldsymbol{0}, I)$, where $I$ is the identity matrix,  and hierarchical priors for the layer-specific parameters
\begin{equation*}
    \gamma_{b,j} \sim N(\mu_j, \sigma^2_j), \; \mu_j \sim N(0,3^2), \; \sigma_j \sim \text{HalfNormal}(0,3^2), \;j=0,1.
\end{equation*}
When a specific layer $b$ is treated as the latent network $\bA^\ast$, this model assumes that the three proxies indirectly inform its structure through the latent variables $\bV$, $\mu_j$, and $\sigma_j$. 
This setup represents non-causal proxies (Figure~\hyperref[fig:DAGs]{1(c)}), as none of the observed layers are assumed to directly influence the latent interference network 
and treatments depends only on $\bA^\ast$.
In each latent network scenario, treatments and outcomes were generated as in Section \ref{subsec:sim_fully}, although without covariates data. 
We focused on the TTE estimands, similar to Section~\ref{subsec:sim_fully}. 
To produce a single ``Observed" network from the three proxy layers, we aggregated edges using either the union (``OR") or intersection (``AND") operations. 
We compared the true and aggregated networks to the BG sampler using the same sampling 
setup as in Section~\ref{subsec:sim_fully}.
%  although with $L=1$ LIP steps with $K=1$ edges updates per step due to the smaller sample size.
In addition, we included the two-stage method with LIP refinement as an additional baseline.
% \begin{figure}[!hbtp]
%     \centering
%     \includegraphics[width=0.55\linewidth]{Graphics/cs_bias_ci.png}
%     \caption{Mean over $300$ iterations of the errors $\hat{\tau} - \tau$ and of the $95\%$ credible intervals for estimating the stochastic estimand, separately for each latent interference network layer $\bA^\ast$. ``True" denotes estimation using the known true network. ``OR" and ``AND" represent union and intersection aggregation of edges across observed layers. ``BG" refers to our proposed Block Gibbs sampler. Smaller absolute errors (closer to zero) indicate better performance.}
%     \label{fig:cs_bias_ci}
% \end{figure}

Table~\ref{tab:model_performance} summarizes the MAPE values across $300$ iterations for each latent interference network layer and estimation method.
The BG method consistently improved accuracy over other methods using the aggregated observed networks 
and two-stage methods with LIP refinement.
Specifically, the BG method had estimation errors closer to those obtained with the true network.
Two-stage methods with LIP refinement also improved over aggregation methods,
 but still lagged behind the BG sampler.
Union aggregation (``OR") outperformed the intersection (``AND")
in two of the four layers (Work and Lunch), while the opposite occurred in the other two layers (Facebook and Leisure).
In addition, the BG method achieved empirical coverage rates close to the nominal level across all layers,
 while the Two-stage method and the aggregated networks did not (Appendix~\ref{apdx.sec:sim}).

\begin{table}[!htbp]
\centering
\begin{tabular}{lcccc}
\toprule
 \emph{Method} & \emph{Facebook} & \emph{Leisure} & \emph{Lunch} & \emph{Work} \\
\midrule
True & 0.070 (0.045) & 0.104 (0.080) & 0.055 (0.041) & 0.040 (0.024) \\
BG & 0.132 (0.059) & 0.157 (0.095) & 0.111 (0.049) & 0.096 (0.036) \\
Two-stage & 0.397 (0.099) & 0.187 (0.059) & 0.169 (0.044) & 0.287 (0.073) \\
OR & 0.783 (0.198) & 0.419 (0.229) & 0.272 (0.057) & 0.320 (0.065) \\
AND & 0.501 (0.093) & 0.233 (0.061) & 0.425 (0.072) & 0.324 (0.055) \\
\bottomrule
\end{tabular}
\caption{Mean (SD) of MAPE values across $300$ iterations of the TTE estimand by layer 
and estimation method.} 
\label{tab:model_performance}
\end{table}

\section{Discussion}
\label{sec:discussion}

In this paper, we have introduced a novel framework
 for estimating causal effects when only proxy measurements of the latent interference
  network are available. Our proposed SCMs provide researchers flexibility in specifying models and choosing policy-relevant estimands.
   The latent interference network introduces considerable inferential challenges. 
   Within a Bayesian framework, these challenges are translated to the  posterior distribution being over a high-dimensional mixed parameter space composed of discrete (latent networks)
    and continuous (model parameters and latent variables) components. 
    To address these challenges, we developed a Block Gibbs sampling algorithm 
    that iteratively updates the continuous and discrete parameters. 
    Specifically, discrete updates utilize Locally Informed Proposals, 
    enabling efficient exploration of the large discrete latent network space.
Crucially, these updates leverage all available data modules (e.g., proxy networks, outcomes, and treatments)
 to propose edge flips. This process actively integrates information from these diverse components, ensuring that the latent network is informed by all relevant data sources.

% Our framework accommodates static, dynamic, and stochastic estimands, addressing various policy-related queries. Optimal policy selection can then be achieved through an additional post-processing step, for example, using Bayesian decision theory \citep{berger2013} or influence maximization approaches \citep{caljon2024optimizing}. 

Our work can be extended to include estimands for peer effects
 and direct interventions on the network structure \citep{Ogburn2022}. 
 Peer effects quantify how baseline outcomes, 
 interpreted as ``treatments", propagate through the interference network.
Furthermore, estimands quantifying the effects of intervening on
 the interference network $\bA^\ast$ structure can also be defined. 
 However, even if $\bA^\ast$ had been observed, in the presence of latent common causes $\bU$ of $\bA^\ast$ and $\bY$
(Figure \ref{fig:DAGs}), these estimands cannot generally be identified from observed data alone
Instead, their estimation relies on additional modeling assumptions and prior specifications in our Bayesian framework.

As with all model-based methods, our approach relies on accurate model specification.
The validity of these specifications can be assessed using prior and posterior predictive checks, which are integral to the Bayesian workflow \citep{gelman2020bayesian}.
We derive the posterior predictive distributions and describe how to draw from them in our settings in 
Appendix~\ref{apdx.sec:posterior}.
A promising direction for future research is extending our 
framework to generalized Bayesian methods \citep{bissiri2016general,matsubara2024generalized},
which can mitigate the impact of model misspecification on inference.

While our method performs well for small to moderate network sizes, 
the dimension of the latent network space grows quadratically $\mathcal{O}(N^2)$ with 
the number of units $N$. 
Consequently, the computational burden of recovering the full posterior distribution
 via MCMC increases significantly for large networks. 
 Adapting our framework to large-scale settings remains a challenging area for future research. 
 Potential directions include developing more scalable approximation algorithms 
 or shifting towards estimators that are robust to network misspecification 
 without requiring full latent network recovery \citep[e.g., by adapting the network-misspecification-robust design-based estimator of][to model-based settings]{Weinstein2023}.

Our findings also carry broader implications for recent methodological advancements.
Existing techniques that rely on flexible outcome modeling, 
such as Graph Neural Networks \citep{ma2021gnn}, 
or those proposing novel experimental designs for networked experiments \citep{ugander2013, Eckles2017} 
typically assume that the true interference network is observed. 
Our analysis highlights the practical reality that researchers often have only proxies.
Extending these methods presents a valuable avenue for future research. 

\acks{ The authors gratefully acknowledge support from the Israel Science Foundation (ISF grant No. 2300/25). BW is supported by the Data Science Fellowship granted by the Israeli Council for Higher Education.}


\newpage

\appendix

% Reset numbering and change section names
% \setcounter{proposition}{0} 
% \renewcommand{\theproposition}{A.\arabic{proposition}}

\renewcommand{\theequation}{\thesection.\arabic{equation}}
\numberwithin{equation}{section}

\renewcommand\thefigure{\thesection.\arabic{figure}}    
\counterwithin{figure}{section}

\renewcommand\thetable{\thesection.\arabic{table}}    
\counterwithin{table}{section}

\renewcommand{\thetheorem}{\thesection.\arabic{theorem}}
\numberwithin{theorem}{section}

% 1. Start recording contents for a group named "appendix"
\startcontents[appendix]

% 2. Print the contents
% Syntax: \printcontents[name]{prefix}{start-level}{code-before}
\printcontents[appendix]{l}{1}{\section*{Appendix -- Table of Contents}}

\section{Structural Causal Models}
\label{apdx.sec:SCM_extra}
We describe here the SCMs corresponding to the other DAGs in Figure~\ref{fig:DAGs}, extending the description of the SCMs of Figure~\hyperref[fig:DAGs]{1(a)} given in \eqref{eq:SCM}.
Throughout this section, we follow the main text and assume that $f_U, f_X, f_{A^\ast},f_{A_b}, f_Z, f_Y$ are fixed functions and $\beps_U, \beps_X, \beps_{A^\ast},\beps_{A_b}, \beps_Z, \beps_Y$ are noise terms. 
We assume the pairwise noise vectors are independent, e.g., $\beps_Z\indep \beps_Y$,  
but do not limit the dependence structure between units (e.g., the dependence of $\varepsilon_{Y_i}$
 and $\varepsilon_{Y_j}$).
Furthermore, we write the structural equations for proxies $\bA_b$ separately for each proxy, as in \eqref{eq:SCM}, but extending this for other scenarios, e.g., auto-regressive proxies, is possible as described in Section~\ref{subsec:struc}.
Also, each of the outcome models $f_Y$ can be modified to depend on summarizing functions $\phi_1,\phi_2,\phi_3$ as in \eqref{eq:SCM_mod}. 
Note that the causal estimands described in Section~\ref{subsec:estimands} are interventions on $\bZ$ and therefore their interpretation remains the same regardless of whether the proxies are causal or not and whether treatment assignment is based on the latent network or the proxies.

\paragraph{Causal proxies and proxy-based treatment assignment.}
The structural equations corresponding to Figure~\hyperref[fig:DAGs]{1(b)} are
\begin{equation*}
    \begin{aligned}
        \bU_i &= f_U\big(\beps_{U_i} \big), &&\quad i \in [N] \\ 
        \bX_i &= f_X\big( \beps_{X_i} \big), &&\quad i \in [N] \\
        \bA^\ast &= f_{A^\ast}\big(\bX, \bU, \beps_{A^\ast}\big), \\
        \bA_b &= f_{A_b}\big(\bA^\ast,\bX,\beps_{A_b}\big) ,&&\quad b \in [B] \\
        Z_i &= f_Z\big(\cA, \bX_i, \bX_{-i},\varepsilon_{Z_i}\big), &&\quad i\in [N] \\
        Y_i &= f_Y\big(Z_i,\bX_i, \bZ_{-i},\bX_{-i},\bA^\ast, \bU_i, \varepsilon_{Y_i} \big) &&\quad i\in [N].
    \end{aligned}
\end{equation*}

\paragraph{Non-causal proxies and latent-based treatment assignment.}
The structural equations corresponding to Figure~\hyperref[fig:DAGs]{1(c)} are
\begin{equation*}
    \begin{aligned}
        \bU_i &= f_U\big(\beps_{U_i} \big), &&\quad i \in [N] \\ 
        \bX_i &= f_X\big( \beps_{X_i} \big), &&\quad i \in [N] \\
        \bV_i &= f_V(\beps_{V_i}), &&\\
        \bA^\ast &= f_{A^\ast}\big(\bV, \bX, \bU, \beps_{A^\ast}\big), \\
        \bA_b &= f_{A_b}\big(\bV,\bX,\beps_{A_b}\big) ,&&\quad b \in [B] \\
        Z_i &= f_Z\big(\bA^\ast, \bX_i, \bX_{-i},\varepsilon_{Z_i}\big), &&\quad i\in [N] \\
        Y_i &= f_Y\big(Z_i,\bX_i, \bZ_{-i},\bX_{-i},\bA^\ast, \bU_i, \varepsilon_{Y_i} \big) &&\quad i\in [N],
    \end{aligned}
\end{equation*}
where $f_{V}$ is a fixed function and $\beps_{V_i}$ are noise terms similar to other mentioned above. Each of $\bV_i$ can be a vector of (latent) variables.

\paragraph{Non-causal proxies and proxy-based treatment assignment.}
The structural equations corresponding to Figure~\hyperref[fig:DAGs]{1(d)} are
\begin{equation*}
    \begin{aligned}
        \bU_i &= f_U\big(\beps_{U_i} \big), &&\quad i \in [N] \\ 
        \bX_i &= f_X\big( \beps_{X_i} \big), &&\quad i \in [N] \\
        \bV_i &= f_V(\beps_{V_i}), &&\\
        \bA^\ast &= f_{A^\ast}\big(\bV, \bX, \bU, \beps_{A^\ast}\big), \\
        \bA_b &= f_{A_b}\big(\bV,\bX,\beps_{A_b}\big) ,&&\quad b \in [B] \\
        Z_i &= f_Z\big(\cA, \bX_i, \bX_{-i},\varepsilon_{Z_i}\big), &&\quad i\in [N] \\
        Y_i &= f_Y\big(Z_i,\bX_i, \bZ_{-i},\bX_{-i},\bA^\ast, \bU_i, \varepsilon_{Y_i} \big) &&\quad i\in [N].
    \end{aligned}
\end{equation*}

\section{Identifiability of Structural Parameters}
\label{apdx.sec:identif}

We consider the problem of identifying the structural parameters
 $\bTheta = (\boeta, \bgamma, \bbeta, \btheta)$ from the observed data $\bD = (\bY, \bZ, \cA, \bX)$.
Let $\mathbf{m}(\bD)$ be a vector of observable moments (e.g., means and cross-covariances).
Denote the parameter space as $\Omega_\Theta$. We assume it is a subset of $\mathbb{R}^p$ for some $p>0$.
We define the \textit{moment map} $\mathcal{M} : \Omega_\Theta \to \mathbb{R}^k$ as the expected value
 of these moments implied by the model parameters:
\begin{equation}
    \label{eq:moment_map}
    \mathcal{M}(\bTheta) = \bbE_{\bD \mid \bTheta}[\mathbf{m}(\bD)].
\end{equation}
Global identification of $\bTheta$ can be established by showing that the moment map $\mathcal{M}(\bTheta)$ is injective over the parameter space $\Omega_\Theta$ \citep{rothenberg1971}.
\begin{definition}[Global Identification via Moments]
    \label{def:global_identification}
The parameter $\bTheta$ is globally identified with respect to the moments $\mathbf{m}(\bD)$ if the map $\mathcal{M}(\bTheta)$ is injective on the parameter space $\Omega_\Theta$. That is, for any $\bTheta_1, \bTheta_2 \in \Omega_\Theta$:
\begin{equation*}
    \mathcal{M}(\bTheta_1) = \mathcal{M}(\bTheta_2) \implies \bTheta_1 = \bTheta_2.
\end{equation*}
\end{definition}
Establishing injectivity of $\mathcal{M}(\bTheta)$ can be challenging in general.
However, we can leverage the structure of the SCMs and the conditional independencies 
it implies to derive explicit moment conditions that relate the observable moments to the parameters $\bTheta$.

Each of the observed variables $\bY$, $\bZ$, and $\cA$
can provide distinct information about the structural parameters $\bTheta$.
The conditional independencies between these variables given $\bA^\ast$ and $\bX$,
imply that any dependence observed in the data must be mediated through $\bA^\ast$.
Therefore, by analyzing the covariances and higher-order moments involving these variables, we can triangulate the information about the parameters $\bTheta$.
This moment-based perspective aligns with the identification results for latent 
variable models using three independent views \citep{Allman2009}.

For example, consider the settings of causal proxies with a latent-based treatment assignment (Figure~\hyperref[fig:DAGs]{1(a)}), which implies that $\bY \indep \cA \mid \bA^\ast, \bX$ and $\bZ \indep \cA \mid \bA^\ast, \bX$.
Let $\boldsymbol{\mu}_{\bY}(\bA^\ast; \boeta) = \bbE_{\boeta}[\bY \mid \bA^\ast, \bX]$,
$\boldsymbol{\mu}_{\bZ}(\bA^\ast; \bbeta) = \bbE_{\bbeta}[\bZ \mid \bA^\ast, \bX]$,
and $\boldsymbol{\mu}_{\cA}(\bA^\ast; \bgamma) = \bbE_{\bgamma}[\text{vec}(\cA) \mid \bA^\ast, \bX]$  
be the conditional means of $\bY$, $\bZ$, and $\cA$ given $\bA^\ast$ and $\bX$, where $\text{vec}(\cA)$ denotes the vectorization of the observed proxy networks $\cA$,
and when there are multiple proxies, we compute each vectorization separately.
By the law of total covariance and the conditional independence statements above, we have the following system of covariance moments
 relating the observable moments to the parameters $\bTheta$:
\begin{equation}
    \begin{aligned}
    \label{eq.apdx:cov_moments}
     \text{Cov}(\bY, \text{vec}(\cA) \mid \bX) &= 
     \text{Cov}_{\bA^\ast \mid \bX, \btheta} \left( \boldsymbol{\mu}_{\bY}(\bA^\ast; \boeta), \boldsymbol{\mu}_{\cA}(\bA^\ast; \bgamma) \right) \\
     \text{Cov}(\bZ, \text{vec}(\cA) \mid \bX) &= 
     \text{Cov}_{\bA^\ast \mid \bX, \btheta} \left( \boldsymbol{\mu}_{\bZ}(\bA^\ast; \bbeta), \boldsymbol{\mu}_{\cA}(\bA^\ast; \bgamma) \right) \\
    \end{aligned}
\end{equation}
In addition to the above moments, we can also consider higher-order moments or other cross-moments.
For instance, $\bbE[Y_i Y_j \mid \bX]$, $\bbE[Y_i Y_j \mid A_{ij} = k, \bX]$ for $k\in \{0,1\}$,
$\text{Cov}(\text{vec}(\bA_1), \text{vec}(\bA_2) \mid \bX)$ when $\cA$ contains at least two proxies,
and $\bbE[Z_i Z_j \mid \bX]$.
These additional moments can provide further constraints on the parameters $\bTheta$,
potentially aiding in establishing global identification.
The specific choice of moments will depend on the model structure and the available data.

We now turn to several examples of simplified models
where we can explicitly derive the moment conditions and establish global identification.
Specifically, we consider a class of models where the true network
$\bA^\ast$ is generated from a homogeneous Erd\H{o}s-R\'enyi model,
the observed proxy network $\cA$ is a noisy measurement of $\bA^\ast$,
and the outcomes $\bY$ follows a linear model where exposures $\phi_1$
is the number of treated neighbors.

\subsection{Examples of Global Identification}
Assume the following model:
\begin{equation}
    \label{eq.apdx:simple_model}
    \begin{aligned}
        A^\ast_{ij} &\sim \text{Ber}(\theta) \\
        A_{ij} \mid A^\ast_{ij}=k & \sim \text{Ber}(\gamma_k), \quad k=0,1 \\
        Z_i &\sim \text{Ber}(p_z) \\
        Y_i &= \eta_1 Z_i + \eta_2 \sum_{j \neq i} A^\ast_{ij}Z_j + \varepsilon_i, \\
    \end{aligned}
\end{equation}
where $\varepsilon_i$ are independent and zero mean error terms with $\bbE[\varepsilon_i \mid \bZ, \bA^\ast] = 0$.
In this model, the structural parameters are $\bTheta = (\theta, \gamma_0, \gamma_1, p_z, \eta_1, \eta_2)$.
The true network $\bA^\ast$ is latent. We seek to recover the parameters $\bTheta$ from the observed data $\bD = (\bY,\bZ,\cA)$.
If we are able to do so, then we can recover the true network model $p(\bA^\ast \mid \theta)$,
and therefore the distribution of causal estimands (see Remark~\ref{remark:1}).
The simple model \eqref{eq.apdx:simple_model} is an example 
of causal proxies (Figure~\hyperref[fig:DAGs]{1(a)-(b)}), when there is one proxy and no covariates $\bX$ or latent variables $\bV$.

We consider here different extensions of the model \eqref{eq.apdx:simple_model} and
investigate their identifiability. 
To establish global identification, we utilize Definition~\ref{def:global_identification} 
and determine whether the system of moment equations has a unique solution,
thus ensuring that the moment map $\mathcal{M}(\bTheta)$ is injective. 

Starting with $p_z$, obviously, we can write $\bbE[Z_i] = p_z$ with the marginal first moment 
of $Z_i$. 
Also, note that 
\begin{equation}
    \label{eq.apdx:EY_given_Z}
    \bbE[Y_i \mid Z_i=z] = \eta_1 z + \eta_2 (N-1) p_z \theta,
\end{equation}
therefore, 
\begin{equation}
    \label{eq.apdx:EY}
    \bbE[Y_i] = \eta_1 p_z + \eta_2 (N-1) p_z \theta,
\end{equation}
and we can represent $\eta_1$ by
\begin{equation}
    \label{eq.apdx:eta_1}
    \eta_1 = \bbE[Y_i \mid Z_i=1] - \bbE[Y_i \mid Z_i=0].
\end{equation}
We also have
\begin{equation}
    \label{eq.apdx:m_a}
    \mu_A = \Pr(A_{ij}=1) = \gamma_1 \theta + \gamma_0 (1-\theta).
\end{equation}
Let $d_i^{obs} = \sum_{j \neq i} A_{ij}$ be the observed degree of unit $i$ in the proxy network.
Look at covariance between $Y_i$ and $d_i^{obs}$,
\begin{equation*}
    \begin{aligned}
        \text{Cov}(Y_i, d_i^{obs}) &= 
        \text{Cov}(\eta_1 Z_i + \eta_2 \sum_{j \neq i} A^\ast_{ij}Z_j, \sum_{j \neq i} A_{ij}) \\  
        &= \eta_2 \sum_{j \neq i} \text{Cov}(A^\ast_{ij}Z_j, A_{ij}) \\
        &= \eta_2 (N-1)p_z \theta(\gamma_1 - \mu_A),
    \end{aligned}
\end{equation*}
where we used the fact that 
$\bbE[Z_j A_{ij} A^\ast_{ij}] = p_z \theta \gamma_1$ and
$\bbE[Z_j A_{ij}^\ast] = p_z \theta$. 
Therefore, we can write $\eta_2$ as
\begin{equation}
    \label{eq.apdx:eta_2}
    \eta_2 = \frac{\text{Cov}(Y_i, d_i^{obs})}{(N-1)p_z \theta(\gamma_1 - \mu_A)}
    = \frac{\text{Cov}(Y_i, d_i^{obs})}{(N-1)p_z \theta(1-\theta)(\gamma_1 - \gamma_0)}.
\end{equation}
Omitting $p_z$ and $\eta_1$, as they are trivially identified, we focus on the remaining parameters
$(\theta, \gamma_0, \gamma_1, \eta_2)$.
However, we have three equations \eqref{eq.apdx:EY}, \eqref{eq.apdx:m_a}, and \eqref{eq.apdx:eta_2}
for four unknowns, so without further assumptions the model \eqref{eq.apdx:simple_model} is not identified.
Consider the following extension to \eqref{eq.apdx:simple_model} that yield global identification.
\paragraph{Known $\gamma_0$ or $\gamma_0 = f(\gamma_1)$.}
           If either $\gamma_0$ is known or there exists a known function $f$ such that $\gamma_0 = f(\gamma_1)$,
           then we can solve the three equations for the three unknowns $(\theta, \gamma_1, \eta_2)$.
              For example, if we know that $\gamma_0 = 0$, then from \eqref{eq.apdx:m_a} we have
           $\gamma_1 = \mu_A / \theta$, and substituting this into \eqref{eq.apdx:eta_2} gives
           $\eta_2 = \text{Cov}(Y_i, d_i^{obs}) / \big((N-1)p_z \mu_A (1 - \theta)\big)$.
           We can then substitute $\eta_2$ into \eqref{eq.apdx:EY_given_Z} to solve for $\theta$.
\paragraph{Two proxy networks.}
           Assume we have two proxy networks $\bA_1$ and $\bA_2$ independently generated given $\bA^\ast$ 
           as in \eqref{eq.apdx:simple_model} with the same parameters $\gamma_0$ and $\gamma_1$.
           Thus, given $A^\ast_{ij} = 1$ we have 
           \begin{equation*}
            \lvert A_{1,ij} - A_{2,ij} \rvert 
            = 
            \begin{cases}
                -1 & \text{w.p. } \gamma_1 (1-\gamma_1) \\
                0 & \text{w.p. } 1- 2\gamma_1 (1-\gamma_1) \\
                1 & \text{w.p. } \gamma_1 (1-\gamma_1),
            \end{cases}
           \end{equation*}
           and given $A^\ast_{ij} = 0$ we have
            \begin{equation*}
            \lvert A_{1,ij} - A_{2,ij} \rvert 
            = 
            \begin{cases}
                -1 & \text{w.p. } \gamma_0 (1-\gamma_0) \\
                0 & \text{w.p. } 1- 2\gamma_0 (1-\gamma_0) \\
                1 & \text{w.p. } \gamma_0 (1-\gamma_0).
            \end{cases}
           \end{equation*}
           Therefore, marginally we have
              \begin{equation*}
            \Pr\left(\lvert A_{1,ij} - A_{2,ij} \rvert = 1\right) 
            = 2 \left[\gamma_1 (1-\gamma_1) \theta + \gamma_0 (1-\gamma_0) (1-\theta)\right].
           \end{equation*}
           Define the mean absolute difference between the two proxy networks as
           \begin{equation*}
            m_{AD} = \frac{2}{N(N-1)} \sum_{i<j} \lvert A_{1,ij} - A_{2,ij} \rvert.
           \end{equation*}
           Then, $\bbE [m_{AD}] = 2 \left[\gamma_1 (1-\gamma_1) \theta + \gamma_0 (1-\gamma_0) (1-\theta)\right]$.
           Therefore, we have four equations \eqref{eq.apdx:EY}, \eqref{eq.apdx:m_a}, \eqref{eq.apdx:eta_2}, and 
           the latter expectation for $m_{AD}$ for four unknowns $(\theta, \gamma_0, \gamma_1, \eta_2)$,
           and the model is globally identified.

\paragraph{Treatment as a function of $\bA^\ast$.}
          Assume that treatment assignment depends on the true network $\bA^\ast$.
          Specifically, assume that independently for each unit $i$,
          \begin{equation*}
            \Pr(Z_i = 1 \mid \bA^\ast) = p_z \frac{d^\ast_i}{\bbE[d^\ast_i]} = 
            p_z \frac{d^\ast_i}{(N-1)\theta} \propto d^\ast_i, 
          \end{equation*}  
          where $d^\ast_i = \sum_{j \neq i} A^\ast_{ij}$ is the degree of unit $i$ in the true network.
          We can identify $p_z$ from the marginal first moment $\bbE[Z_i] = p_z$.
          Note that $Z_i$ and $Z_j$ are independent only given $\bA^\ast$,
            but marginally they are dependent. 
            Let $\kappa = \frac{p_z}{(N-1)\theta}$.
          We have
          \begin{equation*}
            \begin{aligned} 
                \bbE[Z_i Z_j] &= \bbE\left[\bbE[Z_i Z_j \mid \bA^\ast]\right] 
                \\ &= 
                \bbE\left[\left(\kappa \sum_{k \neq i} A^\ast_{ik} \right)
                            \left(\kappa \sum_{k \neq j} A^\ast_{jk} \right)
                \right]
                \\ &= 
                \kappa^2 \bbE\left[
                    (A^\ast_{ij})^2
                    + ((N-1)^2 - 1) A^\ast_{ik} A^\ast_{jl}
                \right]
                \\ &=
                \kappa^2 \left[
                    \theta + ((N-1)^2 - 1) \theta^2
                \right]
                \\ &=
                \frac{p_z^2}{(N-1)^2 \theta^2} \left[
                    \theta + ((N-1)^2 - 1) \theta^2
                \right]
                \\ &= 
                \frac{p_z^2}{(N-1)^2} \left[
                    \frac{1}{\theta} + ((N-1)^2 - 1)
                \right]
                \end{aligned}
            \end{equation*}
            solving for $\theta$ gives
            \begin{equation*}
                \theta = \frac{p_z^2}{(N-1)^2 (\bbE[Z_iZ_j] - p_z^2) + p_z^2}.
            \end{equation*}
            With this expression, we can solve for $\eta_2$ using \eqref{eq.apdx:EY}. Then, solve for $\gamma_1$ and $\gamma_0$ using \eqref{eq.apdx:m_a} and \eqref{eq.apdx:eta_2}.
            Furthermore, 
           given $\bA^\ast$, the treatment $Z_i$ is independent of the observed degree $d_i^{obs}$,
           therefore, using the law of total covariance, we can express the covariance between $Z_i$ and $d_i^{obs}$ as
           \begin{equation*}
            \begin{aligned}
                Cov(Z_i, d_i^{obs}) &= 
                Cov\left(
                    \bbE[Z_i \mid \bA^\ast], 
                    \bbE[d_i^{obs} \mid \bA^\ast]
                 \right)
                \\ &=
                Cov\left(
                    \kappa \sum_{j \neq i} A^\ast_{ij}, 
                    \sum_{j \neq i}[ A^\ast_{ij} \gamma_1 + (1 - A^\ast_{ij}) \gamma_0]
                 \right)
                \\ &=
                    Cov\left(
                    \kappa \sum_{j \neq i} A^\ast_{ij}, 
                    \sum_{j \neq i} A^\ast_{ij} (\gamma_1-\gamma_0)
                 \right)
                \\ &=
                (N-1) \kappa \theta (1-\theta) (\gamma_1 - \gamma_0)
                \\ &=
                p_z(1-\theta) (\gamma_1 - \gamma_0).    
            \end{aligned}
           \end{equation*}
           Thus, we can also express $\gamma_1 - \gamma_0 = \frac{Cov(Z_i, d_i^{obs})}{p_z(1-\theta)}$,
           then solve for $\eta_2$ using \eqref{eq.apdx:eta_2} and finally solve 
           for $\gamma_0$ and $\gamma_1$ using \eqref{eq.apdx:m_a}. 
           This implies that the structural parameters are over-identified in this scenario.
\paragraph{Network-correlated errors.}
Assume that the error terms $\varepsilon_i$ are still mean zero but have a covariance structure
$\text{Cov}(\varepsilon_i, \varepsilon_j) = \rho A^\ast_{ij}$ for some $\rho > 0$.
These represent a neighborhood correlation structure in the outcomes.
The parameters (other than $p_z$ and $\eta_1$) are now $(\theta, \gamma_0, \gamma_1, \eta_2, \rho)$.
We have,
\begin{equation}
    \label{eq.apdx:Cov_Yi_Yj}
    Cov(Y_i, Y_j) = \eta_2^2 (N-2) \theta^2 p_z (1-p_z ) + \rho \theta.
    % \approx  
    % \eta_2^2 (N-2) \theta^2 p_z (1-p_z),
\end{equation}
% where we ignored the term $\rho \theta$ in the approximation since for large $N$ it is negligible.
Unlike the independent error case, this moment depends on $\rho$.
Next, we define the auxiliary variable $R_{ij} = \eta_1 Z_i + \eta_2 \sum_{k \neq i,j} A^\ast_{ik} Z_k$.
We can write $Y_i = R_{ij} + \eta_2 A^\ast_{ij} Z_j + \varepsilon_i$.
Thus, $\bbE[R_{ij}] = \bbE[Y_i] - \eta_2 \theta p_z$.
Also, note that $\Pr(A^\ast_{ij} =1 \mid A_{ij}=1) = \frac{\theta \gamma_1}{\mu_A}$ and
$\Pr(A^\ast_{ij} =1 \mid A_{ij}=0) = \frac{\theta (1-\gamma_1)}{1-\mu_A}$.
Then, we have for $k=0,1$
\begin{equation*}
    \begin{aligned}
        \bbE[Y_i Y_j \mid A_{ij}=k] &=
        \bbE\left[
            (R_{ij} + \eta_2 A^\ast_{ij} Z_j + \varepsilon_i)
            (R_{ji} + \eta_2 A^\ast_{ij} Z_i + \varepsilon_j)
            \mid A_{ij}=k
        \right]
         \\ &=
        \bbE[R_{ij} R_{ji}] +
         \bbE\left[ \bbE \left[
            2 R_{ij} \eta_2 A^\ast_{ij} Z_j 
            + 
            (\eta_2 A^\ast_{ij})^2 Z_i Z_j 
            +
            \varepsilon_i \varepsilon_j
            \mid A^\ast_{ij}
            \right]
            \mid A_{ij}=k
            \right]
            \\ &=
            \bbE[R_{ij} R_{ji}] 
            +
            \Pr(A^\ast_{ij}=1 \mid A_{ij}=k)
            \left[ 
                2\eta_2p_z (\bbE[Y_i] - \eta_2 \theta p_z) + \eta_2^2 p_z^2 + \rho.
            \right]
    \end{aligned}
\end{equation*}
Taking the difference for $k=1$ and $k=0$ gives
\begin{equation}
    \label{eq.apdx:diff_Yi_Yj}
    \begin{aligned}
        \bbE[Y_i Y_j \mid A_{ij}=1] - \bbE[Y_i Y_j \mid A_{ij}=0] &=
      \left[
        \frac{\gamma_1\theta}{\mu_A} - \frac{(1-\gamma_1)\theta}{1-\mu_A}
      \right]
        \left[ 
                2\eta_2p_z (\bbE[Y_i] - \eta_2 \theta p_z) + \eta_2^2 p_z^2 + \rho.
                \right]
        \end{aligned}
    \end{equation}
This provides five equations for the five unknowns $(\theta, \gamma_0, \gamma_1, \eta_2, \rho)$,
thereby achieving global identification.
The system formed by equations \eqref{eq.apdx:EY}, \eqref{eq.apdx:m_a},
 \eqref{eq.apdx:eta_2}, \eqref{eq.apdx:Cov_Yi_Yj}, and \eqref{eq.apdx:diff_Yi_Yj}
  provides five independent equations for the five 
  unknowns $(\theta, \gamma_0, \gamma_1, \eta_2, \rho)$,
   thereby achieving global identification.

\subsection{Local Identifiability}
\label{apdx.subsec:local_id}

We now discuss local identifiability of the model parameters $\bTheta$.
 Local identifiability is weaker than global identifiability.
 It means that in a neighborhood of the true parameter value $\bTheta_0$,
 there are no other observationally equivalent parameter values \citep{rothenberg1971}.
 That is, there is identifiability only in a small region around $\bTheta_0$.
 A standard approach to establish local identifiability is through the Fisher Information Matrix (FIM).
 Specifically, \citet[][Theorem~1]{rothenberg1971}
 established that $\bTheta$ is locally identifiable at $\bTheta_0$ if and only if
 the FIM is nonsingular at $\bTheta_0$.
 We describe here the general conditions for local identifiability
 in our setting with latent network $\bA^\ast$.
 We then provide some intuition on when these conditions hold.
 This result is formal and complements the global identification arguments presented earlier.
 However, verifying local identifiability through the FIM in practice 
 is challenging due to the intractability of the observed data likelihood.

Consider for simplicity the causal proxies and latent-based treatment assignment (Figure~\hyperref[fig:DAGs]{1(a)}).
Let $\bTheta = (\boeta, \bbeta_Z, \bgamma, \btheta)$ be the vector of all model parameters.
Omitting covariates for simplicity, the observed data are $\bD = (\bY, \bZ, \cA)$.
The latent network $\bA^\ast$ is missing.
In this case, the complete data (if $\bA^\ast$ were observed) likelihood factorizes as
\begin{equation}
    \label{eq.apdx:complete_data_likelihood}
   p(\bD, \bA^\ast \mid \bTheta) = 
   p(\bY \mid \bZ, \bA^\ast, \boeta)
   p(\bZ \mid \bA^\ast, \bbeta_Z)
   p(\cA \mid \bA^\ast, \bgamma)
   p(\bA^\ast \mid \btheta).
\end{equation}
However, the observed data likelihood is $p(\bD \mid \bTheta) = \sum_{\bA^\ast} p(\bD, \bA^\ast \mid \bTheta)$,
which is impossible to compute in closed form for most models of interest.
The complete data score function is defined as 
\begin{equation}
    \label{eq.apdx:complete_data_score}
    \bS_{comp}(\bD, \bA^\ast \mid \bTheta) = \nabla_{\bTheta} \log p(\bD, \bA^\ast \mid \bTheta).
\end{equation}
Similarly, the observed data score function is
\begin{equation}
    \label{eq.apdx:observed_data_score}
    \bS_{obs}(\bD \mid \bTheta) = \nabla_{\bTheta} \log p(\bD \mid \bTheta).
\end{equation}
We can represent the observed data score function \eqref{eq.apdx:observed_data_score} as
 the posterior expected complete data score function \eqref{eq.apdx:complete_data_score}.
 Since $\nabla_{\bTheta} p(\bD \mid \bTheta) = \sum_{\bA^\ast} \nabla_{\bTheta} p(\bD, \bA^\ast \mid \bTheta)$, 
 we have 
\begin{equation}
    \label{eq.apdx:score_decomp}
    \begin{aligned}
    \bS_{obs}(\bD \mid \bTheta) &=
    \frac{\nabla_{\bTheta}p(\bD \mid \bTheta)}{p(\bD \mid \bTheta)} 
    \\ &= 
    \sum_{\bA^\ast} \frac{\nabla_{\bTheta} p(\bD, \bA^\ast \mid \bTheta)}{p(\bD \mid \bTheta)}
    \\ &=
    \sum_{\bA^\ast} \frac{p(\bD, \bA^\ast \mid \bTheta)}{p(\bD \mid \bTheta)}
    \frac{\nabla_{\bTheta} p(\bD, \bA^\ast \mid \bTheta)}{p(\bD, \bA^\ast \mid \bTheta)}
    \\ &=
    \sum_{\bA^\ast} p(\bA^\ast \mid \bD, \bTheta)
    \bS_{comp}(\bD, \bA^\ast \mid \bTheta)
    \\ &= 
    \mathbb{E}_{\bA^\ast \mid \bD, \bTheta}
    \left[ \bS_{comp}(\bD, \bA^\ast \mid \bTheta) \right].
    \end{aligned}
\end{equation}
The Fisher Information Matrix (FIM) of the complete data is defined as the 
covariance of the complete data scores:
\begin{equation}
    \label{eq.apdx:complete_data_FIM}
    \begin{aligned}
        \mathscr{I}_{comp}(\bTheta) &= 
        \mathbb{E}_{\bD, \bA^\ast \mid \bTheta}
        \left[ \bS_{comp}(\bD, \bA^\ast \mid \bTheta)
        \bS_{comp}(\bD, \bA^\ast \mid \bTheta)^\top \right]
        \\ &= 
        \text{Var}_{\bD, \bA^\ast \mid \bTheta}
        \left[ \bS_{comp}(\bD, \bA^\ast \mid \bTheta)\right].
    \end{aligned}
\end{equation}
By the law of total variance, we can decompose the complete data FIM \eqref{eq.apdx:complete_data_FIM} as
\begin{equation}
    \label{eq.apdx:FIM_decomp}
    \mathscr{I}_{comp}(\bTheta) = 
    \mathbb{E}_{\bD \mid \bTheta}
    \left[
    \text{Var}_{\bA^\ast \mid \bD, \bTheta}
    \left[ \bS_{comp}(\bD, \bA^\ast \mid \bTheta) \right]
    \right]
    +
    \text{Var}_{\bD \mid \bTheta}
    \left[
    \mathbb{E}_{\bA^\ast \mid \bD, \bTheta}
    \left[ \bS_{comp}(\bD, \bA^\ast \mid \bTheta) \right]
    \right].
\end{equation}
From \eqref{eq.apdx:score_decomp}, we see that the second term in \eqref{eq.apdx:FIM_decomp} is precisely the observed data FIM:
\begin{equation}
    \label{eq.apdx:observed_data_FIM}
    \mathscr{I}_{obs}(\bTheta) = 
    \text{Var}_{\bD \mid \bTheta}
    \left[\bS_{obs}(\bD \mid \bTheta)
    \right]
    =
    \text{Var}_{\bD \mid \bTheta}
    \left[
    \mathbb{E}_{\bA^\ast \mid \bD, \bTheta}
    \left[ \bS_{comp}(\bD, \bA^\ast \mid \bTheta) \right]
    \right].
\end{equation}
Moreover, the first term in \eqref{eq.apdx:FIM_decomp} is the expected missing FIM, 
representing the variance of the complete data scores with respect to the 
posterior distribution of the missing data $\bA^\ast$ given the observed data $\bD$ and parameters $\bTheta$.
This is just the missing data FIM: 
\begin{equation}
    \label{eq.apdx:missing_information_FIM}
    \mathscr{I}_{miss}(\bTheta) = 
    \mathbb{E}_{\bD \mid \bTheta}
    \left[
    \text{Var}_{\bA^\ast \mid \bD, \bTheta}
    \left[ \bS_{comp}(\bD, \bA^\ast \mid \bTheta) \right]
    \right].
\end{equation}
Combining \eqref{eq.apdx:FIM_decomp}, \eqref{eq.apdx:observed_data_FIM}, and \eqref{eq.apdx:missing_information_FIM}, we obtain the missing information principle \citep{Louis1982}:
\begin{equation}
    \label{eq.apdx:missing_information_principle}
    \mathscr{I}_{obs}(\bTheta) = 
    \mathscr{I}_{comp}(\bTheta) - \mathscr{I}_{miss}(\bTheta).
\end{equation}
Assuming there are no additional latent variables in the SCM (Figure~\hyperref[fig:DAGs]{1(a)})\footnote{If there are any latent variables, the same approach will work but researchers have to model their process.}, 
the complete data likelihood \eqref{eq.apdx:complete_data_likelihood} composed of several independent components (in terms of the parameters).
Thus, the complete data FIM \eqref{eq.apdx:complete_data_FIM} is block-diagonal with
 blocks corresponding to each of the modules in \eqref{eq.apdx:complete_data_likelihood}.
For example, under \eqref{eq.apdx:complete_data_likelihood}, we have 
\begin{equation*}
    \mathscr{I}_{comp}(\bTheta) = \text{diag}\big(\mathscr{I}_{comp}(\boeta), \mathscr{I}_{comp}(\bbeta_Z), \mathscr{I}_{comp}(\bgamma), \mathscr{I}_{comp}(\btheta)\big)
\end{equation*}
where, for example, 
$\mathscr{I}_{comp}(\boeta) = \bbE_{\bD,\bA^\ast \mid \boeta}\left[-\nabla_{\boeta}^2 \log p(\bY \mid \bZ,\bA^\ast,\boeta)\right]$.
That is, had we observed $\bA^\ast$, the parameters of each module would be independent
and the complete data FIM block-diagonal.
However, the missing data FIM \eqref{eq.apdx:missing_information_FIM} is not block-diagonal, 
as the missing network $\bA^\ast$ induces dependencies between the different modules.
Therefore, from \eqref{eq.apdx:missing_information_principle} the observed data FIM $\mathscr{I}_{obs}$
 is not block-diagonal either.
For the observed data FIM to be positive definite and therefore nonsingular, the missing data FIM cross-blocks must not be too large compared to the complete data FIM blocks.
Namely, the observed data FIM need to be strictly block-diagonal dominant.
This leads to the following sufficient conditions for local identifiability of the parameters $\bTheta$
at point $\bTheta_0$.
\begin{theorem}[Local Identifiability]
    \label{thm.apdx:identif}
    Let $\mathscr{I}_{obs, kk}(\bTheta_0) = \mathscr{I}_{comp, kk}(\bTheta_0) - \mathscr{I}_{miss, kk}(\bTheta_0)$
    be the observed data FIM in module $k$ at point $\bTheta_0$. 
    The parameters $\bTheta$ are locally identifiable at point $\bTheta_0$ if $\mathscr{I}_{obs}$ 
    is a strictly block-diagonal dominant matrix. 
    That is, if for every module $k$,
\begin{equation}
    \label{eq.apdx:diag_dominance}
    \lambda_{\min}\left(\mathscr{I}_{comp, kk}(\bTheta_0) - \mathscr{I}_{miss, kk}(\bTheta_0)\right) 
    > 
    \sum_{j \neq k} \lVert \mathscr{I}_{miss, jk}(\bTheta_0) \rVert_2,
\end{equation}
where $\lambda_{\min}$ denotes the minimum eigenvalue and $\lVert \cdot \rVert_2$ is the spectral norm.
\end{theorem}
\begin{proof}
    We can show it by contradiction. Assume that \eqref{eq.apdx:diag_dominance} holds,
    but $\mathscr{I}_{obs} = \mathscr{I}_{obs}(\bTheta_0)$ is singular at $\bTheta_0$.
    Then, there exists a non-zero vector $\mathbf{c}$ such that
    $\mathscr{I}_{obs} \mathbf{c} = \boldsymbol{0}$.
    Partition $\mathbf{c}$ into blocks $\mathbf{c} = (\mathbf{c}_1^\top, \ldots, \mathbf{c}_K^\top)^\top$.
    Thus, for any block $k$ of $\mathscr{I}_{obs}$, we have
    \begin{equation*}
        \mathscr{I}_{obs, kk} \mathbf{c}_k +
        \sum_{j \neq k} \mathscr{I}_{obs, jk} \mathbf{c}_j = \boldsymbol{0}.
    \end{equation*}
    By the information decomposition \eqref{eq.apdx:missing_information_principle}, we can write it as
    \begin{equation*}
        \left(\mathscr{I}_{comp, kk} - \mathscr{I}_{miss, kk}\right) \mathbf{c}_k +
        \sum_{j \neq k} \left(-\mathscr{I}_{miss, jk}\right) \mathbf{c}_j = \boldsymbol{0}.
    \end{equation*}
    Rearranging terms, we obtain
    \begin{equation*}
        \left(\mathscr{I}_{comp, kk} - \mathscr{I}_{miss, kk}\right) \mathbf{c}_k =
        \sum_{j \neq k} \mathscr{I}_{miss, jk} \mathbf{c}_j.
    \end{equation*}
    Taking norms on both sides and using the triangle inequality, we get
    \begin{equation*}
        \lVert \left(\mathscr{I}_{comp, kk} - \mathscr{I}_{miss, kk}\right) \mathbf{c}_k \rVert_2
        =
        \left\lVert \sum_{j \neq k} \mathscr{I}_{miss, jk} \mathbf{c}_j \right\rVert_2
        \leq
        \sum_{j \neq k} \lVert \mathscr{I}_{miss, jk} \rVert_2 \lVert \mathbf{c}_j \rVert_2
\end{equation*}
    As $\mathscr{I}_{comp, kk} - \mathscr{I}_{miss, kk}$ is positive definite and symmetric since 
    condition \eqref{eq.apdx:diag_dominance} holds,
    then all its eigenvalues are positive, and we have
    \begin{equation*}
        \lVert \left(\mathscr{I}_{comp, kk} - \mathscr{I}_{miss, kk}\right) \mathbf{c}_k \rVert_2
        \geq
        \lambda_{\min}\left(\mathscr{I}_{comp, kk} - \mathscr{I}_{miss, kk}\right) \lVert \mathbf{c}_k \rVert_2.
    \end{equation*}
    Take module $k$ with the largest $\lVert \mathbf{c}_k \rVert_2$.
    Combining the above inequalities, we obtain
    \begin{equation*}
        \lambda_{\min}\left(\mathscr{I}_{comp, kk} - \mathscr{I}_{miss, kk}\right) \lVert \mathbf{c}_k \rVert_2
        \leq
        \sum_{j \neq k} \lVert \mathscr{I}_{miss, jk} \rVert_2 \lVert \mathbf{c}_j \rVert_2
        \leq
        \sum_{j \neq k} \lVert \mathscr{I}_{miss, jk} \rVert_2 \lVert \mathbf{c}_k \rVert_2.
    \end{equation*}
    Dividing both sides by $\lVert \mathbf{c}_k \rVert_2 > 0$, we get
    \begin{equation*}
        \lambda_{\min}\left(\mathscr{I}_{comp, kk} - \mathscr{I}_{miss, kk}\right)
        \leq
        \sum_{j \neq k} \lVert \mathscr{I}_{miss, jk} \rVert_2,
    \end{equation*}
    which contradicts \eqref{eq.apdx:diag_dominance}.
    Thus, $\mathscr{I}_{obs}$ is nonsingular, and, 
    by \citet[][Theorem~1]{rothenberg1971}, the parameters $\bTheta$ are locally identifiable. 
\end{proof}
We can give some structure and inuition regarding the $(j,k)$-th block of $\mathscr{I}_{miss}$.
Let $S_k(\bD, \bA^\ast \mid \bTheta)$ be the score function of module $k$ (e.g., outcome, treatment, proxy, or network model) with respect to its parameters.
The partial First-order Taylor expansion for $S_k$ with respect to $\bA^\ast$
 around the posterior mean $\bbE [\bA^\ast \mid \bD, \bTheta]$ yields
\begin{equation*}
    S_k(\bD,\bA^\ast \mid \bTheta) \approx
     S_k(\bD, \bbE [\bA^\ast \mid \bD, \bTheta] \mid \bTheta) + 
     \nabla_{\bA^\ast} S_k(\bD, \bbE [\bA^\ast \mid \bD, \bTheta] \mid \bTheta) 
     (\bA^\ast - \bbE [\bA^\ast \mid \bD, \bTheta]),
\end{equation*}
where the gradient is with respect to the entries of $\bA^\ast$, and should correspond to the discrete
differences since $\bA^\ast$ is a binary adjacency matrix.
Namely, to changes in one edge at a time.
For brevity, write $S_k(\bD,\bA^\ast \mid \bTheta) = S_k(\bA^\ast)$.
Using the above approximation, we can express the posterior covariance 
between scores of modules $j$ and $k$ as:
\begin{equation*}
    \text{Cov}_{\bA^\ast \mid \bD, \bTheta} 
    \left( S_j(\bA^\ast), S_k(\bA^\ast) \right) \approx
    \nabla_{\bA^\ast} S_j(\bbE [\bA^\ast \mid \bD, \bTheta]) \;
    \text{Cov}_{\bA^\ast \mid \bD, \bTheta}(\bA^\ast) \;
    \nabla_{\bA^\ast} S_k(\bbE [\bA^\ast \mid \bD, \bTheta])^\top, 
\end{equation*}
Thus, the $(j,k)$ cross-block of the missing data FIM \eqref{eq.apdx:missing_information_FIM} can be approximated as
\begin{equation}
    \label{eq.apdx:missing_information_FIM_block_approx}
    \mathscr{I}_{miss, jk} \approx
    \bbE_{\bD \mid \bTheta} \left[
    (\nabla_{\bA^\ast} S_j) \;
    \Sigma_{\bA^\ast \mid \bD, \bTheta} \;
    (\nabla_{\bA^\ast} S_k)^\top
    \right],
\end{equation}
where $\Sigma_{\bA^\ast \mid \bD, \bTheta} = \text{Cov}_{\bA^\ast \mid \bD, \bTheta}(\bA^\ast)$
is the posterior covariance of the latent network.
This expression shows that the missing data FIM between modules $j$ and $k$ depends
on how sensitively their scores respond to changes in $\bA^\ast$ (the gradients), weighted by the uncertainty in $\bA^\ast$ given the data (the posterior covariance).
Specifically, $\nabla_{\bA^\ast} S_k$ captures how much parameter block $k$ (e.g, $\boeta$) is sensitive to 
one edge changes in $\bA^\ast$.
If two modules have similar sensitivities to $\bA^\ast$ for edges with high posterior uncertainty,
then their scores will be correlated, leading to larger off-diagonal blocks in $\mathscr{I}_{miss}$.
If that is the case for many modules pairs and edges in $\bA^\ast$, then the missing data information matrix will have large off-diagonal entries.
By \eqref{eq.apdx:observed_data_FIM} the observed data FIM is the difference between the block-diagonal complete data FIM and the dense missing data FIM.
Therefore, if the off-diagonal blocks of $\mathscr{I}_{miss}$ are large, then the observed data FIM will be close to singular.
On the other hand, if the off-diagonal blocks of $\mathscr{I}_{miss}$ are small compared to the diagonal blocks, then the observed data FIM will be diagonally dominant and thus invertible.

For example, if two modules depend on $\bA^\ast$ through similar sufficient statistics (e.g., both depend on the degree of nodes), then their score functions will respond similarly to changes in $\bA^\ast$, leading to high covariance in their scores and larger off-diagonal blocks in 
$\mathscr{I}_{miss}$.
Conversely, if modules depend on $\bA^\ast$ through distinct sufficient statistics,
e.g., if the proxy networks module $p(\cA \mid \bA^\ast, \bgamma)$
depend on individual edges while the module of outcomes $p(\bY \mid \bZ, \bA^\ast, \boeta)$
depend on aggregated properties (e.g., number of treated neighbors),
then their score functions will respond differently to single edge changes in $\bA^\ast$,
leading to lower covariance in their scores and smaller off-diagonal blocks in $\mathscr{I}_{miss}$.
This analysis suggests that identification improves when the modules depend on $\bA^\ast$ through
distinct structural features, or when the posterior uncertainty in $\bA^\ast$ is low.

\section{Sampling From The Posterior}
\label{apdx.sec:posterior}

\subsection{Non-Causal Proxies}

Under the scenario described in Figures~\hyperref[fig:DAGs]{1(c)-(d)} of non-causal proxy networks $\cA$, we assume there exist latent variables $\bV$ such that $\cA \leftarrow \bV \rightarrow \bA^\ast$ but without a direct path between the proxies $\cA$ and the latent network $\bA^\ast$.

As in the main text, in the case of latent-based treatment assignment  (Figure~\hyperref[fig:DAGs]{1(c)}),
the treatment assignment model is $p(\bZ \mid \bX, \bA^\ast, \bbeta)$.
In that case, the joint posterior distribution \eqref{eq:observed.posterior} should also include the unobserved $\bV$ and can be written as
\begin{equation}
    \label{eq:posterior_noncausal}
    \begin{split}
      p(\boeta, \bbeta, \btheta,\bgamma, \bbeta_{\bV},\bV,\bA^\ast \mid \bD) 
        \propto
        &\; p(\bY \mid \bZ,\bA^\ast,\bX,\boeta) p(\boeta) 
        \\ & 
         \times p(\bZ \mid \bA^\ast, \bX, \bbeta) p(\bbeta)
        \\ & 
      \times p(\cA \mid \bX, \bV, \bgamma) 
       p(\bgamma)
        \\ &
       \times p(\bA^\ast \mid \bX, \bV, \btheta) 
       p(\btheta)
       \\ &
       \times p(\bV \mid \bbeta_{\bV})p(\bbeta_{\bV}),
    \end{split}
\end{equation}
where $p(\bV \mid \bbeta_{\bV})$ is an assumed model for the latent variables (which could often be assumed to be multivariate normal) parameterized by $\bbeta_{\bV}$ with prior $p(\bbeta_{\bV})$.
We obtain that the posterior of $\bA^\ast$ given the data and the other unknowns (the analog of \eqref{eq:a_star_post} that also includes $\bV$) is 
\begin{equation}
    \label{eq:a_star_post_noncausal}
    \begin{aligned}
        p(\bA^\ast \mid \cdot) 
        &\equiv 
        p(\bA^\ast \mid \bD, \boeta, \bbeta, \btheta,\bgamma, \bV) 
        \\ & \propto p(\bY \mid \bZ, \bA^\ast, \bX, \boeta) 
        \\ & \times p(\bA^\ast \mid \bX, \bV,\btheta)
        \\ & \times p(\bZ \mid \bA^\ast, \bX, \bbeta),
    \end{aligned}
\end{equation}
which simplifies the Locally Informed Proposals updates as it does not include the proxy network model $p(\cA \mid \cdot)$. Thus, the proposed updates for $\bA^\ast$ will depend on the proxies $\cA$ only through their information about the latent variable $\bV$ and the joint structure of the parameters $\btheta$ and  $\bgamma$ as will be captured in the posterior sampling through the conditional posterior sampling of $\bV$ (via the continuous kernel in the Block Gibbs algorithm).

For the scenario of non-causal proxies with proxy-based treatment assignment (Figure~\hyperref[fig:DAGs]{1(d)}), the posterior is similar to \eqref{eq:posterior_noncausal} but without the treatment assignment model terms $p(\bZ \mid \bA^\ast, \bX, \bbeta) p(\bbeta)$. That is, in this scenario 
the treatment model is $p(\bZ \mid \cA, \bX, \bbeta)$ and does not provide information on $\bA^\ast$.
Therefore, we can marginalize out $\bbeta$ and $\bZ$ from \eqref{eq:posterior_noncausal} 
(similarly to \eqref{eq:observed.posterior.obs}) to obtain the posterior distribution
\begin{equation}
    \label{eq:posterior_noncausal_rct}
    \begin{split}
      p(\boeta, \btheta,\bgamma, \bbeta_{\bV},\bV,\bA^\ast \mid \bD) 
        \propto
        &\; p(\bY \mid \bZ,\bA^\ast,\bX,\boeta) p(\boeta) 
        \\ & 
      \times p(\cA \mid \bX, \bV, \bgamma) 
       p(\bgamma)
        \\ &
       \times p(\bA^\ast \mid \bX, \bV, \btheta) 
       p(\btheta)
       \\ &
       \times p(\bV \mid \bbeta_{\bV})p(\bbeta_{\bV}).
    \end{split}
\end{equation}
Therefore, in this case, the posterior of $\bA^\ast$ given the data and the other unknowns is
\begin{equation}
    \label{eq:a_star_post_noncausal_rct}
        p(\bA^\ast \mid \cdot) \propto p(\bY \mid \bZ, \bA^\ast, \bX, \boeta) 
        \times p(\bA^\ast \mid \bX, \bV,\btheta),
\end{equation}
which is similar to \eqref{eq:a_star_post_noncausal} but without the treatment model term.
As described below, for estimating the causal effects, we marginalize out all other unknowns except $\bA^\ast$ and $\boeta$ for the posteriors.

\subsection{Extensions}

We provide details on how to adapt the posterior in two cases: when researchers want to augment the outcome model with propensity score and when there are latent network-outcome confounders $\bU$. 

\begin{itemize}
    \item 
    Under latent-based treatment assignment, the treatment model is 
    $p(\bZ \mid \bX, \bA^\ast, \bbeta)$.
    As mentioned above, this model informs the outcome model only through its effect on $\bA^\ast$.
    Also, under proxy-based treatment assignment, the treatment model $p(\bZ \mid \cA, \bX, \bbeta)$ does not provide information on $\bA^\ast$, and therefore does not affect the outcome model.
    In both cases, researchers might want to augment the outcome model with propensity scores 
    $p(\bZ \mid \cdot)$. 
    However, misspecification of the treatment model can adversely affect the estimation of causal effects \citep{zigler2013}.
    To mitigate this issue, it is often suggested to remove 'feedback' between the two models to reduce sensitivity to model misspecification \citep{zigler2013}. That can be achieved by replacing the continuous updates in the Block Gibbs algorithm with samples from a cut-posterior \citep{Bayarri2009, Jacob2017} of the continuous parameters that removed the feedback term between the propensity scores and the outcome model. Such augmentation will not directly affect $\bA^\ast$ updates.
    
    \item Assume there is unmeasured network-outcome confounding represented by latent variable $\bU$, and accounting for it is required for the estimation of, for example, the impact of interventions on the interference network structure. In this case, researchers specify a model $p(\bU \mid \bbeta_U)$ for parameters $\bbeta_U$ with prior $p(\bbeta_U)$ and assume prior independence as before. The network and outcome models now include $\bU$: $p(\bA^\ast \mid \bX, \bU,\btheta)$ and $p(\bY \mid \bZ,\bA^\ast, \bX, \bU,\boeta)$. As $\bU$ will typically be modeled as a continuous latent variable (e.g., multivariate normal), it can be included with the rest of the continuous latent variables in the Block Gibbs algorithm. Namely, in $\bA^\ast$ updates with LIP will include the current state of $\bU$ as a fixed variable, and in the continuous updates,   $\bU$ will be updated together with the other continuous latent variables and parameters. 
    A similar approach of network-outcome confounding adjustment using latent variables has been proposed by \citet{um2024bayesian}, who assumed the interference network is fully observed.
\end{itemize}
\subsection{Block Gibbs Initialization}

As explained in the main text, the proposed Block Gibbs algorithm can be 
initialized by estimators of latent variables obtained from sampling from the `cut' posterior \citep{Bayarri2009, Jacob2017}.
This can be followed by a further refinement of $\bA^\ast$ using several LIP steps, 
which is found to greatly improve mixing and convergence of the Block Gibbs sampler in practice.
We first present the sampling procedure and then discuss the choice of initial values for the Block Gibbs algorithm. 

Sampling from the cut-posterior is a form of Bayesian modularization. It offers several advantages \citep{Bayarri2009}. Notably, it severely reduces, and often eliminates, contamination between modules due to misspecification in some modules.
This approach also allows for separate model checking and selection of the network models (true and observed) from that of the outcome model. Researchers can first perform model diagnostics on the network modules and subsequently repeat the process for the outcome model, effectively separating the diagnostic process for the interference network and the outcome model.
Moreover, by breaking the complex full posterior into manageable parts, modularization provides a computationally tractable approach where direct sampling would be challenging, as evident in our case. 

We consider here the scenario of causal proxies with proxy-based treatment assignment (Figure~\hyperref[fig:DAGs]{1(b)})
that have posterior \eqref{eq:observed.posterior.obs}, but the approach can be adapted to other scenarios with the posteriors \eqref{eq:observed.posterior}, \eqref{eq:posterior_noncausal}, or \eqref{eq:posterior_noncausal_rct}.
The posterior \eqref{eq:observed.posterior.obs} can be written as
\begin{equation*}
    p(\boeta,\btheta,\bgamma, \bA^\ast \mid \bD) = p(\boeta \mid \bD,\btheta,\bgamma,\bA^\ast)p(\btheta,\bgamma, \bA^\ast \mid \bD).
\end{equation*}
%where we use the fact that $p(\boeta \mid \bD, \bA^\ast, \btheta, \gamma)=p(\boeta \mid \bY,\bZ,\bA^\ast,\bX)$. 
This implies that the posterior is composed of two modules: network module $p(\btheta,\bgamma, \bA^\ast \mid \bD)$ and outcome module $p(\boeta \mid \bD,\btheta,\bgamma,\bA^\ast)$. The modules are connected and can thus generate `feedback' loops between one another. To see that, notice that 
\begin{align*}
     p(\btheta,\bgamma, \bA^\ast \mid \bD) &= 
     \int_{\boeta}  p(\boeta,\btheta,\bgamma, \bA^\ast \mid \bD) d\boeta
     \\ &\propto 
   p(\btheta)p(\bgamma)p(\cA \mid \bA^\ast, \bX,\bgamma)
   p(\bA^\ast \mid \bX,\btheta) 
 p(\bY \mid \bZ,\bA^\ast,\bX)
 \\ &\equiv
 p(\btheta,\bgamma,\bA^\ast \mid \cA,\bX)
 p(\bY \mid \bZ,\bA^\ast,\bX)
\end{align*}
where $p(\bY \mid \bZ,\bA^\ast,\bX) = \int_{\boeta} p(\bY \mid \bZ,\bA^\ast,\bX,\boeta)p(\boeta)d\boeta$. Namely, the posterior of $(\btheta,\bgamma,\bA^\ast)$ given the observed data $\bD$ is proportional to the posterior given $\bX$ and $\cA$ (the ``first module" or ``network module") multiplied by $p(\bY \mid \bZ,\bA^\ast,\bX)$, the feedback term between the network and outcome modules. Accordingly, the outcome module can be written as
\begin{align*}
    p(\boeta \mid \bD,\btheta,\bgamma,\bA^\ast) &= \frac{p(\boeta,\btheta,\bgamma, \bA^\ast \mid \bD)}{p(\btheta,\bgamma, \bA^\ast \mid \bD)}
    \\ &\propto \frac{p(\bY \mid \bZ,\bA^\ast,\bX,\boeta)p(\boeta)}{p(\bY \mid \bZ,\bA^\ast,\bX)} 
    \\ &\propto p(\bY \mid \bZ,\bA^\ast,\bX,\boeta)p(\boeta),
\end{align*}
which can be depicted as $p(\boeta \mid \bD,\btheta,\bgamma,\bA^\ast) \equiv p(\boeta \mid \bY,\bZ,\bA^\ast,\bX)$. 
Viewed as two modules with feedback between them, the posterior can thus be represented as
\begin{equation*}
    p(\boeta,\btheta,\bgamma,\bA^\ast \mid \bD) \propto 
    \underbrace{p(\boeta \mid \bY,\bZ,\bA^\ast,\bX)}_{\text{Outcome module}}
    \underbrace{p(\btheta,\bgamma,\bA^\ast \mid \cA,\bX)}_{\text{Network module}}
    \underbrace{p(\bY \mid \bZ,\bA^\ast,\bX)}_{\text{Feedback term}}.
\end{equation*}
The cut-posterior that removes the feedback term from the full posterior is
\begin{equation*}
    p_{cut}(\boeta,\btheta,\bgamma,\bA^\ast \mid \bD) \propto 
    p(\boeta \mid \bY,\bZ,\bA^\ast,\bX)
    p(\btheta,\bgamma,\bA^\ast \mid \cA,\bX)
\end{equation*}
\textbf{Sampling from the network module.}
The network module is proportional to
\begin{equation*}
     p(\btheta,\bgamma,\bA^\ast \mid \cA,\bX) \propto p(\btheta)p(\bgamma)p(\cA \mid \bA^\ast, \bX,\bgamma)
   p(\bA^\ast \mid \bX,\btheta), 
\end{equation*}
and can also be represented as  
\begin{equation*}
    % \label{eq:network_module_2}
    p(\btheta,\bgamma,\bA^\ast \mid \cA,\bX) =
    p(\bA^\ast \mid \cA,\bX,\btheta,\bgamma)
    p(\btheta,\bgamma \mid \cA,\bX),
\end{equation*}
where
\begin{equation*}
        % \label{eq:marginal_network_mod}
\begin{aligned}   
   p(\btheta,\bgamma \mid \cA,\bX) &=
    \sum_{\bA^\ast} p(\btheta,\bgamma,\bA^\ast \mid \cA,\bX)
    \\ &= 
    p(\btheta)p(\bgamma)
    \sum_{\bA^\ast} 
    p(\cA\mid \bA^\ast,\bX,\bgamma) 
    p(\bA^\ast \mid \bX,\btheta).
\end{aligned}
\end{equation*}
Thus, sampling from the network module can be performed by first sampling $(\btheta,\bgamma)$ from the marginalized model $ p(\btheta,\bgamma \mid \cA,\bX)$, and then sampling networks from $p(\bA^\ast \mid \cA,\bX,\btheta,\bgamma)$.
For example, if both the network generation model $p(\bA^\ast \mid \bX,\btheta)$ and the observed network model $p(\cA\mid \bA^\ast,\bX,\bgamma)$ are conditionally dyad independent, i.e., can be written as 
\begin{align*}   
    p(\bA^\ast \mid \bX,\btheta) &= 
     \prod_{i>j} \Pr(A^\ast_{ij} \mid \bX,\btheta)
     % \equiv \prod_{i>j} \nu(A^\ast_{ij}; \btheta) 
     \\
    p(\cA \mid \bA^\ast,\bgamma) &= 
    \prod_{i>j} \Pr(A_{ij} \mid A^\ast_{ij},\bX,\bgamma) 
    % \equiv \prod_{i>j} \xi(A_{ij} ; A^\ast_{ij},\bgamma)
    ,
\end{align*}
% for some functions $\nu$ and $\xi$ that can also depend on covariates $\bX$ and latent variables $\bU$. 
the sampling from the network module will be simplified, as we now show.
Examples for $\Pr(A^\ast_{ij} \mid \cdot)$ include random graphs, SBM, and latent space models. Examples for $\Pr(A_{ij} \mid \cdot)$ are given in Section \ref{subsec:net_models}. When 
$\cA = (\bA_1,\dots,\bA_B)$, if 
\begin{equation*}
    p(\cA\mid \bA^\ast,\bX,\bgamma) = \prod_{b}\prod_{i>j} \text{Pr}_b(A_{b,ij} \mid A^\ast_{ij},\bX,\bgamma_b),
    \end{equation*}
then we proceed the same way. For example, the factorization over $b$ conditional on $\bgamma$ is feasible when $\bgamma$ has a hierarchical prior. Alternatively, in the case of auto-regressive proxy measurements, we can often write
\begin{align*}
    p(\cA\mid \bA^\ast,\bX,\bgamma) &=
    p(\bA_1 \mid \bA^\ast,\bX,\bgamma_1)\prod_
    {b>1}p(\bA_b \mid \bA^\ast,\bX, \bA_1,\ldots, \bA_{b-1},\bgamma).
\end{align*}
For ease of presentation, consider the case of a single proxy network ($B=1$). We obtain,
\begin{equation}
    \label{eq:marginal_network_factorized}
    \begin{aligned}
    p(\btheta,\bgamma \mid \cA,\bX) 
    &\propto
    p(\btheta)p(\bgamma) \sum_{\bA^\ast} \prod_{i>j} 
    \Pr(A_{ij} \mid A^\ast_{ij},\bX,\bgamma)  \Pr(A^\ast_{ij} \mid \bX,\btheta)
    \\ &= 
     p(\btheta)p(\bgamma) \prod_{i>j} \sum_{k=0,1}   \Pr(A_{ij} \mid A^\ast_{ij}=k,\bX,\bgamma)  \Pr(A^\ast_{ij} = k \mid \bX,\btheta).
\end{aligned}
\end{equation}
Sampling from \eqref{eq:marginal_network_factorized} is trivial in most probabilistic programming languages by encoding the posterior with the log-sum-exp trick. We essentially have a mixture of Bernoulli random variables $A_{ij}$, with mixture probabilities $\Pr(A^\ast_{ij} = k \mid \bX,\btheta)$ and  Bernoulli probabilities $ \Pr(A_{ij} \mid A^\ast_{ij}=k,\bX,\bgamma)$, both for $k=0,1$.
In the case of multiple proxies $B>1$, we will obtain a mixture of categorical random variables (the proxies) with the same mixture probabilities and number of categories equal to $2^B$ in the case that all proxy networks are unweighted.
Note that even if $\cA$ or $\bA^\ast$ are discrete with more than two categories, we can proceed the same way by modifying the mixture probabilities accordingly. 
Finally, generating samples of $\bA^\ast$ is reduced to sampling edges independently from
\begin{equation}
    \label{eq:posterior_network_sampling}
    \begin{aligned}
        p(\bA^\ast \mid \cA,\bX,\btheta,\bgamma) &=
        \frac{p(\cA\mid \bA^\ast,\bX,\bgamma) 
        p(\bA^\ast \mid \bX,\btheta)}{\sum_{\bA^\ast} 
        p(\cA\mid \bA^\ast,\bX,\bgamma) 
        p(\bA^\ast \mid \bX,\btheta)}
        \\ &=
        \prod_{i>j}\frac{
        \Pr(A_{ij} \mid A^\ast_{ij},\bX,\bgamma)  \Pr(A^\ast_{ij}\mid \bX,\btheta)}
        {
        \sum_{k=0,1}
        \Pr(A_{ij} \mid A^\ast_{ij}=k,\bX,\bgamma)  \Pr(A^\ast_{ij} = k \mid \bX,\btheta)
        },
    \end{aligned}
\end{equation}
which can be performed by simple sampling of $N(N-1)/2$ Bernoulli random variables with probabilities 
\begin{equation*}
    \frac{
        \Pr(A_{ij} \mid A^\ast_{ij}=1,\bX,\bgamma)  \Pr(A^\ast_{ij} = 1 \mid \bX,\btheta)}
        {
        \sum_{k=0,1}
        \Pr(A_{ij} \mid A^\ast_{ij}=k,\bX,\bgamma)  \Pr(A^\ast_{ij} = k \mid \bX,\btheta)
        }    
\end{equation*}
In the case of non-causal proxies \eqref{eq:posterior_noncausal}, the network module will not include the proxy networks model $p(\cA \mid \cdot)$. Therefore, sampling $\bA^\ast$ from \eqref{eq:posterior_network_sampling} is equivalent to sampling networks from the  model $p(\bA^\ast \mid \bX,\bV,\btheta)$ given covariates $\bX$, and current values of the latent variables $\bV$ and the parameters $\btheta$.

% Therefore, using multiple proxies can impact $\bA^\ast$ samples in two ways. First, given $\btheta$ and $\bgamma$ values, the probabilities in \eqref{eq:posterior_network_sampling} will change due to the composition of $\xi$. Second, the inference for $\btheta$ might be affected by the types and properties of the proxy networks. That implies that even if single or multiple proxy networks can recover $\btheta$ accurately, the posterior probabilities \eqref{eq:posterior_network_sampling} can be different since the posterior proxy network probabilities $\xi$ are included in \eqref{eq:posterior_network_sampling} as well. 

\noindent
\textbf{Sampling from the outcome module.}
Given multiple $\bA^\ast$ draws from the network module, we can use the modified SCM \eqref{eq:SCM_mod} to sample from the network module. Specifically, under the modified SCM, outcomes depend on $\bA^\ast$ only through summarizing values of $\phi_1,\phi_2,\phi_{3}$. 
 Let $\bphi_{m,i}=\big(\phi_1(\bZ_{-i},\bA^{\ast,m}),\phi_2(\bX_{-i},\bA^{\ast,m}),\phi_{3,i}(\bA^{\ast,m})\big)$, be the unit-level summarizing values based on the $m-$th posterior draw from the network module.
 We average $\bphi_{m,i}$ across the posterior draws $\widehat{\bphi}_i=M^{-1}\sum_{m=1}^{M}\bphi_{m,i}$ and then sample $\boeta$ from the outcome module 
 \begin{equation*}
     p(\boeta \mid \bY,\bZ,\bA^\ast,\bX) \propto 
    p(\boeta) p(\bY \mid \bZ,\bA^\ast,\bX,\boeta)
 \end{equation*}
    using the plug-in values $\widehat{\bphi} = (\widehat{\bphi}_1,\ldots, \widehat{\bphi}_N)$ in the outcome model \citep{Bayarri2009}.
  The outcome model often also depends on $\bA^\ast$ through paths that are not possible to summarize with a simple function. In that case, we will plug in the outcome module the sampled $\bA^{\ast,m}$ with the largest conditional cut-posterior probability,
  \begin{equation*}
      p(\bA^\ast \mid \cA, \bX, \btheta, \bgamma) \propto p(\cA \mid \bA^\ast,\bX,\gamma)p(\bA^\ast \mid \bX,\btheta)
  \end{equation*}
  or just $p(\bA^\ast \mid \bX,\bV, \btheta)$ in the case of non-causal proxies $\cA$.

\noindent
\textbf{Choosing initial values for Block Gibbs algorithm.} 
For $\btheta,\bgamma$ we can either take maximum a-posteriori (MAP) or mean-posterior values as initial values.
Given these values, we sample multiple networks $\bA^\ast$ using \eqref{eq:posterior_network_sampling} and 
choose the network with the highest conditional cut-posterior log-probability as initial value $\bA^\ast_0$.
From the outcome module, we took the posterior means of $\boeta$ as initial values.
As the network module does not use any information from the treatment and outcome models,
this initialization strategy for $\bA^\ast$ can be improved.
We do so by running several LIP updates (Algorithm~\ref{algo:lip}) of $\bA^\ast$ given the initial values of $\btheta,\bgamma,\boeta$.
We take this refined $\bA^\ast$ as the initial value for the Block Gibbs algorithm. 
The initialization strategy can be summarized as

\begin{enumerate}
    \item Compute MAP or mean-posterior values of $(\btheta,\bgamma)$ from the marginalized network-module in the cut-posterior \eqref{eq:marginal_network_factorized}.
    \item Sample multiple $\bA^\ast$ using the values of $(\btheta, \bgamma)$ with probabilities \eqref{eq:posterior_network_sampling}. Save $\bA^\ast_0$ as the network with the highest log-probability \eqref{eq:posterior_network_sampling}. Estimate summary network statistics of the outcome model $\phi_1,\phi_2,\phi_3$ using a plugin estimator of the mean across the network samples.
    \item Sample $\eta$ from the network model using plug-in $\widehat{\phi}_1,\widehat{\phi}_2,\widehat{\phi}_3$ and $\bA^\ast_0$ (if necessary), and take mean posterior (or MAP) values of $\eta$.
    \item Refine $\bA^\ast_0$ by running several LIP updates (Algorithm~\ref{algo:lip}) of $\bA^\ast$ given $(\btheta,\bgamma,\boeta)$ from steps 1 and 3. Take the refined $\bA^\ast$ as initial value for the Block Gibbs algorithm.
\end{enumerate}
In the numeric illustrations on fully-synthetic data (Section \ref{subsec:sim_fully}), 
we sample $(\btheta,\bgamma)$ and use their MAP values in step 1 using stochastic variational inference (SVI)
with Multivariate Normal variational family and ClippedAdam optimizer with learning rate of $5 \times 10^{-4}$
which we ran for $2 \times 10^4$ steps. 
In step 3, we estimated the outcome model (given fixed network), with NUTS with $1.2 \times 10^4$
posterior samples ($3,000$ samples from four chains).
 We found that SVI for the marginalized network models was faster with sufficient accuracy, 
 whereas NUTS better explore the posterior space for the outcome model.
 For the refinement in step 4, we ran $L=2\times 10^4$ LIP updates with $K=5$ edge flip proposal in each step.
Start to end, the initialization was very fast to run for all setups considered in this paper.
%  In the numerical illustrations on semi-synthetic data (Section \ref{subsec:sim_semi}), 
%  the SVI approximation in step 1 was replaced with NUTS (and posterior means)  due to the smaller sample size.

\subsection{Posterior Predictive Distribution of Outcomes}
\label{apdx.sec:ppd}
Recall that the observed data are $\bD = (\bY,\bZ,\cA,\bX)$. We saw in Section~\ref{subsec:bayes_g_formula} that we can estimate causal effects using the conditional expectations $\mathbb{E}\big[Y_i \mid \bZ=\bz, \bX, \bA^\ast_t, \boeta_t \big]$ where $\bA^\ast_t,\boeta_t$ are posterior samples of the latent interference network and the parameters of the outcome model, respectively.

Often, we do not have an analytic expression for the conditional expectation $\mathbb{E}\big[Y_i \mid \bZ=\bz, \bX, \bA^\ast_t, \boeta_t \big]$. In that case, we can approximate it. The posterior predictive distribution (PPD) of the outcome $\bY$ given a new treatment vector $\widetilde{\bZ}$ and the observed data $\bD$ is 
\begin{equation*}
    p(\widetilde{\bY} \mid \widetilde{\bZ},\bD) = \sum_{\bA^\ast}\int_{\boeta} p(\widetilde{\bY} \mid \widetilde{\bZ},\bA^\ast,\bX,\boeta)p(\bA^\ast,\boeta \mid \bD) d\boeta,
\end{equation*}
where $p(\bA^\ast,\boeta \mid \bD) = \int_{\btheta,\bgamma}p(\boeta,\btheta,\bgamma,\bA^\ast \mid \bD)d\btheta d\bgamma$ is the posterior distribution of $\bA^\ast$ and $\boeta$. Given each posterior sample $\bA^\ast_t,\boeta_t$ for $t=1,\ldots,T$ we can draw $S>0$ outcomes vectors $\bY^{(1,t)},\ldots,\bY^{(S,t)}$ with $\bY^{(s,t)}=(Y^{(s,t)}_1,\ldots,Y^{(s,t)}_N)$ from the outcome model $ p(\widetilde{\bY} \mid \widetilde{\bZ},\bA^\ast_t,\bX,\boeta_t)$. These draws can approximate the conditional expectation via
\begin{equation*}
    \mathbb{E}\big[Y_i \mid \bZ=\bz, \bX, \bA^\ast_t, \boeta_t \big] \approx S^{-1}\sum_{s=1}^{S} Y^{(s,t)}_i.
\end{equation*}
Therefore, we can approximate the mean expected outcome of unit $i$ using draws from the PPD 
\begin{equation*}
    \widehat{\mu}_i(\bz \mid \bD) = \mathbb{E}\big[Y_i(\bz) \mid \bD \big] \approx (ST)^{-1} \sum_{t=1}^{T}\sum_{s=1}^{S} Y^{(s,t)}_i.
\end{equation*}
As described in estimation scheme developed in Section~\ref{subsec:bayes_g_formula}, we can use $ \widehat{\mu}_i(\bz \mid \bD)$ to estimate the population-level estimands, e.g., $\tau(\bz,\bz')$.

In the case of Figure~\hyperref[fig:DAGs]{1(d)},
where conditioning on $\bA^\ast$ opens a back-door path between $\bZ$ and $\bY$,
an additional conditioning on the proxies $\cA$ is required for identification.
As stated in Section~\ref{subsec:estimands}, the back-door adjustment is therefore $\bbE[Y_i \mid \bZ=\bz,\cA,\bA^\ast,\bX]$.
In that case, the outcome model should be conditioned on $\cA$ as well, and the draws from its posterior predictive distribution will approximate the modified conditional expectation.

In the numerical illustrations, we used $S=1$ and found that for large  $T$ (number of posterior samples), this approximation was accurate. That is, for each posterior sample of $\boeta$ and $\bA^\ast$, we sampled one outcome vector.

\subsection{Gradient-Based Approximation of Locally Informed Proposals}
\label{apdx.sec:gradient_approx_error}

Let $\log p(\bA^\ast \mid \bD,\bTheta)$ denote the log conditional posterior of the latent network $\bA^\ast$ given the data and parameters, as defined in \eqref{eq:a_star_post}. 
Let $\bA^{\ast'} \in H(\bA^\ast)$ be a network in the Hamming ball of size one, meaning that $\bA^{\ast'}$ differ from $\bA^\ast$ in only one coordinate (i.e., one edge flip).
Denote the difference in the log-conditional posterior by 
\begin{equation*}
    \Delta(\bA^{\ast'},\bA^\ast) = 
    \log p(\bA^{\ast'} \mid \cdot) -
    \log p(\bA^{\ast} \mid \cdot).
\end{equation*}
Assume balancing function $g(a) = \sqrt{a}$ in \eqref{eq:local_balanced_ip}.
The \emph{exact} LIP kernel is \eqref{eq:lip_logpost}
\begin{equation*}
    Q(\bA^{\ast'} \mid \bA^\ast) \propto 
    \exp \left( \frac{1}{2}\Delta(\bA^{\ast'},\bA^\ast) \right),
\end{equation*}
subject to $\bA^{\ast'}$ being in the Hamming ball of size one, as mentioned above.
We approximate the differences using a first-order Taylor expansion. Let $\tilde{\Delta}(\bA^{\ast'}, \bA^\ast)$ be the gradient-based approximation of the differences $\Delta$, as defined in \eqref{eq:grad_approx_diff}. Write the \emph{approximate} LIP kernel is 
\begin{equation*}
    Q^\nabla(\bA^{\ast'} \mid \bA^\ast) \propto 
    \exp \left( \frac{1}{2}\tilde{\Delta}(\bA^{\ast'},\bA^\ast) \right).
\end{equation*}
We analyze the efficiency loss introduced by using $Q^\nabla$ instead of $Q$ by leveraging theoretical results from \citet{zanella2019informed} and \citet{grathwohl_2021}.
\citet{zanella2019informed} established that LIPs with locally balanced functions are asymptotically optimal in terms of Peskun ordering. Specifically, for any function $\psi$, define the asymptotic variance by
\begin{equation*}
    \text{Var}_p(\psi,Q)=\lim_{T\to \infty}\frac{1}{T}\text{Var}
    \left(
    \sum_{t=1}^{T}\psi(\bA^\ast_t)
    \right),
\end{equation*}
where $\bA^\ast_t$ are the samples from proposal $Q$ which is a $p = p(\bA^\ast \mid \cdot)$ stationary Markov transition kernel. The smaller $\text{Var}_p(\psi,Q)$, the more efficient the corresponding MCMC algorithm is in estimating the expectation of $\psi$ under $p$. 
Second, define the spectral of the Markov transition kernel $Q$ by 
\begin{equation*}
    \text{Gap}(Q) =1 - \lambda_2,
\end{equation*}
where $\lambda_2$ is the second largest eigenvalue of $Q$, and always satisfies $\text{Gap}(Q)\ge 0$ \citep{zanella2019informed}. The value of $\text{Gap}(Q)$ is closely related to the convergence properties of $Q$, where values close to $0$ indicate slow convergence while large values correspond to fast convergence \citep{levin2017markov}. 
\citet{zanella2019informed} showed that LIP with locally balanced function minimizes the variance $\text{Var}_p(\psi,Q)$ and maximize the spectral gap $\text{Gap}(Q)$.

\citet[][Theorem~1]{grathwohl_2021} quantified the performance gap when replacing exact differences with gradient approximations. Assuming the gradient of the log-posterior is $L$-Lipschitz, they showed
\begin{enumerate}[(a)]
    \item ${Var}_p(\psi,Q^\nabla) \leq  \frac{{Var}_p(\psi,Q)}{c} + 
    \frac{1-c}{c} {Var}_p(\psi)$
    \item $\text{Gap}(Q^\nabla) \geq c \text{Gap}(Q)$
\end{enumerate}
where $c = \exp(-\frac{1}{2}L)$. That is, using the gradient-based approximation is $c\in (0,1]$ times ``more efficient" than the exact differences \citep{zanella2019informed}.
Thus, $c$ values closer to $1$ imply a minimal efficiency loss due to the approximations.
While \citet{grathwohl_2021} derived $L$ values based on the global spectral norm of the Hessian, we refine this bound for our specific setting. Since our proposal is restricted to single-edge flips, the approximation error is governed by the coordinate-wise curvature rather than the global curvature.

Using the Lagrange remainder form of the Taylor expansion, the difference can be written as
\begin{equation*}
    \Delta(\bA^{\ast'},\bA^\ast) = 
    \log p(\bA^{\ast'} \mid \cdot) -
    \log p(\bA^{\ast} \mid \cdot) = \nabla \log p(\bA^\ast \mid \cdot) ^\top \delta + \frac{1}{2}\delta^\top \mathcal{H}\delta,
\end{equation*}
where $\delta$ is a one-hot vector representing the flipped edge in the upper triangle of $\bA^\ast$, and $\mathcal{H}$ is the Hessian matrix at some point between $\bA^\ast$ and $\bA^{\ast'}$. 
The approximation error is strictly the quadratic term
\begin{equation}
\label{eq.apdx:diff_approx_error}
\text{Error}= 
    \left\lvert 
    \Delta(\bA^{\ast'},\bA^\ast) -
    \nabla \log p(\bA^\ast  \mid \cdot) ^\top \delta 
    \right\rvert 
    = 
    \left\lvert 
    \frac{1}{2}\delta^\top \mathcal{H}\delta
    \right\rvert .
\end{equation}
Since $\delta$ has only one non-zero element, we can bound the error using the diagonal of the Hessian. Let $e=(i,j)$ be the flipped edge, such that $\delta_e=1$ and the rest are zero.
Then $\delta^\top \mathcal{H}\delta=\mathcal{H}_{e,e}$ is the $e,e = ij,ij$ diagonal element of the Hessian, where $\mathcal{H}_{e,e} = \mathcal{H}_{ij,ij}=\frac{\partial^2\log p(\bA^\ast\mid \cdot)}{(\partial A^\ast_{ij})^2}$.
Therefore, we can bound the error in \eqref{eq.apdx:diff_approx_error} by
\begin{equation*}
    \text{Error} \leq \frac{1}{2} \max_e 
    \lvert \mathcal{H}_{e,e} \rvert.
\end{equation*}
Thus, we can replace the global Lipschitz constant $L$ with the maximum absolute diagonal element of the Hessian
\begin{equation}
    \label{eq:c_lip}
    c = \exp\left(-\frac{1}{2} \max_e \lvert \mathcal{H}_{e,e} \rvert\right).
\end{equation}

\paragraph{Derivation for the Simplified Model.}
We verify this constant for the simplified model \eqref{eq:simple_model} introduced in Section~\ref{sec:identifiability}, extended here to include heterogeneous prior network and proxy measurements models:
\begin{equation*}
    \begin{aligned}
        A^\ast_{ij} &\sim \text{Ber}(\theta_{ij}), \\
        A_{ij} \mid A^\ast_{ij} = k &\sim \text{Ber}(\gamma_{k,ij}), \; k=0,1,\\
        Z_i &\sim \text{Ber}(p_z), \\
        Y_i &= \eta_1 Z_i + \eta_2 \sum_{j \neq i} A^\ast_{ij} Z_j + \varepsilon_i, \\
    \end{aligned}
\end{equation*}
where $\varepsilon_i$ are iid $N(0,\sigma^2)$. 
The log-posterior decomposes into the outcome likelihood, proxy likelihood, and network prior
\begin{equation*}
    \log p(\bA^\ast \mid \cdot) = \log p(\bY \mid \bZ,\bA^\ast, \boeta) +
    \log p(\bA \mid \bA^\ast , \bgamma) +
    \log p(\bA^\ast \mid \btheta).
\end{equation*}
Let $R_i = Y_i - \eta_1Z_i$. We have
\begin{equation*}
    \begin{aligned}
        \log p(\bA^\ast \mid \btheta)
        &= \sum_{i < j}A^\ast_{ij}\log(\theta_{ij}) +
        (1-A^\ast_{ij}) \log(1-\theta_{ij}) 
        \\
        \log p(\bA \mid \bA^\ast , \bgamma)
        &= \sum_{i < j}A^\ast_{ij}
        \left[
            A_{ij}\log(\gamma_{1,ij}) +
        (1-A_{ij}) \log(1-\gamma_{1,ij})
        \right] 
        \\ &+
        \sum_{i < j}
        (1-A^\ast_{ij})
        \left[
            A_{ij}\log(\gamma_{0,ij}) +
        (1-A_{ij}) \log(1-\gamma_{0,ij})
        \right] 
        \\
        \log p(\bY \mid \bZ,\bA^\ast, \boeta) &= 
        C - \frac{1}{2\sigma^2}\sum_{i=1}^{N} 
        \left( 
            R_i - \eta_2 \sum_{j\neq i}A^\ast_{ij}Z_j
        \right)^2,
    \end{aligned}
\end{equation*}
where $C$ is some constant. 
As both $\log p(\bA^\ast \mid \btheta)$ and $\log p(\bA \mid \bA^\ast , \bgamma)$ are linear in $A^\ast_{ij}$, the curvature will be determined only by the outcome model.
Recall that $A^\ast_{ij} =A^\ast_{ji}$ as the network is assumed to be symmetric.
Thus, the gradient with respect to edge $e=ij$ is 
\begin{equation*}
    \nabla_{ij}\log p(\bA^\ast \mid \cdot) =
    \frac{\partial \log p(\bA^\ast \mid \cdot)}{\partial A^\ast_{ij}} = 
    \widetilde{C} + \frac{\eta_2}{\sigma^2}
    \left[
    \left(R_i-\eta_2 \sum_{k\neq i}A_{ki}Z_k \right)Z_j
    +
    \left(R_j-\eta_2 \sum_{k\neq j}A_{kj}Z_k \right)Z_i
    \right],
\end{equation*}
for constant $\widetilde{C}$. 
Therefore, the $e,e = ij,ij$ element on the diagonal of the Hessian is
\begin{equation*}
    \mathcal{H}_{e,e} = 
    \frac{\partial^2\log p(\bA^\ast\mid \cdot)}{(\partial A^\ast_{ij})^2}
    =
    -\frac{\eta_2^2}{\sigma^2}\left(Z_j^2 + Z_i^2\right)
    = 
    -\frac{\eta_2^2}{\sigma^2}\left(Z_j + Z_i\right),
\end{equation*}
as $Z_i$ are binary.
Clearly, $\lvert \mathcal{H}_{e,e} \rvert$ is maximized when $Z_i=Z_j=1$. Substituting this into \eqref{eq:c_lip} gives the efficiency factor
\begin{equation*}
    c = \exp\left( -\frac{\eta_2^2}{\sigma^2}\right).
\end{equation*}
This result indicates that the efficiency of the gradient-based approximation depends on the signal-to-noise ratio of the interference effect, but importantly, does not degrade with the network size $N$.
For moderate interference-to-noise ratio with $\eta_2^2 = 0.5\sigma^2$, we have $c \approx 0.6$, suggesting the approximation does not severely degrade the asymptotic efficiency of the Markov chain. 

\section{Numerical Illustrations}
\label{apdx.sec:sim}

For the continuous kernels in the Block Gibbs algorithm we used NumPyro \citep{phan2019composable}. Code is available at \url{https://github.com/barwein/BayesProxyNets}.

The prior specification is as follows in all models and setups.
$\theta \sim N(0,3^2), \; \gamma \sim N(0,3^2), \eta\sim N(0,3^2),\; \sigma \sim \text{HalfNormal}(0,3^2)$. For the continuous relaxation using the CONCRETE distribution \citep{maddison2017concrete}, we assumed a temperature equal to $0.5$, implying a continuous density in $[0,1]$ with a U-shape \citep[see Figure~3(b) in][]{maddison2017concrete}.

Figure~\ref{fig:gradient_approx} shows how the first-order Taylor approximations using gradients correctly approximate, up to a scale, the manual differences 
for the Locally Informed Proposals for $\bA^\ast$ updates in the fully-synthetic data setting.

\begin{figure}[!hbtp]
    \centering
    \includegraphics[width=0.35\linewidth]{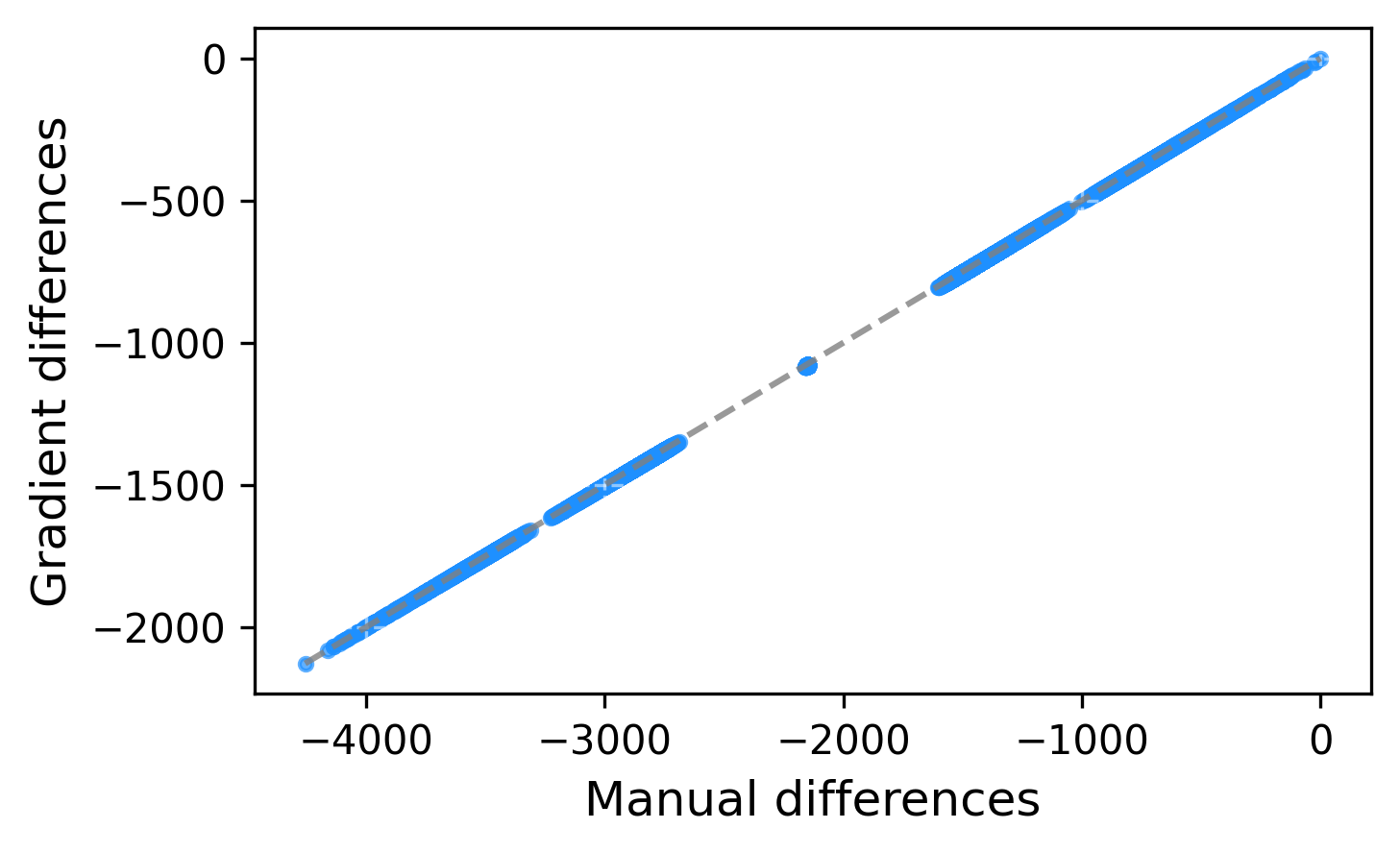}
    \caption{Gradient approximations $\widetilde{\Delta}\left(ij,t \mid \cdot \right)$ versus the differences $\Delta\left(ij,t \mid \cdot \right)$. Values are on the log-softmax scale. The approximation is accurate up to a scale.}
    \label{fig:gradient_approx}
\end{figure}

\subsection{Fully-Synthetic Data}
The full specification of the data generating process for $N=500$ units is
\begin{equation*}
    \begin{aligned}
        X_{1,i} &\sim N(0,1) \\
        X_{2,i} &\sim Ber(0.1) \\
        \Pr(A^\ast_{ij}=1 \mid \bX,\btheta) &= s\big(-2 +  \widetilde{X}_{2,ij} \big) \\
           \Pr(A_{ij}=1 \mid A^\ast_{ij}, \bX, \bgamma) &=
            s\big(A^\ast_{ij}\gamma_0 + (1-A^\ast_{ij})(\gamma_1 + \gamma_2\widetilde{X}_{1,ij} + \gamma_3\widetilde{X}_{2,ij})\big) \\
        %  \Pr(A_{2,ij}=1 \mid A_{1,ij}, A^\ast_{ij}, \bX, \bgamma) &= s\big(A^\ast_{ij}(logit(0.8) + 1.5 A_{1,ij}) + (1-A^\ast_{ij})(logit(0.2) + 1.5A_{1,ij})\big) \\
        \Pr(Z_i=1) &= p_z \frac{d^\ast_i}{N-1} \\
        % \bpsi &\sim N\Big(0, (\boldsymbol{D} - 0.5\bA^\ast)^{-1}\Big) \\
        \varepsilon_i &\sim N(0,1) \\
        \mu_i &= -1 + 3Z_i + 2\phi_1 - 0.5X_{1,i}\\
       Y_i &= \mu_i + \varepsilon_i,
    \end{aligned}
\end{equation*}
where we set $\gamma_0 = \text{logit}(0.95) - \gamma_2/2$, $\gamma_1=\text{logit}(0.05)+\gamma_2/2$,
 and took $\gamma_2 \in \{2,2.5,3,3.5,4\}$.
In each iteration, we resampled all the data. Given each $\bA^\ast$ sample, a proxy network is generated for each $\gamma_2$ value.

Figure~\ref{fig.apdx:A_star_deg} shows the degree distribution of $\bA^\ast$ in one iteration. The right heavy-tails of degrees are units with an extremely high number of connections in the network, representing highly connected units.

Figure~\ref{fig.apdx:rel_error} shows the mean ($\pm$ SD) of the Relative Error $\left \vert \frac{\hat{\tau}-\tau}{\tau}\right \vert$ of the population-level estimand $\tau$ across $300$ iterations.
The figure is similar to the results of the unit-level MAPE values presented in Figure~\ref{fig:sim_rrmse}, showing that all BG methods had minimal error, while using a misspecified model (``BG (misspec)") slightly increased the relative error. The continuous relaxation method had a small mean relative error, although with a very high standard deviation across the simulation iterations, highlighting its instability.

Figure~\ref{fig.apdx:rrmse} shows the mean ($\pm$ SD) 
of Relative RMSE values across $300$ iterations.
Relative RMSE over $T$ posterior samples is defined as
$\text{rRMSE} = \sqrt{T^{-1}\sum_{t=1}^{T}\left(\frac{\hat{\tau}_t - \tau}{\tau}\right)^2}$,
where $\hat{\tau}_t$ is point estimate of the estimand at posterior sample $t$,
and $\tau$ is the true estimand. Lower values of rRMSE are better.

Figure~\ref{fig.apdx:coverage} presents the frequentist coverage rate in $300$ iterations. Coverage is the proportion of iterations where the $95\%$ credible interval included the true estimand $\tau$. Figure~\ref{fig.apdx:coverage} show that the BG algorithm (with one or two proxies) achieves a nominal coverage rate.

Figure~\ref{fig.apdx:traceplots} shows traceplots of $1.2 \times 10^4$ MCMC samples of the BG sampler in one iteration.
The traceplots corresponds to $\gamma_1$, $\eta$ (multiplying $\phi_1$),
and the dynamic and static estimands. The chains appear to mix well,
with no apparent trends or drifts.

Figure~\ref{fig.apdx:mwg_scaling} shows the scaling and mixing of the BG across $10$ iterations 
in one setup ($\gamma_2=3$ and one proxy network).
The figure shows, for varying sample sizes $N$, 
the average ($\pm$ SD) of wall-clock time (in seconds) 
to obtain $1.2 \times 10^4$ MCMC samples, 
and the relative minimal effective sample size (ESS) per second compared
to $N=100$.
Even though the running time increase with $N$, 
for $N \geq 500$ the relative min ESS/sec stabilizes, indicating that 
increasing the sample size increase both computation time and sampling efficiency 
simultaneously.

\begin{figure}[!hbtp]
    \centering
    \includegraphics[width=0.35\linewidth]{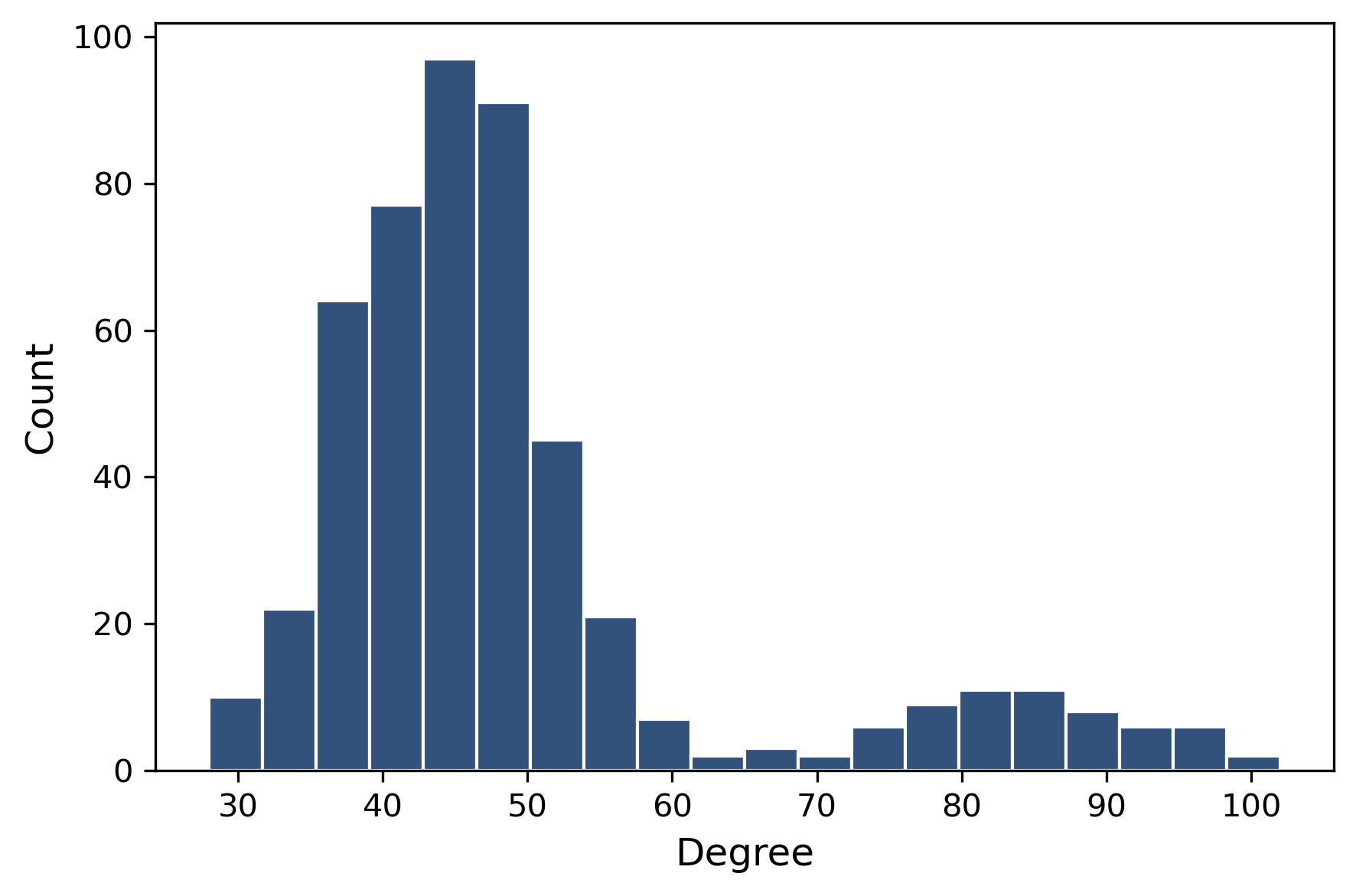}
    \caption{Degree distribution of $\bA^\ast$ in one iteration.}
    \label{fig.apdx:A_star_deg}
\end{figure}

\begin{figure}[!hbtp]
    \centering
    \includegraphics[width=0.6\linewidth]{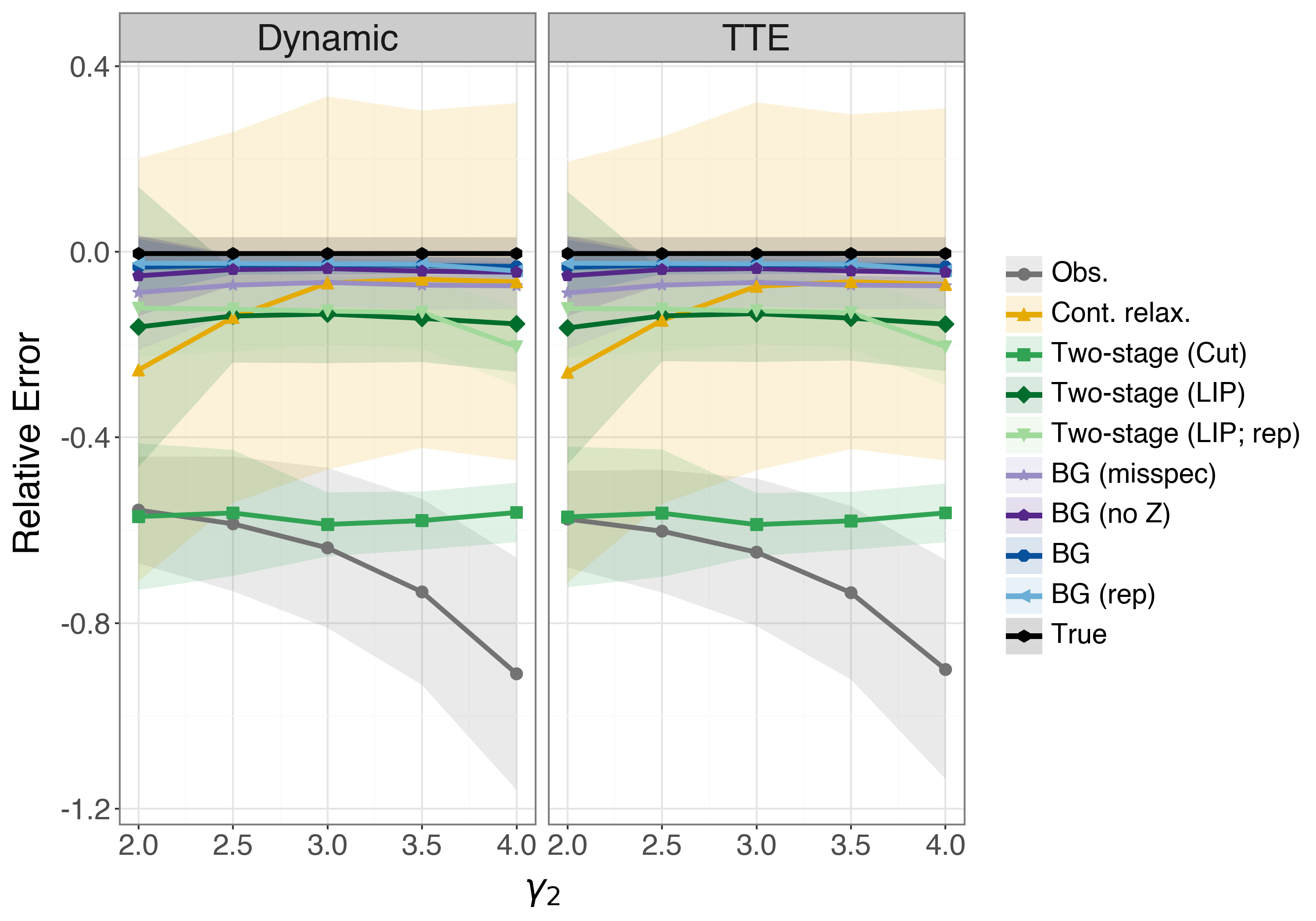}
    \caption{Mean ($\pm \; SD$) of Relative Error $\left \vert \frac{\hat{\tau}-\tau}{\tau}\right \vert$ 
    values across $300$ iterations.}
    \label{fig.apdx:rel_error}
\end{figure}

\begin{figure}[!hbtp]
    \centering
    \includegraphics[width=0.6\linewidth]{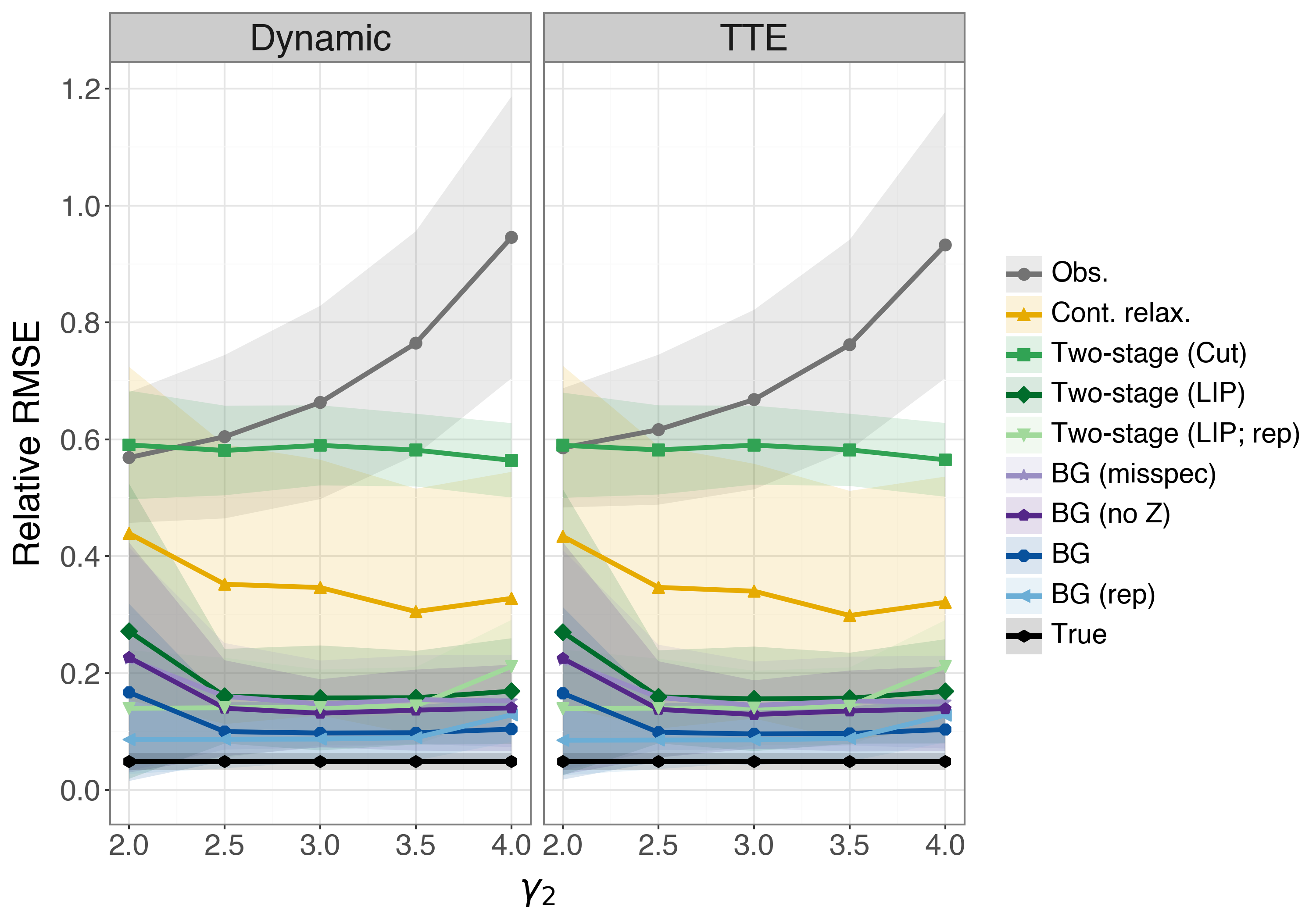}
    \caption{Mean ($\pm \; SD$) of Relative RMSE 
    values across $300$ iterations.}
    \label{fig.apdx:rrmse}
\end{figure}

\begin{figure}[!hbtp]
   \centering
   \includegraphics[width=0.6\linewidth]{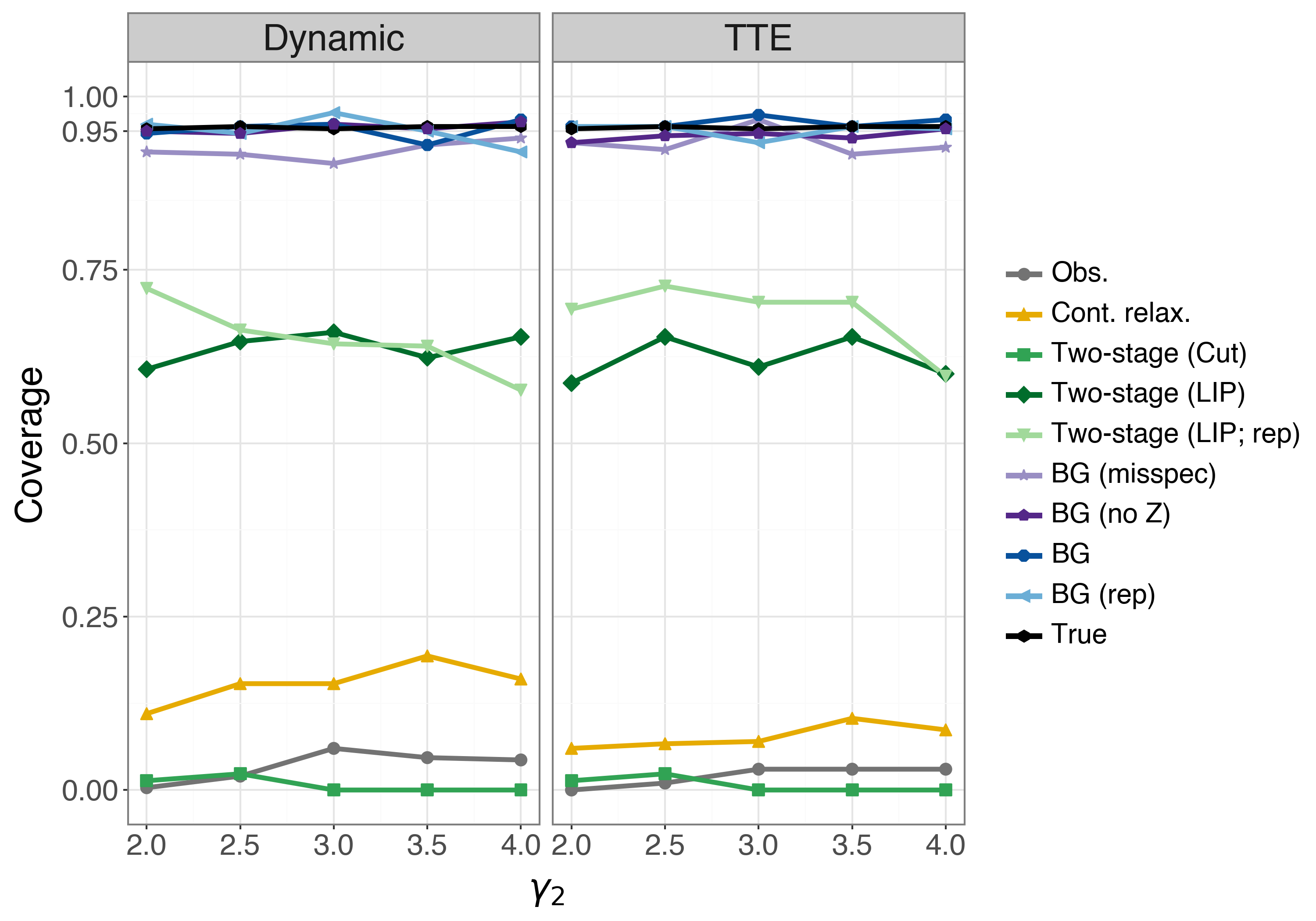}
   \caption{Coverage rate over $300$ iterations.}
    \label{fig.apdx:coverage}
\end{figure}

\begin{figure}[!hbtp]
    \centering
    \includegraphics[width=0.8\linewidth]{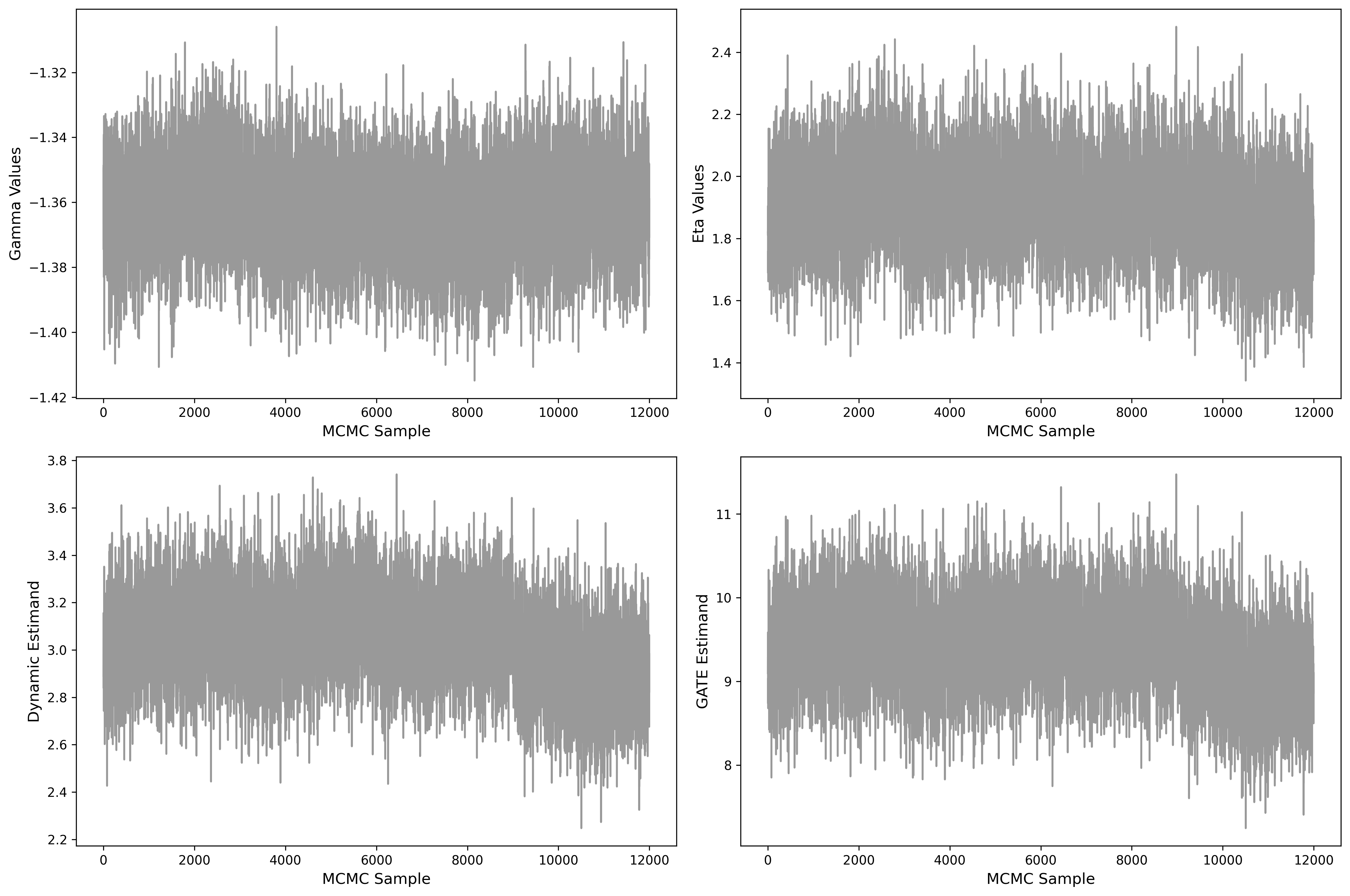}
    \caption{Traceplots of $1.2 \times 10^4$ MCMC samples of the BG sampler
    in one iteration. The traceplots corresponds to $\gamma_1$, $\eta$ (multiplying $\phi_1$),
    and the dynamic and static estimands.}
    \label{fig.apdx:traceplots}
\end{figure}

\begin{figure}[!hbtp]
    \centering
    \includegraphics[width=0.5\linewidth]{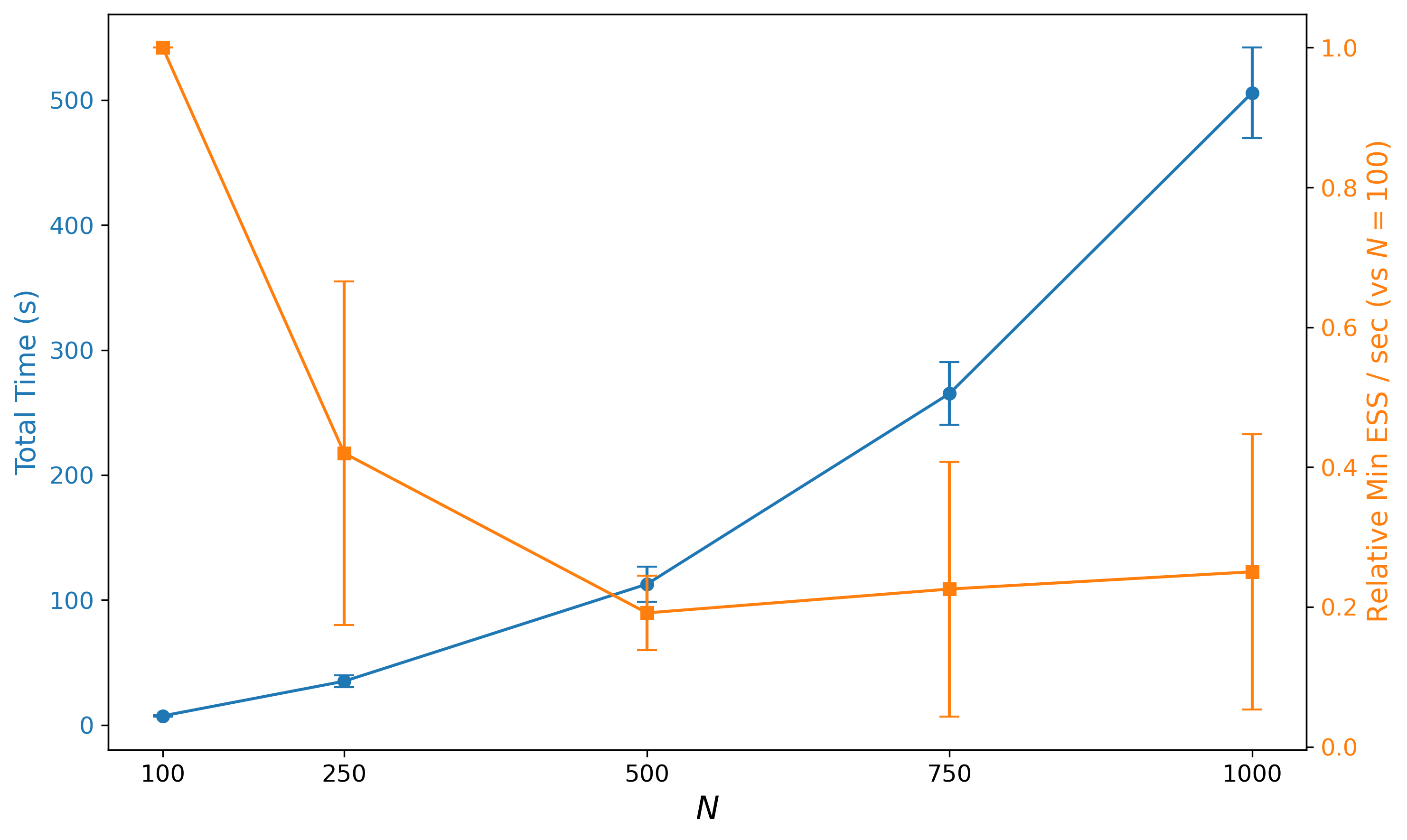}
    \caption{Scaling and mixing of the BG across $10$ iterations in one setup.
    For varying sample sizes $N$, 
    the figure shows the average ($\pm$ SD) of wall-clock time (in seconds) 
    to obtain $1.2 \times 10^4$ MCMC samples (left y-axis, blue), 
    and the relative minimal effective sample size (ESS) per second compared
    to $N=100$ (right y-axis, orange). Even though the running time increase with $N$, 
    for $N \geq 500$ the relative min ESS/sec stabilizes, indicating that 
    increasing the sample size increase both computation time and sampling efficiency 
    simultaneously.
    Clock-wall times were obtained on a PC equipped with a
     12-core Apple Silicon (ARM64) processor and 48 GB of RAM, running macOS (Darwin Kernel 24.6.0).}
    \label{fig.apdx:mwg_scaling}
\end{figure}

\subsection{Semi-Synthetic Data}
% The data are available at the following link: \url{https://www.icpsr.umich.edu/web/ICPSR/studies/37070}.

Figure~\ref{fig.apdx:multilayer_error} summarizes the mean ($95\%$ posterior interval) of the
error $\hat{\tau}-\tau$ in estimating the TTE estimand in each of the setups. 
The results show that the Block Gibbs sampler had errors close to zero across most iterations and network layers.
The Two-stage method had smaller errors than all aggregated networks, other than in the 
Work layer where the ``OR'' aggregation had slightly smaller errors.
Table~\ref{tab.apdx:coverage} show the empirical coverage of the TTE estimand 
across $300$ iterations by layer and estimation method.
The results shows that the BG method achieved nominal coverage across all layers.
However, both the Two-stage and the aggregated networks had poor coverage in most layers,
with ``AND'' network being the worst across all layers other than Facebook.

Figure~\ref{fig.apdx:multilayer_nets} displays the four networks used in the analysis. 
Network data is available at \url{https://manliodedomenico.com/data.php}.

\begin{table}[!htbp]
\centering
\begin{tabular}{lcccc}
\toprule
 \emph{Method} & \emph{Facebook} & \emph{Leisure} & \emph{Lunch} & \emph{Work} \\
\midrule
True & 0.953 & 0.933 & 0.957 & 0.960 \\
BG & 0.923 & 0.933 & 0.947 & 0.917 \\
Two-stage & 0.577 & 0.913 & 0.853 & 0.854 \\
OR & 0.450 & 0.900 & 0.607 & 0.964 \\
AND & 0.717 & 0.545 & 0.000 & 0.003 \\
\bottomrule
\end{tabular}
\caption{Empricial coverage across $300$ iterations of the TTE estimand by layer 
and estimation method.} 
\label{tab.apdx:coverage}
\end{table}

\begin{figure}[!hbtp]
    \centering
    \includegraphics[width=0.7\linewidth]{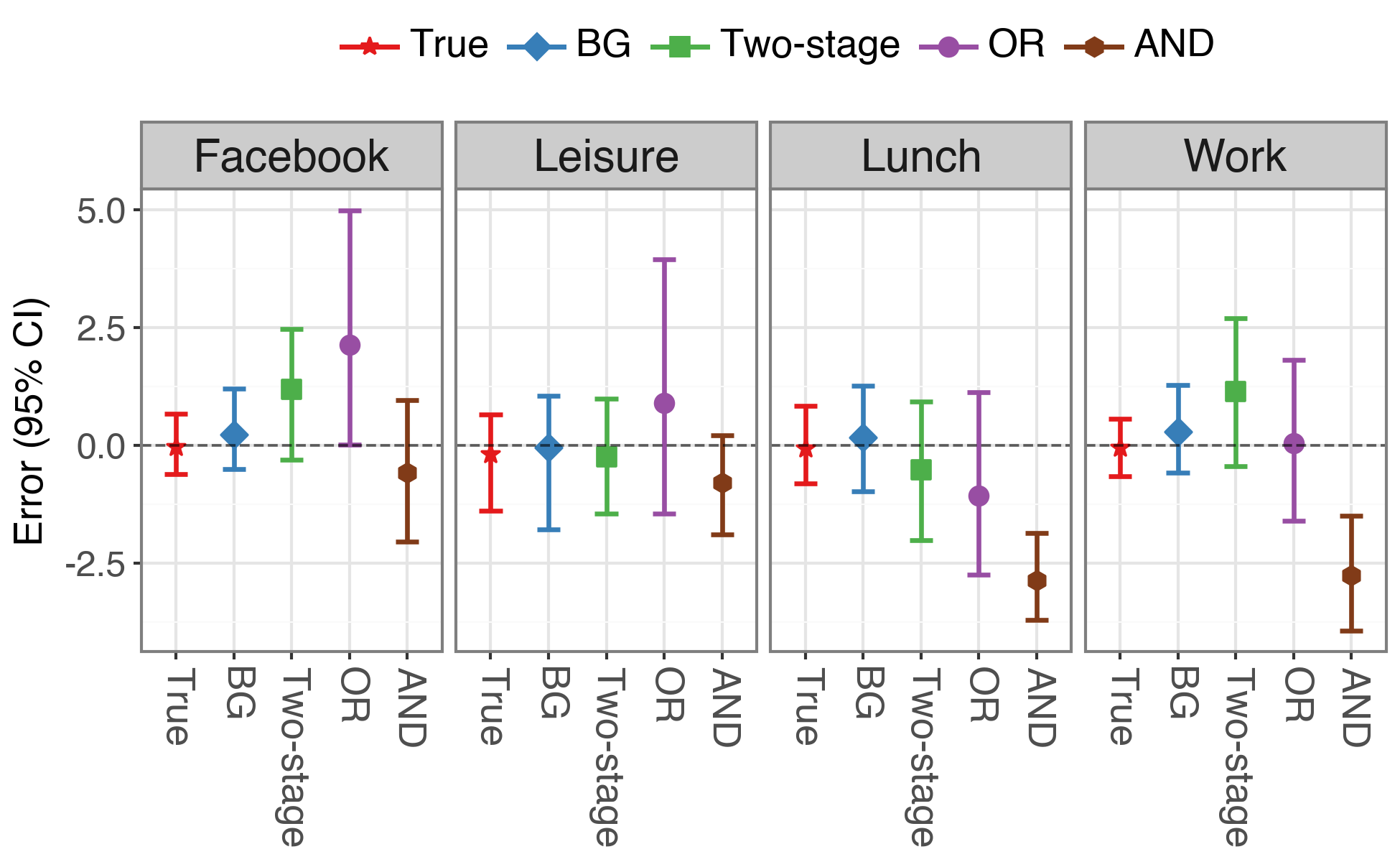}
    \caption{Mean ($95\%$ posterior interval) of error $\hat{\tau}-\tau$ values 
    across $300$ iterations in each setup.}
    \label{fig.apdx:multilayer_error}
\end{figure}

\begin{figure}[!hbtp]
    \centering
    \includegraphics[width=0.65\linewidth]{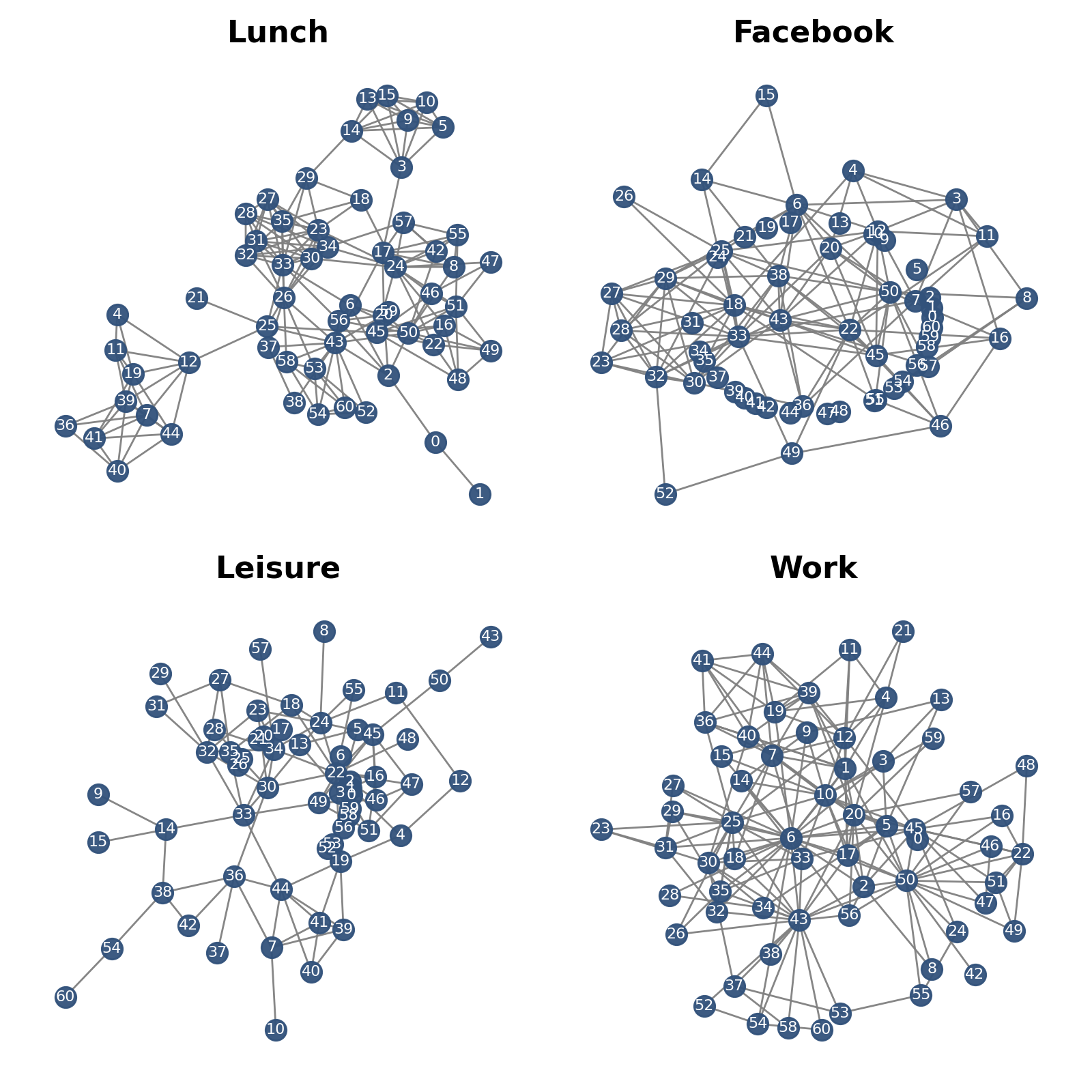}
    \caption{Multilayer networks.}
    \label{fig.apdx:multilayer_nets}
\end{figure}

\newpage

% \vskip 0.2in
\bibliography{all_biblo}

\end{document}